%% file: draft.tex
\documentclass{amsart}
\usepackage{fullpage}

\usepackage[varg]{txfonts}

\usepackage[subrefformat=parens]{subcaption}
\usepackage{tikz}
\usetikzlibrary{calc}

\newtheorem{Theorem}{Theorem}
\newtheorem{Lemma}[Theorem]{Lemma}
\newtheorem{claim}{\it Claim}
\newtheorem{Corollary}[Theorem]{Corollary}
\newtheorem{Example}[Theorem]{Example}
\newtheorem{Proposition}[Theorem]{Proposition}

\usepackage[lined,ruled,commentsnumbered]{algorithm2e}

\usepackage{enumerate}
\usepackage{hyperref}
\usepackage{cite}

\title{Linear-semiorders and their incomparability graphs}
\author{Asahi~Takaoka}
\address{
  College of Information and Systems, 
  Muroran Institute of Technology, 
  Mizumoto 27-1, Muroran, 
  Hokkaido, 050--8585, Japan 
}
\email{takaoka@mmm.muroran-it.ac.jp}
\date{\today}



\begin{document}
\begin{abstract}
A linear-interval order is the intersection 
of a linear order and an interval order. 
For this class of orders, several structural results have been known. 
This paper introduces a new subclass of linear-interval orders. 
We call a partial order a \emph{linear-semiorder} 
if it is the intersection of a linear order and a semiorder. 
We show a characterization 
and a polynomial-time recognition algorithm 
for linear-semiorders. 
We also prove that being a linear-semiorder is a comparability invariant, 
showing that incomparability graphs of linear-semiorders 
can be recognized in polynomial time. 
\end{abstract}

\maketitle

\input{main.tex}

\subsubsection*{Acknowledgments}
We are grateful to the anonymous referees 
of the preliminary version of this paper 
for their time and valuable suggestions. 
A part of this work was done 
while the author was in Kanagawa University. 


\input{draft.bbl}
\end{document}

%% file: main.tex
\section{Introduction}
A graph is an \emph{intersection graph} 
if there is a set of objects such that 
each vertex corresponds to an object 
and two vertices are adjacent if and only if 
the corresponding objects have a nonempty intersection. 
The set of objects is called a \emph{representation} of the graph. 
Intersection graphs of geometric objects have been widely investigated 
because of their interesting structures and applications~\cite{BLS99,Golumbic04,Spinrad03}. 

Well-known examples of intersection graphs 
are interval graphs and permutation graphs. 
An \emph{interval graph} is the intersection graph 
of intervals on the real line. 
Let $L_1$ and $L_2$ be two horizontal lines 
in the $xy$-plane with $L_1$ above $L_2$. 
A \emph{permutation graph} is the intersection graph 
of line segments joining a point on $L_1$ and a point on $L_2$. 

A common generalization of the two graph classes 
is trapezoid graphs~\cite{CK87-CN,DGP88-DAM}. 
An interval on $L_1$ and an interval on $L_2$ define 
a trapezoid between $L_1$ and $L_2$. 
A \emph{trapezoid graph} is the intersection graph of 
such trapezoids. 
The structure of trapezoid graphs is well investigated, and 
several recognition algorithms are presented~\cite{GT04,MC11-DAM,Spinrad03}. 

There is some correspondence between partial orders 
and intersection graphs of geometric objects 
between the two lines~\cite{GRU83-DM},~\cite[Theorem 1.11]{GT04}. 
A partial order $P$ on a set $V$ is a \emph{trapezoid order} 
if for each element $v \in V$, 
there is a trapezoid $T(v)$ between $L_1$ and $L_2$ 
such that for any two elements $u, v \in V$, 
we have that $u \prec v$ in $P$ if and only if 
$T(u)$ lies completely to the left of $T(v)$. 
The set of trapezoids $\{T(v) \colon\ v \in V\}$ is called 
a \emph{trapezoid representation} of $P$. 

By restricting trapezoids in the representation, 
many order classes have been introduced~\cite{BLR10-Order,BMR98-Order,Ryan98-Order}. 
An \emph{up-triangle order}~\cite{BLR10-Order} is a partial order 
representable by triangles spanned by a point on $L_1$ and 
an interval on $L_2$. 
An up-triangle order is also known as 
a \emph{PI order}~\cite{BLS99,COS08-ENDM,CK87-CN}, 
where \emph{PI} stands for \emph{point-interval}, 
and as a \emph{linear-interval order}~\cite{Mertzios15-SIAMDM} 
because it is the intersection of a linear order and an interval order. 
We use the term linear-interval orders to denote such orders. 
Several structural results have been shown 
for this order class~\cite{COS08-ENDM,CK87-CN,Takaoka18-DM}, 
including polynomial-time recognition algorithms~\cite{Mertzios15-SIAMDM,Takaoka20-DAM,Takaoka20a-DAM}. 
As noted in~\cite{Mertzios15-SIAMDM}, 
this is one of the first results for recognizing orders that are 
the intersection of orders from different classes. 

This paper deals with up-triangle orders 
representable by triangles spanned by a point on $L_1$ and 
a \emph{unit-length} interval on $L_2$ (Fig.~\ref{figs:chevron}). 
Such an order is the intersection of a linear order and a semiorder; 
hence we call it a \emph{linear-semiorder}. 
In Section~\ref{sec:charac}, we show 
a characterization of linear-semiorders 
in terms of linear extensions. 
We then propose a polynomial-time recognition algorithm 
for linear-semiorders in Section~\ref{sec:recog}. 
In Section~\ref{sec:compa}, we prove 
that being a linear-semiorder is a comparability invariant, 
showing that incomparability graphs of linear-semiorders 
can also be recognized in polynomial time. 
In Section~\ref{sec:byproducts}, 
we show two byproducts of the characterization. 
We finally discuss our results and 
further research in Section~\ref{sec:conclusion}. 

\begin{figure}
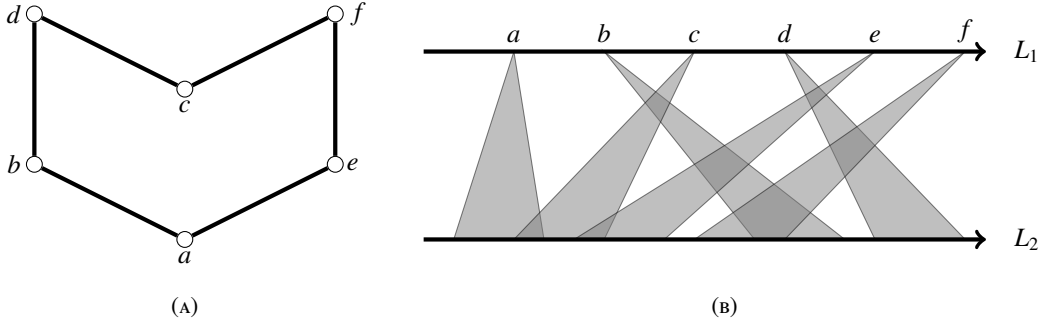

\centering
\subcaptionbox{\label{figs:chevron:order}}{\input{figs/chevron/order.tex}}
\subcaptionbox{\label{figs:chevron:representation}}{\input{figs/chevron/representation}}
\caption{A partial order (the dual of chevron) with a triangle representation. }
\label{figs:chevron}
\end{figure}

\section{Preliminaries}
A (strict) \emph{partially ordered set} 
is a pair $(V, P)$, 
where $V$ is a set and $P$ is a binary relation on $V$ 
that is irreflexive, transitive, and therefore asymmetric. 
The set $V$ is called a \emph{ground set} and 
the relation $P$ is called a (strict) \emph{partial order} on $V$. 
In this paper, we will deal only with partial orders on finite sets. 

We denote partial orders by $\prec$ instead of $P$, 
that is, we write $u \prec v$ in $P$ if $(u, v) \in P$. 
Two elements $u, v \in V$ are \emph{comparable} in $P$ 
if $u \prec v$ or $u \succ v$; 
otherwise $u$ and $v$ are \emph{incomparable}, 
denoted by $u \parallel v$. 
A partial order $P$ on a set $V$ is a \emph{linear order} 
if any two distinct elements of $V$ are comparable in $P$. 

A partial order $P$ on a set $V$ is an \emph{interval order} 
if for each element $v \in V$, 
there is a (closed) interval $I(v)$ on the real line 
such that for any two elements $u, v \in V$, 
we have that $u \prec v$ in $P$ if and only if 
$I(u)$ lies completely to the left of $I(v)$. 
Here, the interval $I(u) = [l(u), r(u)]$ 
\emph{lies completely to the left of} $I(v)  = [l(v), r(v)]$, 
and we write $I(u) \ll I(v)$, 
if $r(u) < l(v)$. 
The set of intervals $\{I(v) \colon\ v \in V\}$ is called 
an \emph{interval representation} of $P$. 

An interval representation is 
\emph{unit} if every interval has unit length, and 
it is \emph{proper} if no interval properly contains another. 
An interval order is a \emph{semiorder} 
if it has a unit interval representation. 
It is known that a partial order is a semiorder if and only if 
it has a proper interval representation~\cite{BW99-DM}. 

Let $P_1$ and $P_2$ be two partial orders on the same ground set $V$. 
The \emph{intersection} of $P_1$ and $P_2$ is 
the partial order $P = P_1 \cap P_2$. 
Equivalently, the intersection of $P_1$ and $P_2$ is 
the partial order $P$ on $V$ such that 
$u \prec v$ in $P$ if and only if $u \prec v$ in both $P_1$ and $P_2$. 
We call an order a \emph{linear-semiorder} 
if it is the intersection of a linear order and a semiorder. 

In addition to partially ordered sets, 
this paper deals with graphs. 
All graphs in this paper are finite 
with no loops or multiple edges. 
Unless stated otherwise, graphs are assumed to be undirected, 
but we also deal with graphs with directed edges. 
We write $uv$ for the \emph{undirected edge} 
joining two vertices $u$ and $v$ and 
write $(u, v)$ for the \emph{directed edge} 
from $u$ to $v$. 
For a graph $G = (V, E)$, we sometimes 
write $V(G)$ for the vertex set $V$ and 
write $E(G)$ for the edge set $E$. 
\par
Let $P$ be a partial order on a set $V$. 
The \emph{comparability graph} of $P$ is the graph $G$ 
such that $uv \in E(G)$ if and only if $u$ and $v$ are comparable in $P$. 
The \emph{incomparability graph} of $P$ is the graph $G$ 
such that $uv \in E(G)$ if and only if $u \parallel v$ in $P$. 
Note that 
the incomparability graph of $P$ is the complement of 
the comparability graph of $P$, where 
the \emph{complement} of a graph $G$ 
is the graph $\overline{G}$ 
such that $V(\overline{G}) = V(G)$ and 
$uv \in E(\overline{G})$ if and only if $uv \notin E(G)$ 
for any two vertices $u, v \in V(\overline{G})$. 
\par
Let $G$ be a graph, and 
let $E \subseteq E(G)$ be a set of (undirected) edges of $G$. 
We call a set $F$ of directed edges an \emph{orientation} of $E$ 
if $F$ is obtained from $E$ by 
orienting each edge of $E$, that is, 
replacing each edge $uv \in E$ with either $(u, v)$ or $(v, u)$. 
An orientation of $G$ is an orientation of $E(G)$. 
An orientation $F$ of $G$ is \emph{transitive} 
if $(u, v) \in F$ and $(v, w) \in F$ then $(u, w) \in F$ 
for any three vertices $u, v, w \in V(G)$. 
Note that a partial order $P$ on a set $V$ can be viewed as 
a transitive orientation $F$ of the comparability graph of $P$, 
in which $(u, v) \in F$ if and only if $u \prec v$ in $P$. 
We also note that a graph is a comparability graph if and only if 
it has a transitive orientation. 
\par
A linear-time algorithm is known for computing 
a transitive orientation of a comparability graph~\cite{MS99-DM}. 
However, the algorithm might produce an orientation that is not transitive 
if the given graph is not a comparability graph. 
Thus, to recognize comparability graphs, 
we must verify the transitivity after computing the orientation. 
\par
An orientation $F$ of a graph $G$ is \emph{quasi-transitive} 
if $(u, v) \in F$ and $(v, w) \in F$ then $uw \in E(G)$, 
that is, either $(u, w) \in F$ or $(w, u) \in F$. 
In other words, an orientation $F$ of $G$ is quasi-transitive 
if for any three vertices $u, v, w \in V$ 
with $uv, vw \in E(G)$ and $uw \notin E(G)$, 
either $(u, v), (w, v) \in F$ or $(v, u), (v, w) \in F$. 
We can see that an orientation is transitive if and only if 
it is quasi-transitive and acyclic. 
\par
The linear-time algorithm~\cite{MS99-DM} produces a linear extension 
of an orientation $F$ of a given graph $G$ 
such that $F$ is transitive if $G$ is a comparability graph. 
Obviously, $F$ is acyclic. 
We can test whether $F$ is quasi-transitive by checking 
either $(u, v), (w, v) \in F$ or $(v, u), (v, w) \in F$ 
for any pair of a non-edge $uw \notin E(G)$ and a vertex $v \in V(G)$ 
with $uv, vw \in E(G)$. 
Thus we have the following. 
\begin{Proposition}\label{prop:obtain transitive orientation}
Comparability graphs can be recognized in $O(n\bar{m})$ time, 
where $n$ and $\bar{m}$ are the number of vertices and non-edges 
of the given graph. 
A transitive orientation of a graph $G$ can be obtained 
in the same time bound if $G$ is a comparability graph. 
\end{Proposition}
\par
Let $G$ be a graph. 
A sequence of four vertices $(v_0, v_1, v_2, v_3)$ is 
an \emph{induced cycle of length $4$} if 
$v_0v_1, v_1v_2, v_2v_3, v_3v_0 \in E(G)$ and 
$v_0v_2, v_1v_3 \notin E(G)$, 
denoted by $C_4$. 
A graph consisting of four vertices $v_0, v_1, v_2, v_3$ is 
a \emph{claw} if $v_0v_1, v_0v_2, v_0v_3 \in E(G)$ 
and $v_1v_2, v_2v_3, v_3v_1 \notin E(G)$, 
denoted by $K_{1, 3}$. 
A sequence of three vertices $(v_0, v_1, v_2)$ is 
an \emph{induced path of length $2$} if 
$v_0v_1, v_1v_2 \in E(G)$ and $v_0v_2 \notin E(G)$, 
denoted by $P_3$.

\section{Characterization}\label{sec:charac}
Let $P$ be a partial order on a set $V$. 
A linear order $L$ on $V$ is a \emph{linear extension} of 
$P$ if $u \prec v$ in $L$ whenever $u \prec v$ in $P$. 
Hence, a linear extension $L$ of $P$ has all the relations of $P$ with 
additional relations making $L$ linear. 
We define some properties of linear extensions. 

The order $\mathbf{2+2}$ of $P$ is a partial order 
consisting of four elements $x$, $y$, $z$, $w$ of $V$ such that 
$x \prec y$ and $z \prec w$ whereas $x \parallel w$ and $z \parallel y$ in $P$. 
Note that $x \parallel z$ and $y \parallel w$ in $P$; 
otherwise, we would have $x \prec w$ or $z \prec y$ in $P$. 
We say that a linear extension $L$ of $P$ 
fulfills the \emph{$\mathbf{2+2}$ rule} if 
$y \prec z$ or $w \prec x$ in $L$ 
for each induced suborder $\mathbf{2+2}$ of $P$. 
If $y \prec z$ in $L$ then $x \prec y \prec z \prec w$ in $L$; 
if $w \prec x$ in $L$ then $z \prec w \prec x \prec y$ in $L$. 
Equivalently, a linear extension $L$ of $P$ 
fulfills the $\mathbf{2+2}$ rule if 
there are no four elements $x, y, z, w$ of $V$ 
such that $x \prec y$, $z \prec w$, $x \parallel w$, and 
$z \parallel y$ in $P$ whereas 
$x \prec w$ and $z \prec y$ in $L$ 
(Fig.~\ref{figs:forbidden}\subref{figs:forbidden:2+2}). 
We call such an induced suborder 
a \emph{forbidden configuration for $\mathbf{2+2}$}. 

The order $\mathbf{3+1}$ of $P$ is a partial order 
consisting of four elements $x$, $y$, $z$, $w$ of $V$ such that 
$x \prec y \prec z$ whereas $x \parallel w$ and $w \parallel z$ in $P$. 
Note that $y \parallel w$ in $P$; 
otherwise, we would have $x \prec w$ or $w \prec z$ in $P$. 
We say that a linear extension $L$ of $P$ 
fulfills the \emph{$\mathbf{3+1}$ rule} if 
$w \prec x$ or $z \prec w$ in $L$ 
for each induced suborder $\mathbf{3+1}$ of $P$. 
If $w \prec x$ in $L$ then $w \prec x \prec y \prec z$ in $L$; 
if $z \prec w$ in $L$ then $x \prec y \prec z \prec w$ in $L$. 
Equivalently, a linear extension $L$ of $P$ 
fulfills the $\mathbf{3+1}$ rule if 
there are no four elements $x, y, z, w$ of $V$ 
such that $x \prec y \prec z$, $x \parallel w$, and 
$w \parallel z$ in $P$ whereas 
$x \prec w \prec z$ in $L$ 
(Fig.~\ref{figs:forbidden}\subref{figs:forbidden:3+1}). 
We call such an induced suborder 
a \emph{forbidden configuration for $\mathbf{3+1}$}. 

\begin{figure}[t]
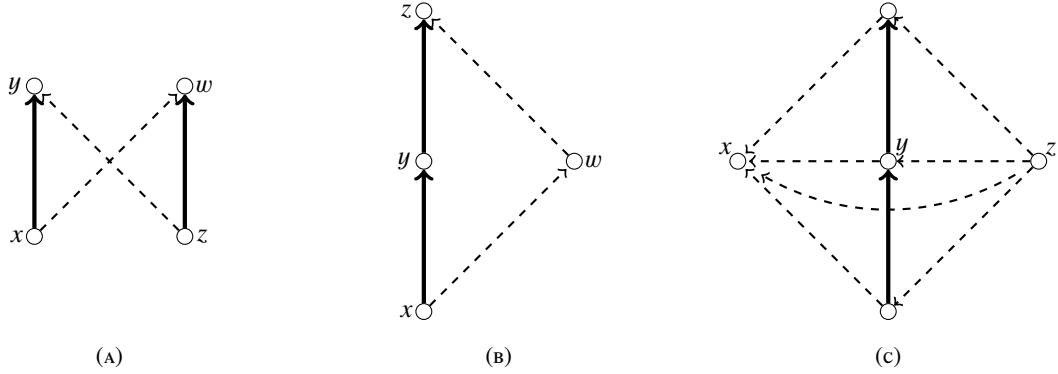

  \centering
  \subcaptionbox{\label{figs:forbidden:2+2}}{\input{figs/forbidden/2+2.tex}}
  \subcaptionbox{\label{figs:forbidden:3+1}}{\input{figs/forbidden/3+1.tex}}
  \subcaptionbox{\label{figs:intransitivity}}{\input{figs/intransitivity.tex}}
  \caption{
    (a) A forbidden configuration for $\mathbf{2+2}$. 
    (b) A forbidden configuration for $\mathbf{3+1}$. 
    (c) A partial order with $(x, y), (y, z) \in P \cup R_1 \cup R_2$ 
    but $(x, z) \notin P \cup R_1 \cup R_2$. 
    An arrow $u \to v$ denotes that $u \prec v$ in $P$; 
    a dashed arrow $u \dasharrow v$ denotes that $u \prec v$ in $L$ 
    but $u \parallel v$ in $P$. 
    } 
  \label{figs:forbidden}
\end{figure}

Note that the orders $\mathbf{2+2}$ and $\mathbf{3+1}$ 
characterize interval orders and semiorders as follows. 
A partial order is an interval order if and only if 
it does not contain $\mathbf{2+2}$ 
as an induced suborder~\cite{Fishburn70-JMP}. 
A partial order is a semiorder if and only if 
it does not contain $\mathbf{2+2}$ or $\mathbf{3+1}$ 
as an induced suborder~\cite{SS58-JSL}. 

Our previous work~\cite{Takaoka18-DM} shows that 
a partial order is a linear-interval order if and only if 
it has a linear extension fulfilling the $\mathbf{2+2}$ rule. 
In this paper, we show the following characterization for linear-semiorders. 
\begin{Theorem}\label{theorem:characterization}
A partial order is a linear-semiorder if and only if 
it has a linear extension fulfilling the $\mathbf{2+2}$ and $\mathbf{3+1}$ rules. 
\end{Theorem}
\begin{proof}
The necessity and sufficiency follow immediately from 
Lemmas~\ref{lemma:rightarrow} and~\ref{lemma:leftarrow}, respectively.
\end{proof}

\begin{Lemma}\label{lemma:rightarrow}
If a partial order $P$ on a set $V$ has a linear order $L$ and 
a semiorder $S$ with $L \cap S = P$, then 
$L$ is a linear extension of $P$ 
fulfilling the $\mathbf{2+2}$ and $\mathbf{3+1}$ rules. 
\end{Lemma}
\begin{proof}
The linear order $L$ can be viewed as a linear extension of $P$. 
It is shown in~\cite{CK87-CN,Takaoka18-DM} that 
$L$ fulfills the $\mathbf{2+2}$ rule. 
We now show that $L$ fulfills the $\mathbf{3+1}$ rule. 
Suppose that $L$ has a forbidden configuration for $\mathbf{3+1}$ 
consisting of four elements $x, y, z, w$ of $V$ such that 
$x \prec y \prec z$, $x \parallel w$, and 
$z \parallel w$ in $P$ whereas 
$x \prec w \prec z$ in $L$. 
Let $\{I(v) \colon\ v \in V\}$ be a proper interval representation of 
the semiorder $S$, 
and let $I(v) = [l(v), r(v)]$ for each element $v$ of $V$. 
Since $x \prec y \prec z$ in $P$, we have $r(x) < l(y) \leq r(y) < l(z)$. 
Since $x \prec w$ in $L$ and $x \parallel w$ in $P$, 
we have $x \not\prec w$ in $S$. Thus $l(w) \leq r(x)$. 
Similarly, since $w \prec z$ in $L$ and $w \parallel z$ in $P$, 
we have $w \not\prec z$ in $S$. Thus $l(z) \leq r(w)$. 
Therefore, $I(w) \supseteq I(y)$, 
contradicting that no interval properly contains another. 
\end{proof}

\begin{Lemma}\label{lemma:leftarrow}
If a partial order $P$ on a set $V$ has a linear extension $L$ 
fulfilling the $\mathbf{2+2}$ and $\mathbf{3+1}$ rules, 
then there is a semiorder $S$ with $L \cap S = P$. 
The semiorder $S$ and its proper interval representation 
can be obtained in $O(n^3)$ time, 
where $n$ is the number of elements of $V$. 
\end{Lemma}
\begin{proof}
In this proof, we denote partial orders 
as a set of ordered pairs of elements. For example, 
we write $(u, v) \in P$ if $u \prec v$ in $P$, 
and write $(u, v) \in L - P$ if $u \parallel v$ in $P$ and $u \prec v$ in $L$. 
\par
We define the following two binary relations on $V$. 
\begin{align*}
R_1 & = \{(x, y) \in V^2 \colon\ \text{there is an element $z \in V$ with $(x, z) \in P$ and $(y, z) \in L - P$}\}; 
\\
R_2 & = \{(x, y) \in V^2 \colon\ \text{there is an element $z \in V$ with $(z, x) \in L - P$ and $(z, y) \in P$}\}. 
\end{align*}
Let $Q = P \cup R_1 \cup R_2$. 
Note that the relation $Q$ is not necessarily transitive. 
See a partial order in Fig.~\ref{figs:forbidden}\subref{figs:intransitivity}. 
We have $(x, y) \in R_2$ and $(y, z) \in R_1$, 
but $(x, z) \notin P \cup R_1 \cup R_2$. 

\begin{claim}\label{claim:Q}
There is a linear order $L_Q$ on $V$ such that 
$x \prec y$ in $L_Q$ whenever $(x, y) \in Q$. 
\end{claim}
\begin{proof}[Proof of Claim~\ref{claim:Q}]
We say that 
a sequence of distinct elements $(v_0, v_1, \ldots, v_{k-1})$ of $V$ 
is a \emph{cycle} of $Q$ if $(v_i, v_{i+1}) \in Q$ 
for every index $i$ with $0 \leq i < k$ (indices are modulo $k$). 
The \emph{length} of the cycle is the number $k$. 
To prove the claim, 
we show by a case analysis that $Q$ has no cycles. 
\par
Suppose that $Q$ has a cycle of length 2. 
We first suppose $(v_0, v_1) \in R_1$. 
Then there is an element $u \in V$ 
with $(v_0, u) \in P$ and $(v_1, u) \in L - P$. 
If $(v_1, v_0) \in P$ then 
$(v_0, u) \in P$ implies $(v_1, u) \in P$, a contradiction. 
If $(v_1, v_0) \in R_1$ then 
there is an element $w \in V$ with 
$(v_1, w) \in P$ and $(v_0, w) \in L - S$. 
Thus $L$ has a forbidden configuration for $\mathbf{2+2}$, a contradiction. 
If $(v_1, v_0) \in R_2$ then 
there is an elements $w \in V$ with 
$(w, v_0) \in P$ and $(w, v_1) \in L - S$. 
Thus $L$ has a forbidden configuration for $\mathbf{3+1}$, a contradiction. 
Therefore, $(v_0, v_1) \notin R_1$. 
A similar argument would show $(v_0, v_1) \notin R_2$. 
We finally suppose $(v_0, v_1) \in P$. 
Trivially, $(v_1, v_0) \notin P$. 
We also have $(v_1, v_0) \notin R_1$ and $(v_1, v_0) \notin R_2$ 
by symmetric arguments. Thus $(v_0, v_1) \notin P$. 
Therefore, $Q$ has no cycles of length 2. 
\par
Suppose that $Q$ has a cycle of length greater than 2. 
Let $(v_0, v_1, \ldots, v_{k-1})$ be such a cycle with minimal length, 
that is, there are no relations $(v_i, v_j) \in Q$ with $j \neq i + 1$. 
Suppose further that 
there is an index $i$ with $(v_i, v_{i+1}) \in P$. 
If $(v_{i+1}, v_{i+2}) \in P$ then 
$(v_i, v_{i+2}) \in P$, a contradiction. 
Suppose $(v_{i+1}, v_{i+2}) \in R_1$. 
Then there is an element $u \in V$ 
with $(v_{i+1}, u) \in P$ and $(v_{i+2}, u) \in L - P$. 
We have $(v_i, u) \in P$ from $(v_i, v_{i+1}), (v_{i+1}, u) \in P$. 
Now $(v_i, u) \in P$ and $(v_{i+2}, u) \in L - P$ 
imply $(v_i, v_{i+2}) \in R_1$, a contradiction. 
Suppose $(v_{i+1}, v_{i+2}) \in R_2$. 
Then there is an element $u \in V$ 
with $(u, v_{i+1}) \in L - P$ and $(u, v_{i+2}) \in P$. 
We have that $v_i \parallel v_{i+2}$ in $P$; otherwise, 
$(v_i, v_{i+2}) \in Q$ or $(v_{i+2}, v_i) \in Q$, 
a contradiction. 
Thus the elements $u, v_{i+2}, v_i, v_{i+1}$ induce $\mathbf{2+2}$. 
Since $L$ fulfills the $\mathbf{2+2}$ rule and $(u, v_{i+1}) \in L - P$, 
we have $(u, v_i) \in L - P$. 
Now $(u, v_i) \in L - P$ and $(u, v_{i+2}) \in P$ imply 
$(v_i, v_{i+2}) \in R_2$, a contradiction. 
Therefore, there is no index $i$ with $(v_i, v_{i+1}) \in P$. 
\par
Suppose there is an index $i$ with 
$(v_i, v_{i+1}), (v_{i+1}, v_{i+2}) \in R_1$. 
Then there are two elements $u_0, u_1 \in V$ with 
$(v_i, u_0), (v_{i+1}, u_1) \in P$ and 
$(v_{i+1}, u_0), (v_{i+2}, u_1) \in L - S$. 
If $(v_i, u_1) \in P$ then $(v_{i+2}, u_1) \in L - P$ implies 
$(v_i, v_{i+2}) \in R_1$, a contradiction. 
If $(v_i, u_1) \in L - P$ then the elements $v_i, u_0, v_{i+1}, u_1$ 
induce a forbidden configuration for $\mathbf{2+2}$, a contradiction. 
Thus $(u_1, v_i) \in L$, and 
we have $(v_{i+2}, u_0) \in L$ from $(v_{i+2}, u_1), (u_1, v_i), (v_i, u_0) \in L$. 
If $(v_{i+2}, u_0) \in P$ then 
the elements $v_{i+2}, u_0, v_{i+1}, u_1$ 
induce a forbidden configuration for $\mathbf{2+2}$, a contradiction. 
If $(v_{i+2}, u_0) \in L - P$ then $(v_i, u_0) \in P$ implies 
$(v_i, v_{i+2}) \in R_1$, a contradiction. 
Therefore, there is no index $i$ with 
$(v_i, v_{i+1}), (v_{i+1}, v_{i+2}) \in R_1$. 
A similar argument would show that there is no index $i$ with 
$(v_i, v_{i+1}), (v_{i+1}, v_{i+2}) \in R_2$. 
\par
Now, there is an index $i$ with $(v_i, v_{i+1}) \in R_1$; 
hence $(v_{i+1}, v_{i+2}) \in R_2$ and $(v_{i+2}, v_{i+3}) \in R_1$. 
Note that the vertices $v_i, v_{i+1}, v_{i+2}$ are distinct by assumption, 
but $v_{i+3}$ could be identical to $v_1$. 
Then there are three elements $u_0, u_1, u_2 \in V$ with 
$(v_i, u_0), (u_1, v_{i+2}), (v_{i+2}, u_2) \in P$ and 
$(v_{i+1}, u_0), (u_1, v_{i+1}), (v_{i+3}, u_2) \in L - P$. 
We have $u_0 \neq u_2$; otherwise, 
the elements $u_1, v_{i+2}, u_2=u_0, v_{i+1}$ 
induce a forbidden configuration for $\mathbf{3+1}$, a contradiction. 
If $(v_{i+2}, u_0) \in P$ then the elements $u_1, v_{i+2}, u_0, v_{i+1}$ 
induce a forbidden configuration for $\mathbf{3+1}$, a contradiction. 
If $(v_{i+2}, u_0) \in L - P$ then $(v_i, u_0) \in P$ implies 
$(v_i, v_{i+2}) \in R_1$, a contradiction. 
Thus $(u_0, v_{i+2}) \in L$, and 
we have $(v_{i+1}, u_2) \in L$ from $(v_{i+1}, u_0), (u_0, v_{i+2}), (v_{i+2}, u_2) \in L$. 
If $(v_{i+1}, u_2) \in P$ then $(v_{i+3}, u_2) \in L - P$ implies 
$(v_{i+1}, v_{i+3}) \in R_1$, a contradiction. 
If $(v_{i+1}, u_2) \in L - P$ then the elements $u_1, v_{i+2}, u_2, v_{i+1}$ 
induce a forbidden configuration for $\mathbf{3+1}$, a contradiction. 
\par
Therefore, the relation $Q$ has no cycles, and thus the claim holds. 
\end{proof}

Assume that the elements $v_1, v_2, \ldots, v_n$ of $V$ are 
indexed so that $i < j$ if $v_i \prec v_j$ in $L_Q$. 
\begin{claim}\label{claim:L_Q}
The elements $v_1, v_2, \ldots, v_n$ 
satisfies the following properties: 
\begin{enumerate}[(a)]
\item 
There are no three indices $i, j, k$ 
with $i < j < k$ such that 
$(v_i, v_k) \in L - P$ and $(v_j, v_k) \in P$. 
\item 
There are no three indices $i, j, k$ 
with $i < j < k$ such that 
$(v_i, v_k) \in L - P$ and $(v_i, v_j) \in P$. 
\item 
There are no four indices $i, j, k, h$ 
with $i < j < k < h$ such that 
$(v_i, v_h) \in L - P$ and $(v_j, v_k) \in P$. 
\end{enumerate}
\end{claim}
\begin{proof}[Proof of Claim~\ref{claim:L_Q}]
Suppose there are three indices $i, j, k$ 
with $i < j < k$ such that 
$(v_i, v_k) \in L - P$ and $(v_j, v_k) \in P$. 
Then $(v_j, v_i) \in R_1$, 
which implies $v_j \prec v_i$ in $L_Q$, 
a contradiction. 
Thus the property (a) holds. 
Suppose there are three indices $i, j, k$ 
with $i < j < k$ such that 
$(v_i, v_k) \in L - P$ and $(v_i, v_j) \in P$. 
Then $(v_k, v_j) \in R_2$, 
which implies $v_k \prec v_j$ in $L_Q$, 
a contradiction. 
Thus the property (b) holds. 
\par
Suppose there are four indices $i, j, k, h$ 
with $i < j < k < h$ such that 
$(v_i, v_h) \in L - P$ and $(v_j, v_k) \in P$. 
If $(v_i, v_k) \in P$ then the indices $i, k, h$ satisfy 
$i < k < h$ with 
$(v_i, v_h) \in L - P$ and $(v_i, v_k) \in P$, 
contradicting the property (b). 
If $(v_i, v_k) \in L - P$ then the indices $i, j, k$ satisfy 
$i < j < k$ with $(v_i, v_k) \in L - P$ and $(v_j, v_k) \in P$, 
contradicting the property (a). 
Thus $(v_k, v_i) \in L$. 
A similar argument would show $(v_h, v_j) \in L$. 
Then 
$(v_i, v_h), (v_h, v_j), (v_j, v_k), (v_k, v_i) \in L$, 
a contradiction. 
Thus the property (c) holds. 
\end{proof}

We define a function $f: \{1, 2, \ldots, n\} \to \mathbb{N}$ 
recursively as follows. 
For the base case, we set $f(1) = 0$; 
for each index $j$ with $1 < j \leq n$, we set 
\[
f(j) = \begin{cases}
    \max\{i \colon\ \text{$f(j-1) < i$ and $(v_i, v_j) \in P$}\} & 
        \text{if there is an index $i$ with $f(j-1) < i$ and $(v_i, v_j) \in P$, } \\
    f(j-1) &  \text{otherwise. }
  \end{cases}
\]

\begin{claim}\label{claim:f}
The function $f$ satisfies the following properties: 
\begin{enumerate}[(a)]
\item $0 \leq f(i) < i$. 
\item If $i < j$ then $f(i) \leq f(j)$. 
\item If $(v_i, v_j) \in P$ then $i \leq f(j)$. 
\item If $(v_i, v_j) \in L - P$ then $f(j) < i$. 
\end{enumerate}
\end{claim}
\begin{proof}[Proof of Claim~\ref{claim:f}]
Trivially, $f$ satisfies the properties (a)--(c). 
Suppose there are two indices $i$ and $j$ with 
$(v_i, v_j) \in L - P$ and $i \leq f(j)$. 
If $(v_{f(j)}, v_j) \in P$ then $i \neq f(j)$, 
that is, $i < f(j)$. 
Hence the indices $i, f(j), j$ satisfy $i < f(j) < j$ with 
$(v_i, v_j) \in L - P$ and $(v_{f(j)}, v_j) \in P$, 
contradicting the property (a) of Claim~\ref{claim:L_Q}. 
If $(v_{f(j)}, v_j) \notin P$ then 
from the definition of $f$, there is an index $k$ 
with $k < j$ such that $f(k) = f(j)$ and $(v_{f(k)}, v_k) \in P$. 
Then $i \leq f(j) = f(k) < k < j$ with 
$(v_i, v_j) \in L - P$ and $(v_{f(k)}, v_k) \in P$, 
contradicting the property (b) or (c) of Claim~\ref{claim:L_Q}. 
Therefore, $f$ satisfies the property (d). 
\end{proof}

Let $\mathcal{I} = \{[f(i) + \frac{i-1}{n}, i] \colon\ 1 \leq i \leq n\}$ 
be the set of intervals on the real line. 
The property (a) of Claim~\ref{claim:f} ensures that 
intervals of $\mathcal{I}$ are well-defined. 
The property (b) of Claim~\ref{claim:f} implies that 
no interval of $\mathcal{I}$ properly contains another. 
Let $S$ be the semiorder 
obtained from the proper interval representation $\mathcal{I}$. 
The properties (c) and (d) of Claim~\ref{claim:f} imply, respectively, that 
$v_i \prec v_j$ in $S$ if $(v_i, v_j) \in P$ and 
$v_i \not\prec v_j$ in $S$ if $(v_i, v_j) \in L - P$. 
Therefore, $L \cap S = P$. 

The relation $Q$ is obtained in $O(n^3)$ time from $P$ and $L$. 
The function $f$ and the representation $\mathcal{I}$ of $S$ are 
obtained in $O(n^2)$ time from $L_Q$. 
Thus $S$ and $\mathcal{I}$ can be obtained from $P$ and $L$ in $O(n^3)$ time. 
\end{proof}

In Fig.~\ref{figs:S}, we illustrate the construction of the semiorder 
in the proof of Lemma~\ref{lemma:leftarrow}. 
Consider a partial order $P$ in Fig.~\ref{figs:S}\subref{figs:S:chevron}. 
Figure~\ref{figs:chevron} shows that $P$ is a linear-semiorder. 
Let $L$ be the linear order with 
$a \prec b \prec c \prec d \prec e \prec f$. 
We can observe that $L$ fulfills 
the $\mathbf{2+2}$ and $\mathbf{3+1}$ rules. 
We have 
$R_1 \cup R_2 = \{(a, c), (c, b), (c, e), (e, b), (e, d), (f, d)\}$. 
Thus there are two linear orders that 
contain all ordered pairs of $Q = P \cup R_1 \cup R_2$, that is, 
a linear order with 
$a \prec c \prec e \prec b \prec f \prec d$ and 
a linear order with 
$a \prec c \prec e \prec f \prec b \prec d$. 
We choose the latter as $L_Q$. 
We obtain the six intervals 
$I(a) = [0, 1]$, 
$I(c) = [\frac{1}{6}, 2]$, 
$I(e) = [1+\frac{2}{6}, 3]$, 
$I(f) = [3+\frac{3}{6}, 4]$, 
$I(b) = [3+\frac{4}{6}, 5]$, and 
$I(d) = [5+\frac{5}{6}, 6]$ 
(Fig.~\ref{figs:S}\subref{figs:S:representation}). 
We can verify that no interval properly contains another. 
The semiorder $S$ defined by the intervals is 
in Fig.~\ref{figs:S}\subref{figs:S:order}, 
and we can observe $L \cap S = P$. 

\begin{figure}[t]
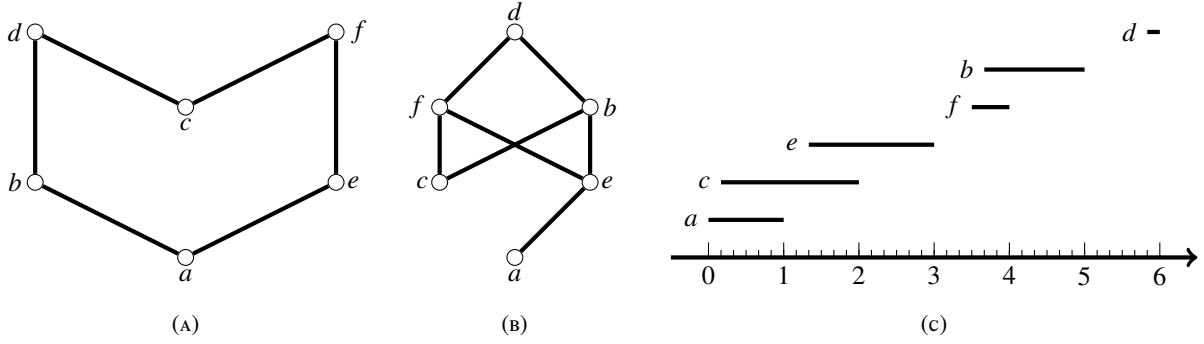

  \centering
  \subcaptionbox{\label{figs:S:chevron}}{\input{figs/chevron/order.tex}}
  \subcaptionbox{\label{figs:S:order}}{\input{figs/S/order}}
  \subcaptionbox{\label{figs:S:representation}}{\input{figs/S/representation}}
  \caption{
    (a) A linear-semiorder $P$. 
    (b) The semiorder $S$. 
    (c) The proper interval representation of $S$. 
    }
  \label{figs:S}
\end{figure}

\section{Recognition algorithm}\label{sec:recog}
In this section, we show a recognition algorithm for linear-semiorders. 
The following is our main result. 

\begin{Theorem}\label{theorem:recognition:main}
A linear-semiorder can be recognized in $O(n \bar{m})$ time. 
A linear extension fulfilling the $\mathbf{2+2}$ and $\mathbf{3+1}$ rules 
can be obtained in $O(m \bar{m}^2)$ time. 
Here, $n$, $m$, and $\bar{m}$ are the number of 
elements, comparable pairs, and incomparable pairs of the given order, respectively. 
\end{Theorem}

Let $P$ be a partial order on a set $V$, and 
let $G$ be the comparability graph of $P$. 
In this section, we consider 
$V$ as the vertex set of $G$ and 
$P$ as the transitive orientation of $G$. 
That is, we consider $P$ as the set of directed edges 
such that for any three vertices $u, v, w$ of $G$, 
if $(u, v), (v, w) \in P$ then $(u, w) \in P$. 
Note $n = |V|$ and $m = |E(G)|$. 

Let $\overline{G}$ be the complement of $G$. 
Recall that $\overline{G}$ is the incomparability graph of $P$, 
and hence $\bar{m} = |E(\overline{G})|$. 
Let $F$ be an orientation of $\overline{G}$, and 
let $P + F$ denote the union of directed edges of $P$ and $F$. 
Note that $P + F$ is the orientation of the complete graph on $V$. 
Thus if $P + F$ is acyclic 
then $P + F$ can be viewed as the linear order on $V$.

A \emph{forbidden configuration for $\mathbf{2+2}$} 
is a sequence of vertices $(x, y, z, w)$ 
with $(x, y), (z, w) \in P$ and $(z, y), (x, w) \in F$ 
(Fig.~\ref{figs:forbidden}\subref{figs:forbidden:2+2}). 
A \emph{forbidden configuration for $\mathbf{3+1}$} 
is a sequence of vertices $(x, y, z, w)$ 
with $(x, y), (y, z) \in P$ and $(w, z), (x, w) \in F$ 
(Fig.~\ref{figs:forbidden}\subref{figs:forbidden:3+1}). 
The following can be obtained from Theorem~\ref{theorem:characterization}. 
\begin{Proposition}\label{prop:consequence of characterization}
A partial order $P$ on a set $V$ is a linear-semiorder if and only if 
there is an orientation $F$ of the incomparability graph of $P$ 
such that $P + F$ is acyclic and 
contains no forbidden configurations for $\mathbf{2+2}$ or $\mathbf{3+1}$. 
\end{Proposition}

Let $C_4$ denote an induced cycle of length 4, and 
let $K_{1, 3}$ denote a claw. 
%
We denote by $E_o$ the set of edges 
of $C_4$ or $K_{1, 3}$ of $\overline{G}$. 
Note $E_o \subseteq E(\overline{G})$ by definition. 
The proposed algorithm orients the edges of $E_o$, 
instead of the edges of $\overline{G}$. 
In the following, we say that 
even if $F$ is an orientation of not all edges of $\overline{G}$, 
a sequence of vertices $(x, y, z, w)$ is 
a forbidden configuration for $\mathbf{2+2}$ if 
$(x, y), (z, w) \in P$ and $(z, y), (x, w) \in F$. 
Similarly, we say that $(x, y, z, w)$ is 
a forbidden configuration for $\mathbf{3+1}$ if 
$(x, y), (y, z) \in P$ and $(w, z), (x, w) \in F$. 
\begin{Lemma}\label{lemma:recognition:1}
A partial order $P$ on a set $V$ is a linear-semiorder if and only if 
there is an orientation $F$ of $E_o$ such that 
$P + F$ is acyclic and 
contains no forbidden configurations for $\mathbf{2+2}$ or $\mathbf{3+1}$. 
If the orientation $F$ exists, 
every linear extension of $P + F$ fulfills 
the $\mathbf{2+2}$ and $\mathbf{3+1}$ rules. 
\end{Lemma}

Before proving Lemma~\ref{lemma:recognition:1}, 
we show the following two observations. 
\begin{Proposition}\label{prop:for 2+2}
Let $x$, $y$, $z$, $w$ be four vertices of $V$. 
If $(x, y), (z, w) \in P$ and $xw, yz \in E(\overline{G})$ 
then $xz, yw \in E(\overline{G})$, 
that is, $(x, z, y, w)$ is $C_4$ of $\overline{G}$. 
\end{Proposition}
\begin{proof}
If $(x, z) \in P$ then $(z, w) \in P$ implies $(x, w) \in P$, contradicting $xw \in E(\overline{G})$. 
If $(z, x) \in P$ then $(x, y) \in P$ implies $(z, y) \in P$, contradicting $yz \in E(\overline{G})$. 
Thus $xz \in E(\overline{G})$. 
Similarly, we have $yw \in E(\overline{G})$. 
\end{proof}
\begin{Proposition}\label{prop:for 3+1}
Let $x$, $y$, $z$, $w$ be four vertices of $V$. 
If $(x, y), (y, z) \in P$ and $xw, zw \in E(\overline{G})$ 
then $yw \in E(\overline{G})$, 
that is, the four vertices induce $K_{1, 3}$ of $\overline{G}$. 
\end{Proposition}
\begin{proof}
If $(y, w) \in P$ then $(x, y) \in P$ implies $(x, w) \in P$, contradicting $xw \in E(\overline{G})$. 
If $(w, y) \in P$ then $(y, z) \in P$ implies $(w, z) \in P$, contradicting $zw \in E(\overline{G})$. 
Thus $yw \in E(\overline{G})$. 
\end{proof}

\begin{proof}[Proof of Lemma~\ref{lemma:recognition:1}]
The necessity is trivial 
from Proposition~\ref{prop:consequence of characterization}. 
We now show the sufficiency. 
Let $F$ be an orientation of $E_o$ such that 
$P + F$ is acyclic and 
contains no forbidden configurations for $\mathbf{2+2}$ or $\mathbf{3+1}$. 
Let $L$ be a linear extension of $P + F$, and let $F' = L - P$. 
Note that $F'$ is an orientation of $\overline{G}$ and $F' \supseteq F$. 
Suppose that $P + F'$ contains a forbidden configuration for $\mathbf{2+2}$, that is, 
there is a sequence of vertices $(x, y, z, w)$ with 
$(x, y), (z, w) \in P$ and $(z, y), (x, w) \in F'$. 
Since $xw, yz \in E(\overline{G})$, 
Proposition~\ref{prop:for 2+2} implies that 
$(x, z, y, w)$ is $C_4$ of $\overline{G}$. 
Thus $xw, yz \in E_o$. 
Since $F \subseteq F'$, we have $(z, y), (x, w) \in F$, 
contradicting that $P + F$ contains no forbidden configurations for $\mathbf{2+2}$. 
Similarly, when we suppose that 
$P + F'$ contains a forbidden configuration for $\mathbf{3+1}$, 
we have a contradiction from Proposition~\ref{prop:for 3+1}. 
\end{proof}

The order $\mathbf{2+1}$ of $P$ is the partial order 
consisting of three elements $x$, $y$, $z$ of $V$ such that 
$x \prec y$ whereas $x \parallel z$ and $y \parallel z$ in $P$. 
A \emph{forbidden configuration for $\mathbf{2+1}$} 
is a sequence of vertices $(x, y, z)$ 
with $(x, y) \in P$ and $(y, z), (z, x) \in F$. 

If $P + F$ is acyclic then 
$P + F$ contains no forbidden configurations for $\mathbf{2+1}$. 
We show in Theorem~\ref{theorem:recognition:structural result} that 
if $P + F$ contains no forbidden configurations for 
$\mathbf{2+2}$, $\mathbf{3+1}$, or $\mathbf{2+1}$, 
then even if $P + F$ is not acyclic, 
there is an orientation $F'$ of $E_o$ 
such that $P + F'$ is acyclic and 
contains no forbidden configurations for $\mathbf{2+2}$ or $\mathbf{3+1}$. 

\begin{Lemma}\label{lemma:recognition:2}
We can compute in $O(n \bar{m})$ time 
an orientation $F$ of $E_o$ such that 
$P + F$ contains no forbidden configurations for 
$\mathbf{2+2}$, $\mathbf{3+1}$, or $\mathbf{2+1}$. 
\end{Lemma}

We show two observations to prove Lemma~\ref{lemma:recognition:2}. 

\begin{Proposition}\label{prop:orientation of 2+2}
Let $F$ be an orientation of $E_o$ such that 
$P + F$ contains no forbidden configurations for $\mathbf{2+1}$. 
The orientation $P + F$ contains 
no forbidden configurations for $\mathbf{2+2}$ if and only if 
for any $C_4 = (x, y, z, w)$ of $\overline{G}$, either 
$(x, y), (z, y), (z, w), (x, w) \in F$ or 
$(y, x), (y, z), (w, z), (w, x) \in F$. 
\end{Proposition}
\begin{proof}
The sufficiency is trivial. We now show the necessity. 
Assume without loss of generality $(x, z), (y, w) \in P$. 
Suppose $(x, y) \in F$. 
If $(w, x) \in F$ then $(y, w, x)$ is 
a forbidden configuration for $\mathbf{2+1}$, a contradiction. 
Thus $(x, w) \in F$. 
If $(y, z) \in F$ then $(x, z, y, w)$ is 
a forbidden configuration for $\mathbf{2+2}$, a contradiction. 
Thus $(z, y) \in F$. 
If $(w, z) \in F$ then $(y, w, z)$ is 
a forbidden configuration for $\mathbf{2+1}$, a contradiction. 
Thus $(z, w) \in F$. 
\par
Conversely, suppose $(y, x) \in F$. 
If $(z, y) \in F$ then $(x, z, y)$ is 
a forbidden configuration for $\mathbf{2+1}$, a contradiction. 
Thus $(y, z) \in F$. 
If $(x, w) \in F$ then $(y, w, x, z)$ is 
a forbidden configuration for $\mathbf{2+2}$, a contradiction. 
Thus $(w, x) \in F$. 
If $(z, w) \in F$ then $(x, z, w)$ is 
a forbidden configuration for $\mathbf{2+1}$, a contradiction. 
Thus $(w, z) \in F$. 
\end{proof}

\begin{Proposition}\label{prop:orientation of 3+1}
Let $F$ be an orientation of $E_o$ such that 
$P + F$ contains no forbidden configurations for $\mathbf{2+1}$. 
The orientation $P + F$ contains 
no forbidden configurations for $\mathbf{3+1}$ if and only if 
for any $K_{1, 3}$ of $\overline{G}$ 
consisting of $x, y, z, w \in V$ 
with $xw, yw, zw \in E(\overline{G})$, either 
$(x, w), (y, w), (z, w) \in F$ or 
$(w, x), (w, y), (w, z) \in F$. 
\end{Proposition}
\begin{proof}
The sufficiency is trivial. We now show the necessity. 
Assume without loss of generality $(x, y), (y, z) \in P$. 
Suppose $(x, w) \in F$. 
If $(w, z) \in F$ then $(x, y, z, w)$ is 
a forbidden configuration for $\mathbf{3+1}$, a contradiction. 
Thus $(z, w) \in F$. 
If $(w, y) \in F$ then $(y, z, w)$ is 
a forbidden configuration for $\mathbf{2+1}$, a contradiction. 
Thus $(y, w) \in F$. 
\par 
Conversely, suppose $(w, x) \in F$. 
If $(y, w) \in F$ then $(x, y, w)$ is 
a forbidden configuration for $\mathbf{2+1}$, a contradiction. 
Thus $(w, y) \in F$. 
If $(z, w) \in F$ then $(y, z, w)$ is 
a forbidden configuration for $\mathbf{2+1}$, a contradiction. 
Thus $(w, z) \in F$. 
\end{proof}

\begin{proof}[Proof of Lemma~\ref{lemma:recognition:2}]
The \emph{neighborhood} of a vertex $v$ of $\overline{G}$ is 
the set $N_{\overline{G}}(v) = \{u \in V \colon\ uv \in E(\overline{G})\}$. 
We define that 
the \emph{upper set} of $v \in V$ is the set $U(v) = \{u \in V \colon\ (v, u) \in P\}$ and 
the \emph{lower set} of $v \in V$ is the set $L(v) = \{u \in V \colon\ (u, v) \in P\}$. 
We first show a series of claims. 
Let $uv$ be an edge of $\overline{G}$. 
\begin{claim}\label{claim:2-1}
If $U(u) \cap N_{\overline{G}}(v) \neq \emptyset$ and $N_{\overline{G}}(u) \cap U(v) \neq \emptyset$, then 
for any two vertices $w \in U(u) \cap N_{\overline{G}}(v)$ and $z \in N_{\overline{G}}(u) \cap U(v)$, 
we have that $(u, v, w, z)$ is $C_4$. 
Conversely, each $C_4$ has an edge $uv$ with 
$U(u) \cap N_{\overline{G}}(v) \neq \emptyset$ and 
$N_{\overline{G}}(u) \cap U(v) \neq \emptyset$. 
\end{claim}
\begin{proof}[Proof of Claim~\ref{claim:2-1}]
The second claim is trivial. 
We now show the first claim. 
If $(w, z) \in P$ then $(u, w) \in P$ implies $(u, z) \in P$, 
contradicting $uz \in E(\overline{G})$. 
If $(z, w) \in P$ then $(v, z) \in P$ implies $(v, w) \in P$, 
contradicting $vw \in E(\overline{G})$. 
Thus $wz \in E_{\overline{G}}(G)$, that is, $(u, v, w, z)$ is $C_4$. 
\end{proof}

By a similar argument, we have the following. 
\begin{claim}
If $L(u) \cap N_{\overline{G}}(v) \neq \emptyset$ and $N_{\overline{G}}(u) \cap L(v) \neq \emptyset$, then 
for any two vertices $w \in L(u) \cap N_{\overline{G}}(v)$ and $z \in N_{\overline{G}}(u) \cap L(v)$, 
we have that $(u, v, w, z)$ is $C_4$. 
Conversely, each $C_4$ has an edge $uv$ with 
$L(u) \cap N_{\overline{G}}(v) \neq \emptyset$ and 
$N_{\overline{G}}(u) \cap L(v) \neq \emptyset$. 
\end{claim}

The following is trivial. 
\begin{claim}
If $U(u) \cap N_{\overline{G}}(v) \neq \emptyset$ and $L(u) \cap N_{\overline{G}}(v) \neq \emptyset$, then 
for any two vertices $w \in U(u) \cap N_{\overline{G}}(v)$ and $z \in L(u) \cap N_{\overline{G}}(v)$, 
the vertices $z, u, w, v$ induce $K_{1, 3}$. 
Conversely, each $K_{1, 3}$ has an edge $uv$ with 
$U(u) \cap N_{\overline{G}}(v) \neq \emptyset$ and $L(u) \cap N_{\overline{G}}(v) \neq \emptyset$. 
\end{claim}

Therefore, we can compute $E_o$ by Algorithm~\ref{algo:E_o}. 
Each iteration on lines 2--23 takes $O(n)$ time for each edge of $\overline{G}$, 
and hence the running time is $O(n\bar{m})$. 

\input{figs/algo1.tex}

Let $P_3$ denote an induced path of length 2. 
The following is trivial. 
\begin{claim}\label{claim:2-*}
If $U(u) \cap N_{\overline{G}}(v) \neq \emptyset$ then 
for any vertex $w \in U(u) \cap N_{\overline{G}}(v)$, 
we have that $(u, v, w)$ is $P_3$. 
Conversely, each $P_3$ has an edge $uv$ with 
$U(u) \cap N_{\overline{G}}(v) \neq \emptyset$. 
\end{claim}

Now, we show an algorithm to find the orientation of $F$ of $E_o$ 
(Algorithm~\ref{algo:lemma2}). 
We use the 2-satisfiability problem in the algorithm. 
An instance of the 2-satisfiability problem is 
a \emph{2-CNF formula}, 
a Boolean formula in conjunctive normal form 
with at most two literals per clause. 
Assume that the vertices are linearly ordered, and 
we assign a Boolean variable $x_{(u, v)}$ for each edge $uv$ of $\overline{G}$, 
where $u$ precedes $v$ in the ordering. 
We denote the negation of $x_{(u, v)}$ by $x_{(v, u)}$. 
The algorithm first constructs the 2-CNF formula $\varphi$ (lines 1--21), 
then finds a satisfying truth assignment $\tau$ of $\varphi$ (line 22), 
and finally orients each edge $uv \in E_o$ as $(u, v) \in F$ 
if $x_{(u, v)} = 0$ in $\tau$ (line 23). 

\input{figs/algo2.tex}

The clauses added in lines 3--10 ensure that 
for any $C_4 = (x, y, z, w)$ of $\overline{G}$, either 
$(x, y), (z, y), (z, w), (x, w) \in F$ or 
$(y, x), (y, z), (w, z), (w, x) \in F$. 
Similarly, the clauses added in lines 11--18 ensure that 
for any $K_{1, 3}$ of $\overline{G}$ 
consisting of $x, y, z, w \in V$ 
with $xw, yw, zw \in E(\overline{G})$, either 
$(x, w), (y, w), (z, w) \in F$ or 
$(w, x), (w, y), (w, z) \in F$. 
Moreover, the clauses added in lines 19 and 20 ensure that 
$P + F$ contains no forbidden configurations for $\mathbf{2+1}$. 
Therefore, we have from Propositions~\ref{prop:orientation of 2+2} 
and~\ref{prop:orientation of 3+1} that 
$P + F$ contains no forbidden configurations for 
$\mathbf{2+2}$, $\mathbf{3+1}$, or $\mathbf{2+1}$. 

Each iteration on lines 2--21 takes $O(n)$ time for each edge of $E_o$, 
and hence $\varphi$ can be constructed in $O(n\bar{m})$ time. 
The 2-CNF formula $\varphi$ has at most $n (n-1) / 2$ Boolean variables and 
$O(n\bar{m})$ clauses. 
Since a satisfying truth assignment of $\varphi$ can be obtained 
in time linear to the size of $\varphi$~\cite{APT79-IPL}, 
Algorithm~\ref{algo:lemma2} takes $O(n\bar{m})$ time. 
\end{proof}

We call a sequence of vertices $(v_0, v_1, v_2, v_3, v_4, v_5)$ 
a \emph{regular obstruction} of $P + F$ if 
$(v_0, v_1), (v_2, v_3), (v_4, v_5) \in P$ and 
$(v_2, v_1), (v_4, v_3), (v_0, v_5) \in F$ 
(Fig.~\ref{figs:obstructions}\subref{figs:obstructions:regular}). 
We call $(v_0, v_1, v_2, v_3, v_4, v_5)$ 
a \emph{skewed obstruction} of $P + F$ if 
$(v_0, v_1), (v_1, v_2), (v_4, v_5) \in P$ and 
$(v_3, v_2), (v_4, v_3), (v_0, v_5), (v_5, v_3) \in F$ 
(Fig.~\ref{figs:obstructions}\subref{figs:obstructions:skewed}). 

\begin{figure*}[t]
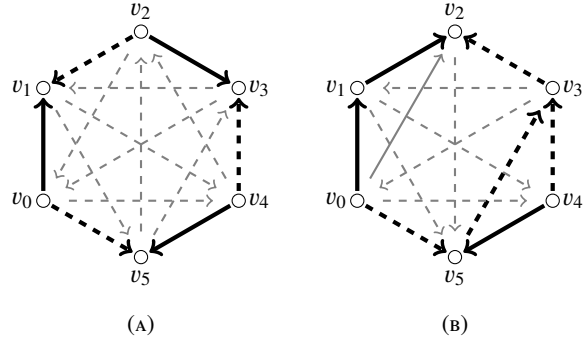

  \centering
  \subcaptionbox{\label{figs:obstructions:regular}}{\input{figs/obstructions/regular}}
  \subcaptionbox{\label{figs:obstructions:skewed}}{\input{figs/obstructions/skewed}}
  \caption{
    (a) A regular obstruction. 
    (b) A skewed obstruction. 
    An arrow $u \to v$ denotes edge $(u, v) \in P$, and 
    a dashed arrow $u \dasharrow v$ denotes edge $(u, v) \in F$. 
    }
  \label{figs:obstructions}
\end{figure*}

\begin{Lemma}\label{lemma:recognition:3}
Let $F$ be an orientation of $E_o$ such that 
$P + F$ contains no forbidden configurations for 
$\mathbf{2+2}$, $\mathbf{3+1}$, or $\mathbf{2+1}$. 
The orientation $P + F$ is acyclic if and only if 
$P + F$ contains no obstructions. 
\end{Lemma}

Before proving Lemma~\ref{lemma:recognition:3}, 
we introduce some observations. 

\begin{Proposition}\label{prop:for regular obstruction}
Let $F$ be an orientation of $E_o$ such that 
$P + F$ contains no forbidden configurations for 
$\mathbf{2+2}$ or $\mathbf{2+1}$. 
Let $(v_0, v_1, v_2, v_3, v_4, v_5)$ be a regular obstruction of $P + F$. 
All vertices of the obstruction are distinct. 
We have 
$(v_5, v_2), (v_5, v_3), (v_4, v_2), 
 (v_1, v_4), (v_1, v_5), (v_0, v_4), 
 (v_3, v_0), (v_3, v_1), (v_2, v_1) \in F$ 
(Fig.~\ref{figs:obstructions}\subref{figs:obstructions:regular}). 
\end{Proposition}
\begin{proof}
We have $v_i \neq v_{i+1}$ for any $i = 0, 1, \ldots, k-1$ (indices are modulo $k$) 
since the graphs $G$ and $\overline{G}$ have no loops. 
We have $v_i \neq v_{i+2}$ for any $i = 0, 1, \ldots, k-1$ since 
$(v_i, v_{i+1}) \in P$ and $(v_{i+2}, v_{i+1}) \in F$ if $i$ is even, and 
$(v_{i+1}, v_i) \in F$ and $(v_{i+1}, v_{i+2}) \in P$ if $i$ is odd. 
If $v_0 = v_3$ then $(v_2, v_3), (v_0, v_1) \in P$ implies $(v_2, v_1) \in P$, 
contradicting $(v_2, v_1) \in F$. Thus $v_0 \neq v_3$. 
Similarly, we have $v_1 \neq v_4$ and $v_2 \neq v_5$. 
Therefore, all vertices are distinct. 
\par
If $(v_2, v_5) \in P$ then $(v_2, v_5, v_0, v_1)$ is 
a forbidden configuration for $\mathbf{2+2}$, a contradiction. 
If $(v_5, v_2) \in P$ then 
$(v_4, v_5), (v_2, v_3) \in P$ implies $(v_4, v_3) \in P$, 
contradicting $(v_4, v_3) \in F$. 
Thus $v_2v_5 \in E(\overline{G})$. 
Proposition~\ref{prop:for 2+2} implies that 
$(v_2, v_4, v_3, v_5)$ is $C_4$ in $\overline{G}$. 
Since $(v_4, v_3) \in F$, Proposition~\ref{prop:orientation of 2+2} implies 
$(v_5, v_2), (v_5, v_3), (v_4, v_2) \in F$. 
By similar arguments, we have 
$(v_1, v_4), (v_1, v_5), (v_0, v_4) \in F$ and 
$(v_3, v_0), (v_3, v_1), (v_2, v_0) \in F$. 
\end{proof}

\begin{Proposition}\label{prop:for skewed obstruction}
Let $F$ be an orientation of $E_o$ such that 
$P + F$ contains no forbidden configurations for 
$\mathbf{2+2}$, $\mathbf{3+1}$, or $\mathbf{2+1}$. 
Let $(v_0, v_1, v_2, v_3, v_4, v_5)$ be a skewed obstruction of $P + F$. 
All vertices of the obstruction are distinct. 
We have $(v_0, v_2) \in P$ and 
$(v_2, v_5), 
 (v_1, v_4), (v_1, v_5), (v_0, v_4), 
 (v_3, v_0), (v_3, v_1) \in F$ 
(Fig.~\ref{figs:obstructions}\subref{figs:obstructions:skewed}). 
\end{Proposition}
\begin{proof}
We have $v_i \neq v_{i+1}$ for any $i = 0, 1, \ldots, k-1$ (indices are modulo $k$) 
since the graphs $G$ and $\overline{G}$ have no loops. 
We have $v_0 \neq v_2$ from $(v_0, v_1), (v_1, v_2) \in P$. 
We have $v_1 \neq v_3$ from $(v_1, v_2) \in P$ and $(v_3, v_2) \in F$. 
We have $v_2 \neq v_4$ from $(v_3, v_2), (v_4, v_3) \in F$. 
We have $v_3 \neq v_5$ from $(v_4, v_3) \in F$ and $(v_4, v_5) \in P$. 
We have $v_4 \neq v_0$ from $(v_4, v_5) \in P$ and $(v_0, v_5) \in F$. 
We have $v_5 \neq v_1$ from $(v_0, v_5) \in F$ and $(v_0, v_1) \in P$. 
If $v_0 = v_3$ then $(v_0, v_1), (v_1, v_2) \in P$ implies $(v_0, v_2) \in P$, 
contradicting $(v_3, v_2) \in F$. Thus $v_0 \neq v_3$. 
If $v_1 = v_4$ then $(v_0, v_1), (v_4, v_5) \in P$ implies $(v_0, v_5) \in P$, 
contradicting $(v_0, v_5) \in F$. Thus $v_1 \neq v_4$. 
If $v_2 = v_5$ then $(v_0, v_1), (v_1, v_2) \in P$ implies $(v_0, v_2) \in P$, 
contradicting $(v_0, v_5) \in F$. Thus $v_2 \neq v_5$. 
Therefore, all vertices are distinct. 
\par
We have $(v_0, v_2) \in P$ from $(v_0, v_1), (v_1, v_2) \in P$. 
If $(v_2, v_5) \in P$ then 
$(v_0, v_2) \in P$ implies $(v_0, v_5) \in P$, 
contradicting $(v_0, v_5) \in F$. 
If $(v_5, v_2) \in P$ then $(v_4, v_5, v_2, v_3)$ is 
a forbidden configuration for $\mathbf{3+1}$, a contradiction. 
Thus $v_2v_5 \in E(\overline{G})$. 
Proposition~\ref{prop:for 3+1} implies that 
the vertices $v_0, v_1, v_2, v_5$ induce $K_{1, 3}$ in $\overline{G}$. 
Since $(v_0, v_5) \in F$, Proposition~\ref{prop:orientation of 3+1} implies 
$(v_2, v_5), (v_1, v_5) \in F$. 
\par
If $(v_1, v_4) \in P$ then 
$(v_0, v_1), (v_4, v_5) \in P$ implies $(v_0, v_5) \in P$, 
contradicting $(v_0, v_5) \in F$. 
If $(v_4, v_1) \in P$ then $(v_4, v_1, v_2, v_3)$ is 
a forbidden configuration for $\mathbf{3+1}$, a contradiction. 
Thus $v_1v_4 \in E(\overline{G})$. 
Proposition~\ref{prop:for 2+2} implies that 
$(v_0, v_4, v_1, v_5)$ is $C_4$ in $\overline{G}$. 
Since $(v_0, v_5) \in F$, Proposition~\ref{prop:orientation of 2+2} implies 
$(v_0, v_4), (v_1, v_5), (v_1, v_4) \in F$. 
\par
If $(v_0, v_3) \in P$ then $(v_0, v_3, v_4, v_5)$ is 
a forbidden configuration for $\mathbf{2+2}$, a contradiction. 
If $(v_3, v_0) \in P$ then 
$(v_0, v_2) \in P$ implies $(v_3, v_2) \in P$, 
contradicting $(v_3, v_2) \in F$. 
Thus $v_0v_3 \in E(\overline{G})$. 
Proposition~\ref{prop:for 3+1} implies that 
the vertices $v_0, v_1, v_2, v_3$ induce $K_{1, 3}$ in $\overline{G}$. 
Since $(v_3, v_2) \in F$, 
Proposition~\ref{prop:orientation of 3+1} implies 
$(v_3, v_0), (v_3, v_1) \in F$. 
\end{proof}

A sequence of vertices $(a_0, b_0, a_1, b_1, \ldots, a_{k-1}, b_{k-1})$ 
with $k \geq 2$ is 
an \emph{alternating cycle of length $2k$} 
if $(a_i, b_i) \in P$ and $(b_i, a_{i+1}) \in F$ 
for any $i = 0, 1, \ldots, k-1$ (indices are modulo $k$). 
See Fig.~\ref{figs:alternating-cycles} for example. 

\begin{figure*}[t]
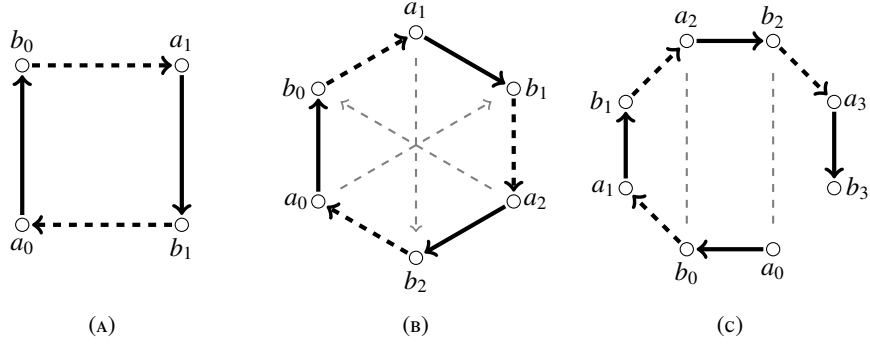

  \centering
  \subcaptionbox{\label{figs:alternating-cycles:4-cycle}}{\input{figs/alternating-cycles/4-cycle}}
  \subcaptionbox{\label{figs:alternating-cycles:6-cycle}}{\input{figs/alternating-cycles/6-cycle}}
  \subcaptionbox{\label{figs:alternating-cycles:8-cycle}}{\input{figs/alternating-cycles/8-cycle}}
  \caption{
    (a) An alternating cycle of length 4. 
    (b) An alternating cycle of length 6. 
    (c) An alternating cycle of length greater than 6. 
    Arrows and dashed arrows denote the same type of edges 
    as in Fig.~\ref{figs:obstructions}. 
    Dashed lines denote edges of $\overline{G}$. 
  }
  \label{figs:alternating-cycles}
\end{figure*}

\begin{Proposition}\label{prop:alternating 4-cycle}
Let $F$ be an orientation of $E_o$ such that 
$P + F$ contains no forbidden configurations for $\mathbf{2+1}$. 
The orientation $P + F$ contains no alternating cycles of length 4. 
\end{Proposition}
\begin{proof}
Suppose that $P + F$ contains an alternating cycle 
$(a_0, b_0, a_1, b_1)$ of length 4 
(Fig.~\ref{figs:alternating-cycles}\subref{figs:alternating-cycles:4-cycle}). 
Proposition~\ref{prop:for 2+2} implies that 
$(a_0, a_1, b_1, b_0)$ is $C_4$, and hence $a_0a_1 \in E_o$. 
If $(a_0, a_1) \in F$ then $(a_1, b_1, a_0)$ is 
a forbidden configuration for $\mathbf{2+1}$, a contradiction. 
If $(a_1, a_0) \in F$ then $(a_0, b_0, a_1)$ is 
a forbidden configuration for $\mathbf{2+1}$, a contradiction. 
\end{proof}

Now, we prove Lemma~\ref{lemma:recognition:3}. 
\begin{proof}[Proof of Lemma~\ref{lemma:recognition:3}]
We have from Propositions~\ref{prop:for regular obstruction} 
and~\ref{prop:for skewed obstruction} that 
if $P + F$ contains an obstruction, then $P + F$ contains a cycle. 
Now, we suppose that $P + F$ contains a cycle 
but no obstructions. 

\begin{claim}\label{claim:3-1}
If $P + F$ contains an alternating cycle of length 6, 
then $P + F$ contains a regular obstruction. 
\end{claim}
\begin{proof}[Proof of Claim~\ref{claim:3-1}]
Let $(a_0, b_0, a_1, b_1, a_2, b_3)$ be an alternating cycle of length 6 
with $(a_i, b_i) \in P$ and $(b_i, a_{i+1}) \in F$ 
for any $i = 0, 1, 2$ (indices are modulo 3). 
See Fig.~\ref{figs:alternating-cycles}\subref{figs:alternating-cycles:6-cycle}. 
We first claim that the vertices are distinct. 
We have $a_i \neq b_i$ and $b_i \neq a_{i+1}$ 
for any $i = 0, 1, 2$ (indices are modulo 3) 
since the graphs $G$ and $\overline{G}$ have no loops. 
We have $a_i \neq a_{i+1}$ for any $i = 0, 1, 2$ 
since $(a_i, b_i) \in P$ and $(b_i, a_{i+1}) \in F$. 
We also have $b_i \neq b_{i+1}$ for any $i = 0, 1, 2$ 
since $(b_i, a_{i+1}) \in F$ and $(a_{i+1}, b_{i+1}) \in P$. 
If $a_0 = b_1$ then $(a_1, b_1), (a_0, b_0) \in P$ implies $(a_1, b_0) \in P$, 
contradicting $(b_0, a_1) \in F$. Thus $a_0 \neq b_1$. 
Similarly, we have $a_1 \neq b_2$ and $a_2 \neq b_0$. 
Therefore, all vertices are distinct. 
\par
If $(a_1, b_2) \in P$ then $(a_1, b_2, a_0, b_0)$ is 
an alternating cycle of length 4, 
contradicting Proposition~\ref{prop:alternating 4-cycle}. 
If $(b_2, a_1) \in P$ then $(a_2, b_2), (a_1, b_1) \in P$ implies $(a_2, b_1) \in P$, 
contradicting $(b_1, a_2) \in F$. Thus $a_1b_2 \in E(\overline{G})$. 
Proposition~\ref{prop:for 2+2} implies that 
$(a_1, a_2, b_1, b_2)$ is $C_4$ in $\overline{G}$. 
Since $(b_1, a_2) \in F$, 
Proposition~\ref{prop:orientation of 2+2} implies $(a_1, b_2) \in F$. 
By similar arguments, we have $(a_0, b_1), (a_2, b_0) \in F$. 
Now, $(a_0, b_0, a_2, b_2, a_1, b_1)$ is a regular obstruction. 
\end{proof}

\begin{claim}\label{claim:3-2}
If $P + F$ contains an alternating cycle of length greater than 6, 
then $P + F$ contains an alternating cycle of length 6. 
\end{claim}
\begin{proof}[Proof of Claim~\ref{claim:3-2}]
Suppose that $P + F$ contains no alternating cycles of length 6. 
We show that $P + F$ contains no alternating cycles of length greater than 6. 
Suppose to the contrary that 
$P + F$ contains an alternating cycle $(a_0, b_0, a_1, b_1, \ldots, a_{k-1}, b_{k-1})$ 
of length $2k$ with $k \geq 4$ 
(Fig.~\ref{figs:alternating-cycles}\subref{figs:alternating-cycles:8-cycle}). 
Assume without loss of generality that the length of the cycle is minimal. 
\par
If $a_0 = b_2$ then $(a_2, b_2), (a_0, b_0) \in P$ implies $(a_2, b_0) \in P$, 
and hence $(a_2, b_0, a_1, b_1)$ is an alternating cycle of length 4, 
contradicting Proposition~\ref{prop:alternating 4-cycle}. 
If $(a_0, b_2) \in P$ then $(a_0, b_2, a_3, b_3, \ldots, a_{k-1}, b_{k-1})$ 
is an alternating cycle of length $2(k-2)$, 
contradicting the minimality of the length of the cycle. 
If $(b_2, a_0) \in P$ then 
$(a_2, b_2), (a_0, b_0) \in P$ implies $(a_2, b_0) \in P$, 
and hence $(a_2, b_0, a_1, b_1)$ is an alternating cycle of length 4, 
contradicting Proposition~\ref{prop:alternating 4-cycle}. 
Therefore, $a_0b_2 \in E(\overline{G})$. 
\par
If $b_0 = a_2$ then $(a_0, b_0), (a_2, b_2) \in P$ implies $(a_0, b_2) \in P$, 
contradicting $a_0b_2 \in E(\overline{G})$. 
If $(a_2, b_0) \in P$ then 
$(a_2, b_0, a_1, b_1)$ is an alternating cycle of length 4, 
contradicting Proposition~\ref{prop:alternating 4-cycle}. 
If $(b_0, a_2) \in P$ then 
$(a_0, b_0), (a_2, b_2) \in P$ implies $(a_0, b_2) \in P$, 
contradicting $a_0b_2 \in E(\overline{G})$. 
Therefore, $b_0a_2 \in E(\overline{G})$. 
\par
Now, Proposition~\ref{prop:for 2+2} implies that 
$(a_0, a_2, b_0, b_2)$ is $C_4$. 
If $(a_0, b_2), (b_0, a_2) \in F$ then 
$(a_0, b_0, a_2, b_2, a_3, b_3, \ldots, a_{k-1}, b_{k-1})$ 
is an alternating cycle of length $2(k-1)$, 
a contradiction. 
If $(b_2, a_0), (a_2, b_0) \in F$ then 
$(a_0, b_0, a_1, b_1, a_2, b_2)$ is an alternating cycle of length 6, 
a contradiction. 
Thus $P + F$ contains no alternating cycles of length greater than 6. 
\end{proof}
Therefore, we can assume that $P + F$ contains no alternating cycles. 

We now introduce some definitions. 
We call an edge $(x, y) \in F$ \emph{type-1 critical} 
if there is $\mathbf{3+1}$ 
consisting of four vertices $x$, $y$, $z$, $w$ 
with $(y, z), (z, w) \in P$ and $(x, y), (x, z), (x, w) \in F$ 
(Fig.~\ref{figs:lemma3:critical-edges}\subref{figs:lemma3:type1}). 
We call an edge $(x, y) \in F$ \emph{type-2 critical} 
if there is $\mathbf{2+2}$ 
consisting of four vertices $x$, $y$, $z$, $w$ 
with $(w, x), (y, z) \in P$ and $(x, y), (x, z), (w, y), (w, z) \in F$ 
(Fig.~\ref{figs:lemma3:critical-edges}\subref{figs:lemma3:type2}). 
We call an edge $(x, y) \in F$ \emph{type-3 critical} 
if there is $\mathbf{3+1}$ 
consisting of four vertices $x$, $y$, $z$, $w$ 
with $(z, w), (w, x) \in P$ and $(x, y), (w, y), (z, y) \in F$ 
(Fig.~\ref{figs:lemma3:critical-edges}\subref{figs:lemma3:type3}). 

\begin{figure*}[t]
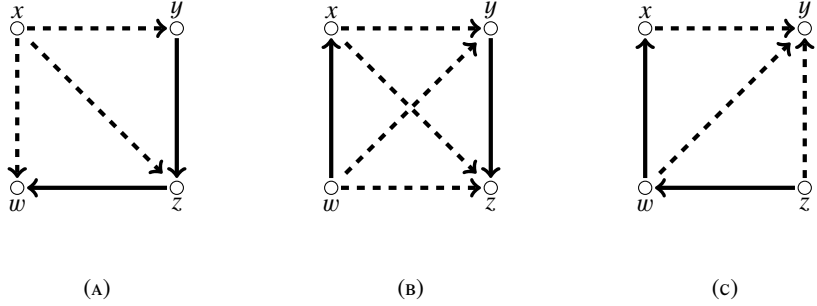

  \centering
  \subcaptionbox{\label{figs:lemma3:type1}}{\input{figs/lemma3/type1}}
  \subcaptionbox{\label{figs:lemma3:type2}}{\input{figs/lemma3/type2}}
  \subcaptionbox{\label{figs:lemma3:type3}}{\input{figs/lemma3/type3}}
  \caption{
    (a) A type-1 critical edge $(x, y)$. 
    (b) A type-2 critical edge $(x, y)$. 
    (c) A type-3 critical edge $(x, y)$. 
    Arrows and dashed arrows denote the same type of edges 
    as in Fig.~\ref{figs:alternating-cycles}. 
  }
  \label{figs:lemma3:critical-edges}
\end{figure*}

For any edge $(u, v) \in F$, 
there is a directed path from $u$ to $v$ 
consisting of a critical edge and edges in $P$. 
For example, in Fig.~\ref{figs:lemma3:critical-edges}\subref{figs:lemma3:type1}, 
there is a path $(x, y, w)$ for $(x, w) \in F$. 
Recall that $P + F$ contains a cycle. 
Therefore, $P + F$ contains a cycle consisting of critical edges and edges in $P$. 

Let $C$ be a shortest cycle consisting of critical edges and edges in $P$. 
If there are three vertices $u, v, w$ consecutive on $C$ with $(u, v), (v, w) \in P$, 
then $(u, w) \in P$. Hence we have 
a shorter cycle consisting of critical edges and edges in $P$, 
contradicting the minimality of length of $C$. 
Thus there are no three vertices $u, v, w$ consecutive on $C$ with $(u, v), (v, w) \in P$. 

Recall that $P + F$ contains no alternating cycles. 
Therefore, there are three vertices $u, v, w$ consecutive on $C$ with $(u, v), (v, w) \in F$. 
If $(u, w) \in P$ or $(w, u) \in P$, then we have 
a shorter cycle consisting of critical edges and edges in $P$, 
a contradiction. 
Thus $uw \in E(\overline{G})$. 
If $(u, w) \in F$ and $(u, w)$ is a critical edge, then we have 
a shorter cycle consisting of critical edges and edges in $P$, a contradiction. 
Thus we have the following. 
\begin{claim}\label{claim:3-3}
If $(u, w) \in F$ then $(u, w)$ is not a critical edge. 
\end{claim}

We distinguish nine cases (Fig.~\ref{figs:lemma3}). 
Case 1-1: $(u, v)$ is type-1 critical and $(v, w)$ is type-1 critical; 
Case 1-2: $(u, v)$ is type-1 critical and $(v, w)$ is type-2 critical; 
Case 1-3: $(u, v)$ is type-1 critical and $(v, w)$ is type-3 critical; 
Case 2-1: $(u, v)$ is type-2 critical and $(v, w)$ is type-1 critical; 
Case 2-2: $(u, v)$ is type-2 critical and $(v, w)$ is type-2 critical; 
Case 2-3: $(u, v)$ is type-2 critical and $(v, w)$ is type-3 critical; 
Case 3-1: $(u, v)$ is type-3 critical and $(v, w)$ is type-1 critical; 
Case 3-2: $(u, v)$ is type-3 critical and $(v, w)$ is type-2 critical; 
Case 3-3: $(u, v)$ is type-3 critical and $(v, w)$ is type-3 critical. 

\begin{figure*}[t]
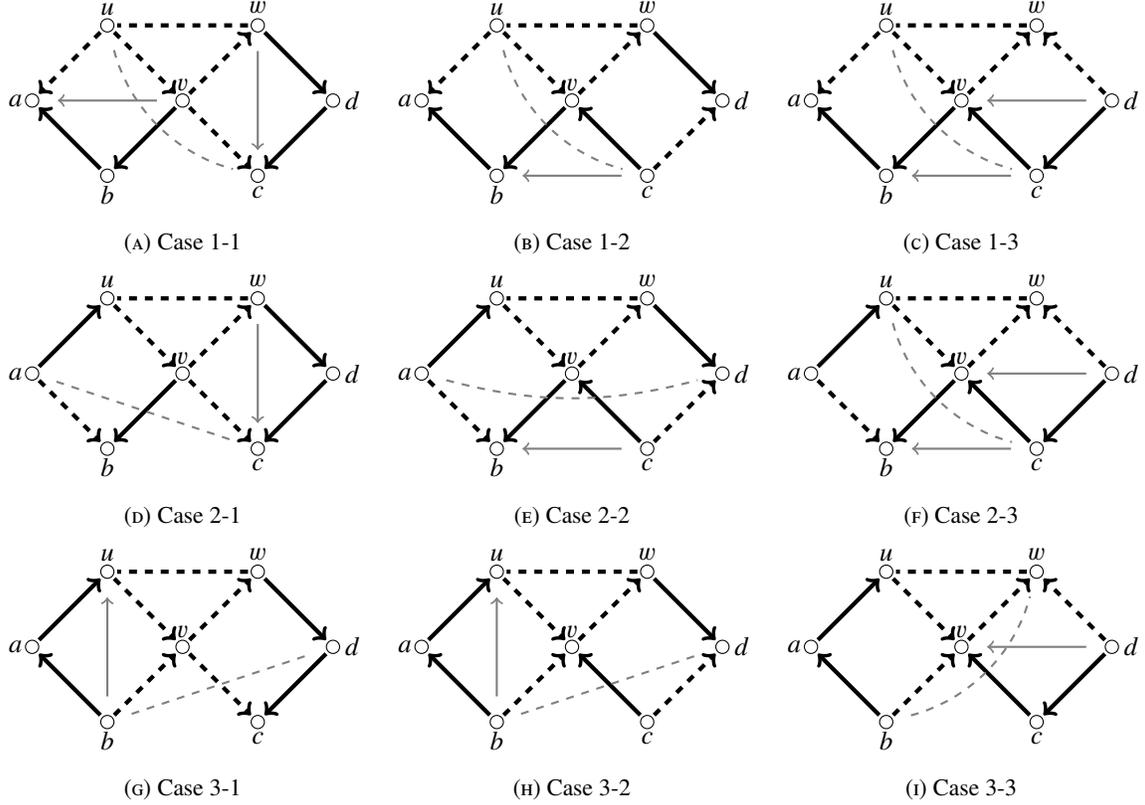

  \centering
  \subcaptionbox{Case 1-1\label{figs:lemma3:case11}}{\input{figs/lemma3/case11}}
  \subcaptionbox{Case 1-2\label{figs:lemma3:case12}}{\input{figs/lemma3/case12}}
  \subcaptionbox{Case 1-3\label{figs:lemma3:case13}}{\input{figs/lemma3/case13}}
\\
  \subcaptionbox{Case 2-1\label{figs:lemma3:case21}}{\input{figs/lemma3/case21}}
  \subcaptionbox{Case 2-2\label{figs:lemma3:case22}}{\input{figs/lemma3/case22}}
  \subcaptionbox{Case 2-3\label{figs:lemma3:case23}}{\input{figs/lemma3/case23}}
\\
  \subcaptionbox{Case 3-1\label{figs:lemma3:case31}}{\input{figs/lemma3/case31}}
  \subcaptionbox{Case 3-2\label{figs:lemma3:case32}}{\input{figs/lemma3/case32}}
  \subcaptionbox{Case 3-3\label{figs:lemma3:case33}}{\input{figs/lemma3/case33}}
  \caption{
    Illustrating the proof of Claims~\ref{claim:3-4} and~\ref{claim:3-5}. 
    Arrows, dashed arrows, and dashed lines denote the same type of edges 
    as in Fig.~\ref{figs:alternating-cycles}. 
  }
  \label{figs:lemma3}
\end{figure*}

\begin{claim}\label{claim:3-4}
In Cases 2-1, 2-2, and 3-2, the orientation $P + F$ contains a regular obstruction. 
In Cases 1-1, 3-1, and 3-3, the orientation $P + F$ contains a skewed obstruction. 
\end{claim}
\begin{proof}[Proof of Claim~\ref{claim:3-4}]
We prove the claim for each case. 
\par
\emph{Case~1-1} 
(Fig.~\ref{figs:lemma3}\subref{figs:lemma3:case11}). 
Since $(u, v)$ is type-1 critical, 
there are two other vertices $a, b$ with 
$(v, b), (b, a) \in P$ and $(u, a) \in F$. 
Note that $(v, b), (b, a) \in P$ implies $(v, a) \in P$. 
We have $a \neq w$ from $(v, a) \in P$ and $(v, w) \in F$. 
We also have $b \neq w$ from $(v, b) \in P$ and $(v, w) \in F$. 
Therefore, the five vertices $u, v, w, a, b$ are distinct. 
\par
Since $(v, w)$ is type-1 critical, 
there are two other vertices $c, d$ with 
$(w, d), (d, c) \in P$ and $(v, c) \in F$. 
Note that $(w, d), (d, c) \in P$ implies $(w, c) \in P$. 
We have $c \neq u$ from $(v, c), (u, v) \in F$. 
We also have $c \neq b$ from $(v, c) \in F$ and $(v, b) \in P$. 
Moreover, we have $c \neq a$ from $(v, c) \in F$ and $(v, a) \in P$. 
Therefore, the six vertices $u, v, w, a, b, c$ are distinct. 
We have $d \neq u$ from $(w, d) \in P$ and $uw \in E(\overline{G})$. 
If $d = b$ then $(v, b), (d, c) \in P$ implies $(v, c) \in P$, 
contradicting $(v, c) \in F$. Thus $d \neq b$. 
Similarly, 
if $d = a$ then $(v, a), (d, c) \in P$ implies $(v, c) \in P$, 
contradicting $(v, c) \in F$. Thus $d \neq a$. 
Therefore, the seven vertices $u, v, w, a, b, c, d$ are distinct. 
\par
If $(c, u) \in P$ then $(w, c) \in P$ implies $(w, u) \in P$, 
contradicting $uw \in E(\overline{G})$. 
If $(u, c) \in P$ then $(u, c, v, a)$ is 
a forbidden configuration for $\mathbf{2+2}$, 
a contradiction. Thus $uc \in E(\overline{G})$. 
Proposition~\ref{prop:for 3+1} implies that 
the vertices $w, d, c, u$ induce $K_{1, 3}$ in $\overline{G}$. 
If $(u, w) \in F$ then $(u, w)$ is type-1 critical, 
contradicting Claim~\ref{claim:3-3}. 
Thus $(w, u) \in F$, and then 
$(v, b, a, u, w, c)$ is a skewed obstruction. 
\par
\emph{Case~2-1} 
(Fig.~\ref{figs:lemma3}\subref{figs:lemma3:case21}). 
Since $(u, v)$ is type-2 critical, 
there are two other vertices $a, b$ with 
$(a, u), (v, b) \in P$ and $(a, b) \in F$. 
We have $a \neq w$ from $(a, u) \in P$ and $uw \in E(\overline{G})$. 
We also have $b \neq w$ from $(v, b) \in P$ and $(v, w) \in F$. 
Therefore, the five vertices $u, v, w, a, b$ are distinct. 
\par
Since $(v, w)$ is type-1 critical, 
there are two other vertices $c, d$ with 
$(w, d), (d, c) \in P$ and $(v, c) \in F$. 
Note that $(w, d), (d, c) \in P$ implies $(w, c) \in P$. 
We have $c \neq u$ from $(v, c), (u, v) \in F$. 
We also have $c \neq b$ from $(v, c) \in F$ and $(v, b) \in P$. 
If $c = a$ then $(w, c), (a, u) \in P$ implies $(w, u) \in P$, 
contradicting $uw \in E(\overline{G})$. Thus $c \neq a$. 
Therefore, the six vertices $u, v, w, a, b, c$ are distinct. 
We have $d \neq u$ from $(w, d) \in P$ and $uw \in E(\overline{G})$. 
If $d = b$ then $(v, b), (d, c) \in P$ implies $(v, c) \in P$, 
contradicting $(v, c) \in F$. Thus $d \neq b$. 
If $d = a$ then $(w, d), (a, u) \in P$ implies $(w, u) \in P$, 
contradicting $uw \in E(\overline{G})$. Thus $d \neq a$. 
Therefore, the seven vertices $u, v, w, a, b, c, d$ are distinct. 
\par
If $(c, a) \in P$ then $(w, c), (a, u) \in P$ implies $(w, u) \in P$, 
contradicting $uw \in E(\overline{G})$. 
If $(a, c) \in P$ then $(a, c, v, b)$ is 
a forbidden configuration for $\mathbf{2+2}$, 
a contradiction. Thus $ac \in E(\overline{G})$. 
Proposition~\ref{prop:for 2+2} implies that 
$(a, w, u, c)$ is $C_4$ in $\overline{G}$. 
If $(u, w) \in F$ then $(u, w)$ is type-2 critical, 
contradicting Claim~\ref{claim:3-3}. 
Thus $(w, u) \in F$, and then 
$(a, u, w, c, v, b)$ is a regular obstruction. 
\par
\emph{Case~2-2} 
(Fig.~\ref{figs:lemma3}\subref{figs:lemma3:case22}). 
As in Case~2-1, there are two vertices $a, b$ with 
$(a, u), (v, b) \in P$ and $(a, b) \in F$, and 
the five vertices $u, v, w, a, b$ are distinct. 
Since $(v, w)$ is type-2 critical, 
there are two other vertices $c, d$ with 
$(c, v), (w, d) \in P$ and $(c, d) \in F$. 
We have $c \neq u$ from $(c, v) \in P$ and $(u, v) \in F$. 
We also have $c \neq b$ from $(c, v), (v, b) \in P$. 
Thus we have $(c, b) \in P$ from $(c, v), (v, b) \in P$. 
We have $c \neq a$ from $(c, b) \in P$ and $(a, b) \in F$. 
Therefore, the six vertices $u, v, w, a, b, c$ are distinct. 
We have $d \neq u$ from $(w, d) \in P$ and $uw \in E(\overline{G})$. 
We also have $d \neq b$ from $(c, d) \in F$ and $(c, b) \in P$. 
If $d = a$ then $(w, d), (a, u) \in P$ implies $(w, u) \in P$, 
contradicting $uw \in E(\overline{G})$. Thus $d \neq a$. 
Therefore, the seven vertices $u, v, w, a, b, c, d$ are distinct. 
\par
If $(d, a) \in P$ then $(w, d), (a, u) \in P$ implies $(w, u) \in P$, 
contradicting $uw \in E(\overline{G})$. 
If $(a, d) \in P$ then $(a, d, c, b)$ is 
a forbidden configuration for $\mathbf{2+2}$, 
a contradiction. Thus $ad \in E(\overline{G})$. 
Proposition~\ref{prop:for 2+2} implies that 
$(a, w, u, d)$ is $C_4$ in $\overline{G}$. 
If $(u, w) \in F$ then $(u, w)$ is type-2 critical, 
contradicting Claim~\ref{claim:3-3}. 
Thus $(w, u) \in F$, and then 
$(a, u, w, d, c, b)$ is a regular obstruction. 
\par
\emph{Case~3-1} 
(Fig.~\ref{figs:lemma3}\subref{figs:lemma3:case31}). 
Since $(u, v)$ is type-3 critical, 
there are two other vertices $a, b$ with 
$(b, a), (a, u) \in P$ and $(b, v) \in F$. 
Note that $(b, a), (a, u) \in P$ implies $(b, u) \in P$. 
We have $a \neq w$ from $(a, u) \in P$ and $uw \in E(\overline{G})$. 
We also have $b \neq w$ from $(b, v), (v, w) \in F$. 
Therefore, the five vertices $u, v, w, a, b$ are distinct. 
\par
Since $(v, w)$ is type-1 critical, 
there are two other vertices $c, d$ with 
$(w, d), (d, c) \in P$ and $(v, c) \in F$. 
We have $c \neq u$ from $(v, c), (u, v) \in F$. 
We also have $c \neq b$ from $(v, c), (b, v) \in F$. 
If $c = a$ then $(w, d), (d, c), (a, u) \in P$ implies $(w, u) \in P$, 
contradicting $uw \in E(\overline{G})$. Thus $c \neq a$. 
Therefore, the six vertices $u, v, w, a, b, c$ are distinct. 
We have $d \neq u$ from $(w, d) \in P$ and $uw \in E(\overline{G})$. 
If $d = b$ then $(w, d), (b, u) \in P$ implies $(w, u) \in P$, 
contradicting $uw \in E(\overline{G})$. Thus $d \neq b$. 
Similarly, 
if $d = a$ then $(w, d), (a, u) \in P$ implies $(w, u) \in P$, 
contradicting $uw \in E(\overline{G})$. Thus $d \neq a$. 
Therefore, the seven vertices $u, v, w, a, b, c, d$ are distinct. 
\par
If $(d, b) \in P$ then $(w, d), (b, u) \in P$ implies $(w, u) \in P$, 
contradicting $uw \in E(\overline{G})$. 
If $(b, d) \in P$ then $(b, d, c, v)$ is 
a forbidden configuration for $\mathbf{3+1}$, 
a contradiction. Thus $bd \in E(\overline{G})$. 
Proposition~\ref{prop:for 2+2} implies that 
$(b, w, u, d)$ is $C_4$ in $\overline{G}$. 
If $(u, w) \in F$ then $(u, w)$ is type-2 critical, 
contradicting Claim~\ref{claim:3-3}. 
Thus $(w, u) \in F$, and then 
$(w, d, c, v, b, u)$ is a skewed obstruction. 
\par
\emph{Case~3-2} 
(Fig.~\ref{figs:lemma3}\subref{figs:lemma3:case32}). 
As in Case~3-1, there are two vertices $a, b$ with 
$(b, a), (a, u) \in P$ and $(b, v) \in F$, and 
the five vertices $u, v, w, a, b$ are distinct. 
Note that $(b, a), (a, u) \in P$ implies $(b, u) \in P$. 
Since $(v, w)$ is type-2 critical, 
there are two other vertices $c, d$ with 
$(c, v), (w, d) \in P$ and $(c, d) \in F$. 
We have $c \neq u$ from $(c, v) \in P$ and $(u, v) \in F$. 
We also have $c \neq b$ from $(c, v) \in P$ and $(b, v) \in F$. 
If $c = a$ then $(b, a), (c, v) \in P$ implies $(b, v) \in P$, 
contradicting $(b, v) \in F$. Thus $c \neq a$. 
Therefore, the six vertices $u, v, w, a, b, c$ are distinct. 
We have $d \neq u$ from $(w, d) \in P$ and $uw \in E(\overline{G})$. 
If $d = b$ then $(w, d), (b, u) \in P$ implies $(w, u) \in P$, 
contradicting $uw \in E(\overline{G})$. Thus $d \neq b$. 
Similarly, 
if $d = a$ then $(w, d), (a, u) \in P$ implies $(w, u) \in P$, 
contradicting $uw \in E(\overline{G})$. Thus $d \neq a$. 
Therefore, the seven vertices $u, v, w, a, b, c, d$ are distinct. 
\par
If $(d, b) \in P$ then $(w, d), (b, u) \in P$ implies $(w, u) \in P$, 
contradicting $uw \in E(\overline{G})$. 
If $(b, d) \in P$ then $(b, d, c, v)$ is 
a forbidden configuration for $\mathbf{2+2}$, 
a contradiction. Thus $bd \in E(\overline{G})$. 
Proposition~\ref{prop:for 2+2} implies that 
$(b, w, u, d)$ is $C_4$ in $\overline{G}$. 
If $(u, w) \in F$ then $(u, w)$ is type-2 critical, 
contradicting Claim~\ref{claim:3-3}. 
Thus $(w, u) \in F$, and then 
$(b, u, w, d, c, v)$ is a regular obstruction. 
\par
\emph{Case~3-3} 
(Fig.~\ref{figs:lemma3}\subref{figs:lemma3:case33}). 
As in Case~3-1, there are two vertices $a, b$ with 
$(b, a), (a, u) \in P$ and $(b, v) \in F$, and 
the five vertices $u, v, w, a, b$ are distinct. 
Since $(v, w)$ is type-3 critical, 
there are two other vertices $c, d$ with 
$(d, c), (c, v) \in P$ and $(d, w) \in F$. 
Note that $(d, c), (c, v) \in P$ implies $(d, v) \in P$. 
We have $c \neq u$ from $(c, v) \in P$ and $(u, v) \in F$. 
We also have $c \neq b$ from $(c, v) \in P$ and $(b, v) \in F$. 
If $c = a$ then $(b, a), (c, v) \in P$ implies $(b, v) \in P$, 
contradicting $(b, v) \in F$. Thus $c \neq a$. 
Therefore, the six vertices $u, v, w, a, b, c$ are distinct. 
We have $d \neq u$ from $(d, v) \in P$ and $(u, v) \in F$. 
We also have $d \neq b$ from $(d, v) \in P$ and $(b, v) \in F$. 
If $d = a$ then $(b, a), (d, v) \in P$ implies $(b, v) \in P$, 
contradicting $(b, v) \in F$. Thus $d \neq a$. 
Therefore, the seven vertices $u, v, w, a, b, c, d$ are distinct. 
\par
If $(w, b) \in P$ then $(b, a), (a, u) \in P$ implies $(w, u) \in P$, 
contradicting $uw \in E(\overline{G})$. 
If $(b, w) \in P$ then $(b, w, d, v)$ is 
a forbidden configuration for $\mathbf{2+2}$, 
a contradiction. Thus $bw \in E(\overline{G})$. 
Proposition~\ref{prop:for 3+1} implies that 
the vertices $b, a, u, w$ induce $K_{1, 3}$ in $\overline{G}$. 
If $(u, w) \in F$ then $(u, w)$ is type-3 critical, 
contradicting Claim~\ref{claim:3-3}. 
Thus $(w, u) \in F$, and then 
$(b, a, u, w, d, v)$ is a skewed obstruction. 
\end{proof}

\begin{claim}\label{claim:3-5}
In Cases 1-2, 1-3, and 2-3, there is a sequence of vertices $(u, x, y, w)$ such that 
\begin{enumerate}[\bfseries --]
\item $(u, x), (y, w) \in F$ and $(x, y) \in P$; 
\item $(u, x)$ is a critical edge with the same type of $(u, v)$; 
\item $(y, w)$ is a critical edge with the same type of $(v, w)$. 
\end{enumerate}
\end{claim}
\begin{proof}[Proof of Claim~\ref{claim:3-5}]
We prove the claim for each case. 
\par
\emph{Case~1-2} 
(Fig.~\ref{figs:lemma3}\subref{figs:lemma3:case12}). 
As in Case~1-1, there are two vertices $a, b$ with 
$(v, b), (b, a) \in P$ and $(u, a) \in F$, and 
the five vertices $u, v, w, a, b$ are distinct. 
Since $(v, w)$ is type-2 critical, 
there are two other vertices $c, d$ with 
$(c, v), (w, d) \in P$ and $(c, d) \in F$. 
We have $c \neq u$ from $(c, v) \in P$ and $(u, v) \in F$. 
We also have $c \neq b$ from $(c, v), (v, b) \in P$. 
Thus we have $(c, b) \in P$ from $(c, v), (v, b) \in P$. 
We have $c \neq a$ from $(c, b), (b, a) \in P$. 
Therefore, the six vertices $u, v, w, a, b, c$ are distinct. 
We have $d \neq u$ from $(w, d) \in P$ and $uw \in E(\overline{G})$. 
We also have $d \neq b$ from $(c, d) \in F$ and $(c, b) \in P$. 
If $d = a$ then $(c, b), (b, a) \in P$ implies $(c, a) \in P$, 
contradicting $(c, d) \in F$. Thus $d \neq a$. 
Therefore, the seven vertices $u, v, w, a, b, c, d$ are distinct. 
\par
If $(u, c) \in P$ then $(c, v) \in P$ implies $(u, v) \in P$, 
contradicting $(u, v) \in F$. 
Suppose $(c, u) \in P$. 
Proposition~\ref{prop:for 2+2} implies that 
$(c, w, u, d)$ is $C_4$ in $\overline{G}$. 
We have $(u, w) \in F$ from $(c, d) \in F$, 
but $(u, w)$ is type-2 critical, 
contradicting Claim~\ref{claim:3-3}. 
Thus $uc \in E(\overline{G})$. 
Then Proposition~\ref{prop:for 3+1} implies that 
the vertices $c, b, a, u$ induce $K_{1, 3}$ in $\overline{G}$. 
We have $(u, c) \in F$ from $(u, a) \in F$, and 
$(u, c)$ is type-1 critical. 
Therefore, $(u, c, v, w)$ is a desired sequence. 
\par
\emph{Case~1-3} 
(Fig.~\ref{figs:lemma3}\subref{figs:lemma3:case13}). 
As in Case~1-1, there are two vertices $a, b$ with 
$(v, b), (b, a) \in P$ and $(u, a) \in F$, and 
the five vertices $u, v, w, a, b$ are distinct. 
Since $(v, w)$ is type-3 critical, 
there are two other vertices $c, d$ with 
$(d, c), (c, v) \in P$ and $(d, w) \in F$. 
Note that $(d, c), (c, v) \in P$ implies $(d, v) \in P$. 
We have $c \neq u$ from $(c, v) \in P$ and $(u, v) \in F$. 
We also have $c \neq b$ from $(c, v), (v, b) \in P$. 
Thus we have $(c, b) \in P$ from $(c, v), (v, b) \in P$. 
We have $c \neq a$ from $(c, b), (b, a) \in P$. 
Therefore, the six vertices $u, v, w, a, b, c$ are distinct. 
We have $d \neq u$ from $(d, v) \in P$ and $(u, v) \in F$. 
We also have $d \neq b$ from $(d, c), (c, b) \in P$. 
Moreover, we have $d \neq a$ from $(d, c), (c, b), (b, a) \in P$. 
Therefore, the seven vertices $u, v, w, a, b, c, d$ are distinct. 
\par
If $(u, c) \in P$ then $(c, v) \in P$ implies $(u, v) \in P$, 
contradicting $(u, v) \in F$. 
Suppose $(c, u) \in P$. 
Proposition~\ref{prop:for 3+1} implies that 
the vertices $d, c, u, w$ induce $K_{1, 3}$ in $\overline{G}$. 
We have $(u, w) \in F$ from $(d, w) \in F$, 
but $(u, w)$ is type-3 critical, 
contradicting Claim~\ref{claim:3-3}. 
Thus $uc \in E(\overline{G})$. 
Then Proposition~\ref{prop:for 3+1} implies that 
the vertices $c, b, a, u$ induce $K_{1, 3}$ in $\overline{G}$. 
We have $(u, c) \in F$ from $(u, a) \in F$, and 
$(u, c)$ is type-1 critical. 
Therefore, $(u, c, v, w)$ is a desired sequence. 
\par
\emph{Case~2-3} 
(Fig.~\ref{figs:lemma3}\subref{figs:lemma3:case23}). 
As in Case~2-1, there are two vertices $a, b$ with 
$(a, u), (v, b) \in P$ and $(a, b) \in F$, and 
the five vertices $u, v, w, a, b$ are distinct. 
Since $(v, w)$ is type-3 critical, 
there are two other vertices $c, d$ with 
$(d, c), (c, v) \in P$ and $(d, w) \in F$. 
Note that $(d, c), (c, v) \in P$ implies $(d, v) \in P$. 
We have $c \neq u$ from $(c, v) \in P$ and $(u, v) \in F$. 
We also have $c \neq b$ from $(c, v), (v, b) \in P$. 
Thus we have $(c, b) \in P$ from $(c, v), (v, b) \in P$. 
We have $c \neq a$ from $(c, b) \in P$ and $(a, b) \in F$. 
Therefore, the six vertices $u, v, w, a, b, c$ are distinct. 
We have $d \neq u$ from $(d, v) \in P$ and $(u, v) \in F$. 
We also have $d \neq b$ from $(d, c), (c, b) \in P$. 
If $d = a$ then $(d, c), (c, b) \in P$ implies $(d, b) \in P$, 
contradicting $(a, b) \in F$. Thus $d \neq a$. 
Therefore, the seven vertices $u, v, w, a, b, c, d$ are distinct. 
\par
If $(u, c) \in P$ then $(c, v) \in P$ implies $(u, v) \in P$, 
contradicting $(u, v) \in F$. 
Suppose $(c, u) \in P$. 
We have from Proposition~\ref{prop:for 3+1} that 
the vertices $d, c, u, w$ induce $K_{1, 3}$ in $\overline{G}$. 
We have $(u, w) \in F$ from $(d, w) \in F$, 
but $(u, w)$ is type-3 critical, 
contradicting Claim~\ref{claim:3-3}. 
Thus $uc \in E(\overline{G})$. 
Then Proposition~\ref{prop:for 2+2} implies that 
$(a, c, u, b)$ is $C_4$ in $\overline{G}$. 
We have $(u, c) \in F$ from $(a, b) \in F$, and 
$(u, c)$ is type-2 critical. 
Therefore, $(u, c, v, w)$ is a desired sequence. 
\end{proof}

By Claim~\ref{claim:3-4}, we can assume that 
any two critical edges consecutive on $C$ is in Cases~1-2, 1-3, and 2-3. 
However, we can obtain an alternating cycle from $C$ 
by repeated applications of Claim~\ref{claim:3-5}, a contradiction. 
Thus the lemma holds. 
\end{proof}

\begin{Lemma}\label{lemma:recognition:4}
Let $F$ be an orientation of $E_o$ such that 
$P + F$ contains no forbidden configurations for 
$\mathbf{2+2}$, $\mathbf{3+1}$, or $\mathbf{2+1}$. 
We can compute in $O(m \bar{m}^2)$ time 
an orientation $F'$ of $E_o$ such that 
$P + F'$ contains no forbidden configurations for 
$\mathbf{2+2}$, $\mathbf{3+1}$, or $\mathbf{2+1}$, and 
$P + F'$ contains no obstructions. 
\end{Lemma}
\begin{proof}
Let $(v_0, v_1, v_2, v_3, v_4, v_5)$ be a regular obstruction. 
We call the edge $(v_0, v_1) \in P$ the \emph{base} of the obstruction 
and the edges $(v_4, v_3), (v_5, v_3), (v_4, v_2), (v_5, v_2) \in F$ 
its \emph{fronts}. 
Similarly, 
let $(v_0, v_1, v_2, v_3, v_4, v_5)$ be a skewed obstruction. 
We call the edge $(v_0, v_1) \in P$ the \emph{base} of the obstruction 
and the edges $(v_4, v_3), (v_5, v_3) \in F$ 
its \emph{fronts}. 
We can observe the following from Propositions~\ref{prop:for regular obstruction} 
and~\ref{prop:for skewed obstruction} (Fig.~\ref{figs:obstructions}). 
\begin{claim}\label{claim:for base and fronts}
If $(u, v) \in P$ is the base of an obstruction and $(x, y) \in F$ is its front, 
then $(u, x), (v, x), (y, u), (y, v) \in F$. 
\end{claim}

Let $F$ be an orientation of $E_o$ such that 
$P + F$ contains no forbidden configurations for 
$\mathbf{2+2}$, $\mathbf{3+1}$, or $\mathbf{2+1}$. 
Then, let $(u, v)$ be an edge in $P$. 
We define that 
\begin{align*}
A_{(u, v)} & = \{(x, y) \in F \colon\ 
\text{there is a regular obstruction such that $(u, v)$ is its base and $(x, y)$ is its front}\}, 
\\
A_{(u, v)}^{-1} & = \{(y, x) \colon\ (x, y) \in A_{(u, v)} \}, \text{ and }
\\
F' & = (F - A_{(u, v)}) + A_{(u, v)}^{-1}.
\end{align*}
In other words, 
$F'$ is an orientation of $E_o$ obtained from $F$ 
by reversing the direction of all edges in $A_{(u, v)}$. 

\begin{figure*}[t]
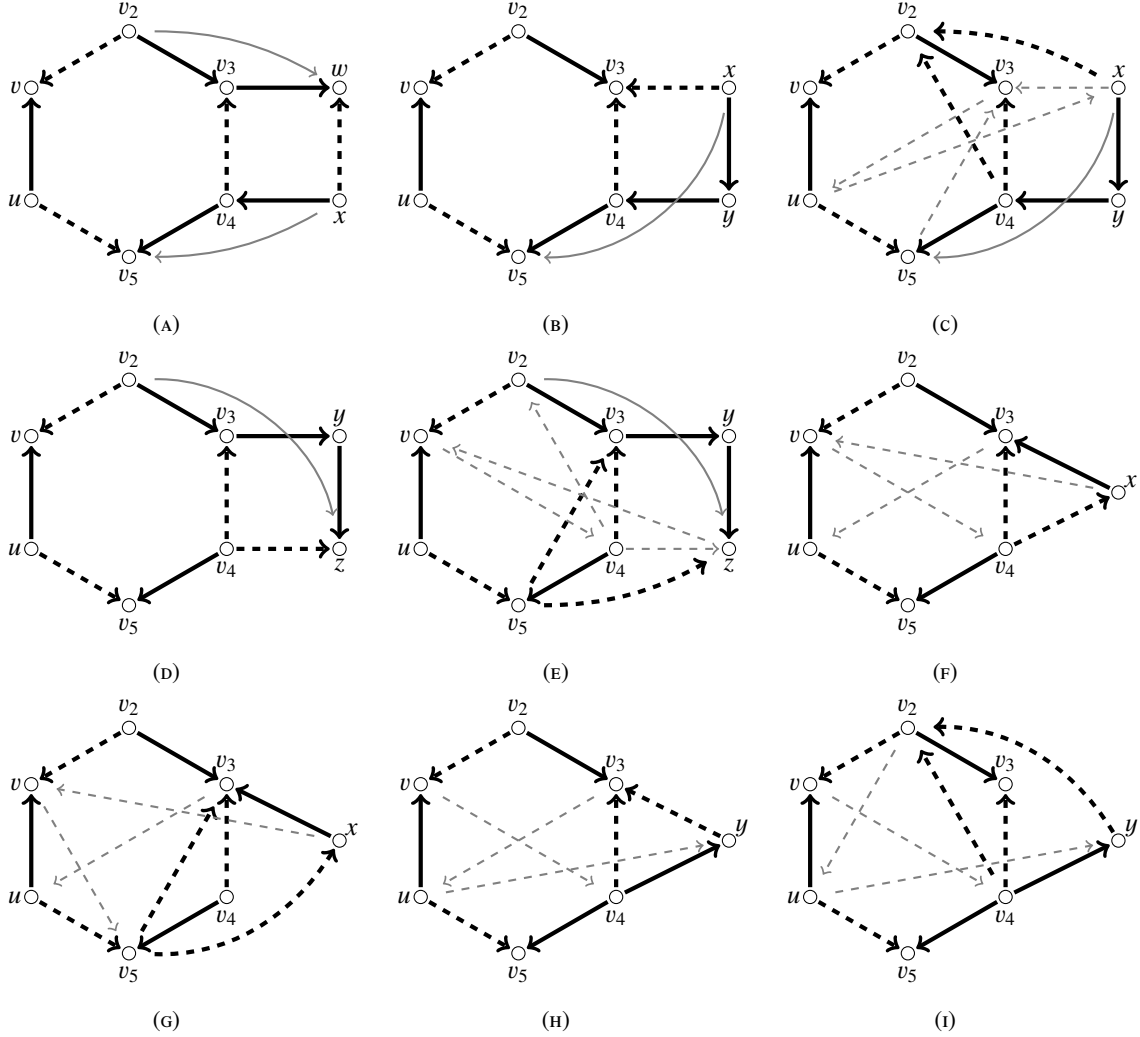

  \centering
  \subcaptionbox{\label{figs:Lemma4:A1}}{\input{figs/lemma4/A1}}
  \subcaptionbox{\label{figs:Lemma4:A2}}{\input{figs/lemma4/A2}}
  \subcaptionbox{\label{figs:Lemma4:A2s}}{\input{figs/lemma4/A2s}}
\\
  \subcaptionbox{\label{figs:Lemma4:A3}}{\input{figs/lemma4/A3}}
  \subcaptionbox{\label{figs:Lemma4:A3s}}{\input{figs/lemma4/A3s}}
  \subcaptionbox{\label{figs:Lemma4:A4}}{\input{figs/lemma4/A4}}
\\
  \subcaptionbox{\label{figs:Lemma4:A4s}}{\input{figs/lemma4/A4s}}
  \subcaptionbox{\label{figs:Lemma4:A5}}{\input{figs/lemma4/A5}}
  \subcaptionbox{\label{figs:Lemma4:A5s}}{\input{figs/lemma4/A5s}}
  \caption{
    Illustrating the proof of Claims~\ref{claim:Lemma4:A:2+2}--\ref{claim:Lemma4:A:2+1}. 
    Arrows and dashed arrows denote the same type of edges 
    as in Fig.~\ref{figs:alternating-cycles}. 
  }
  \label{figs:Lemma4:Af}
\end{figure*}

\begin{claim}\label{claim:Lemma4:A:2+2}
The orientation $P + F'$ contains no forbidden configurations for 
$\mathbf{2+2}$. 
\end{claim}
\begin{proof}[Proof of Claim~\ref{claim:Lemma4:A:2+2}]
Suppose that $P + F'$ contains a forbidden configuration $(x, y, z, w)$ 
for $\mathbf{2+2}$. 
If $(z, y), (x, w) \in A_{(u, v)}^{-1}$ then $(y, z), (w, x) \in F$, 
that is, $(x, y, z, w)$ is an alternating cycle of length 4 in $P + F$, 
contradicting Proposition~\ref{prop:alternating 4-cycle}. 
Thus either $(z, y) \in A_{(u, v)}^{-1}$ or $(x, w) \in A_{(u, v)}^{-1}$. 
Assume without loss of generality 
$(z, y) \in A_{(u, v)}^{-1}$ and $(x, w) \notin A_{(u, v)}^{-1}$. 
Then $(y, z) \in A_{(u, v)} \subseteq F$ and 
$(x, w) \in F$ but $(x, w) \notin A_{(u, v)}$. 
\par
Since $(y, z) \in A_{(u, v)}$, there is a regular obstruction 
$(u, v, v_2, v_3, v_4, v_5)$ of $P + F$ such that 
$(u, v)$ is its base and $(y, z)$ is its front. 
Suppose $(y, z) = (v_4, v_3)$ 
(Fig.~\ref{figs:Lemma4:Af}\subref{figs:Lemma4:A1}). 
We have $(v_2, w) \in P$ from $(v_2, v_3), (v_3, w) \in P$. 
We also have $(x, v_5) \in P$ from $(x, v_4), (v_4, v_5) \in P$. 
Then $(u, v, v_2, w, x, v_5)$ is a regular obstruction, and 
hence $(x, w) \in A_{(u, v)}$, a contradiction. 
Suppose $(y, z) = (v_5, v_3)$. 
We have $(v_2, w) \in P$ from $(v_2, v_3), (v_3, w) \in P$. 
Then $(u, v, v_2, w, x, v_5)$ is a regular obstruction, and 
hence $(x, w) \in A_{(u, v)}$, a contradiction. 
Suppose $(y, z) = (v_4, v_2)$. 
We have $(x, v_5) \in P$ from $(x, v_4), (v_4, v_5) \in P$. 
Then $(u, v, v_2, w, x, v_5)$ is a regular obstruction, and 
hence $(x, w) \in A_{(u, v)}$, a contradiction. 
Suppose $(y, z) = (v_5, v_3)$. 
Then $(u, v, v_2, w, x, v_5)$ is a regular obstruction, and 
hence $(x, w) \in A_{(u, v)}$, a contradiction. 
\end{proof}

\begin{claim}\label{claim:Lemma4:A:3+1}
The orientation $P + F'$ contains no forbidden configurations for 
$\mathbf{3+1}$. 
\end{claim}
\begin{proof}[Proof of Claim~\ref{claim:Lemma4:A:3+1}]
Suppose that $P + F'$ contains a forbidden configuration $(x, y, z, w)$ 
for $\mathbf{3+1}$. Note that $(x, y), (y, z) \in P$ implies $(x, z) \in P$. 
If $(w, z), (x, w) \in A_{(u, v)}^{-1}$ then $(z, w), (w, x) \in F$, 
that is, $(x, z, w)$ is a forbidden configuration for $\mathbf{2+1}$ in $P + F$, 
a contradiction. 
Thus either $(w, z) \in A_{(u, v)}^{-1}$ or $(x, w) \in A_{(u, v)}^{-1}$. 
\par
We first suppose 
$(w, z) \in A_{(u, v)}^{-1}$ and $(x, w) \notin A_{(u, v)}^{-1}$. 
We have $(z, w) \in A_{(u, v)} \subseteq F$ and 
$(x, w) \in F$ but $(x, w) \notin A_{(u, v)}$. 
Since $(z, w) \in A_{(u, v)}$, there is a regular obstruction 
$(u, v, v_2, v_3, v_4, v_5)$ of $P + F$ such that 
$(u, v)$ is its base and $(z, w)$ is its front. 
Suppose $(z, w) = (v_4, v_3)$ 
(Fig.~\ref{figs:Lemma4:Af}\subref{figs:Lemma4:A2}). 
We have $(x, v_5) \in P$ from $(x, y), (y, v_4), (v_4, v_5) \in P$. 
Then $(u, v, v_2, v_3, x, v_5)$ is a regular obstruction, and 
hence $(x, v_3) \in A_{(u, v)}$, a contradiction. 
Suppose $(z, w) = (v_5, v_3)$. 
We have $(x, v_5) \in P$ from $(x, y), (y, v_5) \in P$. 
Then $(u, v, v_2, v_3, x, v_5)$ is a regular obstruction, and 
hence $(x, v_3) \in A_{(u, v)}$, a contradiction. 
\par
Suppose $(z, w) = (v_4, v_2)$ 
(Fig.~\ref{figs:Lemma4:Af}\subref{figs:Lemma4:A2s}). 
We have $(x, v_5) \in P$ from $(x, y), (y, v_4), (v_4, v_5) \in P$. 
We also have $(v_5, v_3), (v_3, u) \in F$ from Proposition~\ref{prop:for regular obstruction}. 
If $(u, x) \in P$ then $(x, v_5) \in P$ implies $(u, v_5) \in P$, 
contradicting $(u, v_5) \in F$. 
If $(x, u) \in P$ then $(x, u, v, v_2)$ is 
a forbidden configuration for $\mathbf{3+1}$ in $P + F$, a contradiction. 
Thus $ux \in E(\overline{G})$. 
Proposition~\ref{prop:for 3+1} implies that 
the vertices $x, y, v_5, u$ induce $K_{1, 3}$ in $\overline{G}$. 
Then we have $(u, x) \in F$ from $(u, v_5) \in F$. 
If $(x, v_3) \in P$ then $(x, v_3, u)$ is a forbidden configuration 
for $\mathbf{2+1}$ in $P + F$, a contradiction. 
If $(v_3, x) \in P$ then $(v_2, v_3) \in P$ implies $(v_2, x) \in P$, 
contradicting $(x, v_2) \in F$. 
Thus $v_3x \in E(\overline{G})$. 
Proposition~\ref{prop:for 3+1} implies that 
the vertices $x, y, v_5, v_3$ induce $K_{1, 3}$ in $\overline{G}$, 
and we have $(x, v_3) \in F$ from $(v_5, v_3) \in F$. 
Then $(u, v, v_2, v_3, x, v_5)$ is a regular obstruction, and 
hence $(x, v_2) \in A_{(u, v)}$, a contradiction.
By similar arguments, we have a contradiction 
when we suppose $(z, w) = (v_5, v_2)$. 
\par
We next suppose $(x, w) \in A_{(u, v)}^{-1}$ and $(w, z) \notin A_{(u, v)}^{-1}$. 
We have $(w, x) \in A_{(u, v)} \subseteq F$ and 
$(w, z) \in F$ but $(w, z) \notin A_{(u, v)}$. 
Since $(w, x) \in A_{(u, v)}$, there is a regular obstruction 
$(u, v, v_2, v_3, v_4, v_5)$ of $P + F$ such that 
$(u, v)$ is its base and $(w, x)$ is its front. 
Suppose $(w, x) = (v_4, v_3)$ 
(Fig.~\ref{figs:Lemma4:Af}\subref{figs:Lemma4:A3}). 
We have $(v_2, z) \in P$ from $(v_2, v_3), (v_3, y), (y, z) \in P$. 
Then $(u, v, v_2, z, v_4, v_5)$ is a regular obstruction, and 
hence $(v_4, z) \in A_{(u, v)}$, a contradiction. 
Suppose $(w, x) = (v_4, v_2)$. 
We have $(v_2, z) \in P$ from $(v_2, y), (y, z) \in P$. 
Then $(u, v, v_2, z, v_4, v_5)$ is a regular obstruction, and 
hence $(v_4, z) \in A_{(u, v)}$, a contradiction. 
\par
Suppose $(w, x) = (v_5, v_3)$ 
(Fig.~\ref{figs:Lemma4:Af}\subref{figs:Lemma4:A3s}). 
We have $(v_2, z) \in P$ from $(v_2, v_3), (v_3, y), (y, z) \in P$. 
We also have $(v, v_4), (v_4, v_2) \in F$ from Proposition~\ref{prop:for regular obstruction}. 
If $(z, v) \in P$ then $(v_2, z) \in P$ implies $(v_2, v) \in P$, 
contradicting $(v_2, v) \in F$. 
If $(v, z) \in P$ then $(u, v, z, v_5)$ is 
a forbidden configuration for $\mathbf{3+1}$ in $P + F$, a contradiction. 
Thus $vz \in E(\overline{G})$. 
Proposition~\ref{prop:for 3+1} implies that 
the vertices $v_2, y, z, v$ induce $K_{1, 3}$ in $\overline{G}$. 
Then we have $(z, v) \in F$ from $(v_2, v) \in F$. 
If $(v_4, z) \in P$ then $(v_4, z, v)$ is a forbidden configuration 
for $\mathbf{2+1}$ in $P + F$, a contradiction. 
If $(z, v_4) \in P$ then $(v_4, v_5) \in P$ implies $(z, v_5) \in P$, 
contradicting $(v_5, z) \in F$. 
Thus $v_4z \in E(\overline{G})$. 
Proposition~\ref{prop:for 3+1} implies that 
the vertices $v_2, y, z, v_4$ induce $K_{1, 3}$ in $\overline{G}$, 
and we have $(v_4, z) \in F$ from $(v_4, v_2) \in F$. 
Then $(u, v, v_2, z, v_4, v_5)$ is a regular obstruction, and 
hence $(v_5, z) \in A_{(u, v)}$, a contradiction.
By similar arguments, we have a contradiction 
when we suppose $(z, w) = (v_5, v_2)$. 
\end{proof}

\begin{claim}\label{claim:Lemma4:A:2+1}
The orientation $P + F'$ contains no forbidden configurations for 
$\mathbf{2+1}$. 
\end{claim}
\begin{proof}[Proof of Claim~\ref{claim:Lemma4:A:2+1}]
Suppose that $P + F'$ contains a forbidden configuration $(x, y, z)$ 
for $\mathbf{2+1}$. 
If $(y, z), (z, x) \in A_{(u, v)}^{-1}$ then 
we have from Claim~\ref{claim:for base and fronts} that 
$(y, z) \in A_{(u, v)}^{-1}$ implies $(u, z) \in F$ and 
$(z, x) \in A_{(u, v)}^{-1}$ implies $(z, u) \in F$, 
a contradiction. 
Thus either $(y, z) \in A_{(u, v)}^{-1}$ or $(z, x) \in A_{(u, v)}^{-1}$. 
\par
We first suppose $(y, z) \in A_{(u, v)}^{-1}$ and $(z, x) \notin A_{(u, v)}^{-1}$. 
We have $(z, y) \in A_{(u, v)} \subseteq F$ and 
$(z, x) \in F$ but $(z, x) \notin A_{(u, v)}$. 
Since $(z, y) \in A_{(u, v)}$, there is a regular obstruction 
$(u, v, v_2, v_3, v_4, v_5)$ of $P + F$ such that 
$(u, v)$ is its base and $(z, y)$ is its front. 
Suppose $(z, y) = (v_4, v_3)$ 
(Fig.~\ref{figs:Lemma4:Af}\subref{figs:Lemma4:A4}). 
We have $(v_3, u), (v, v_4) \in F$ 
from Proposition~\ref{prop:for regular obstruction}. 
If $(v, x) \in P$ then $(u, v), (x, v_3) \in P$ implies $(u, v_3) \in P$, 
contradicting $(v_3, u) \in F$. 
If $(x, v) \in P$ then $(x, v, v_4)$ is 
a forbidden configuration for $\mathbf{2+1}$ in $P + F$, a contradiction. 
Thus $vx \in E(\overline{G})$. 
Proposition~\ref{prop:for 2+2} implies that 
$(u, x, v, v_3)$ is $C_4$ in $\overline{G}$, and 
we have $(x, v) \in F$ from $(v_3, u) \in F$. 
Then $(u, v, x, v_3, v_4, v_5)$ is a regular obstruction, and 
hence $(v_4, x) \in A_{(u, v)}$, a contradiction.
When we suppose $(z, w) = (v_4, v_2)$, 
we have $(x, v_3) \in P$ from $(x, v_2), (v_2, v_3) \in P$. 
Thus by similar arguments, we have a contradiction. 
\par
Suppose $(z, y) = (v_5, v_3)$ 
(Fig.~\ref{figs:Lemma4:Af}\subref{figs:Lemma4:A4s}). 
We have $(v_3, u), (v, v_5) \in F$ 
from Proposition~\ref{prop:for regular obstruction}. 
If $(v, x) \in P$ then $(u, v), (x, v_3) \in P$ implies $(u, v_3) \in P$, 
contradicting $(v_3, u) \in F$. 
If $(x, v) \in P$ then $(x, v, v_5)$ is 
a forbidden configuration for $\mathbf{2+1}$ in $P + F$, a contradiction. 
Thus $vx \in E(\overline{G})$. 
Proposition~\ref{prop:for 2+2} implies that 
$(u, x, v, v_3)$ is $C_4$ in $\overline{G}$, and 
we have $(x, v) \in F$ from $(v_3, u) \in F$. 
Then $(u, v, x, v_3, v_4, v_5)$ is a regular obstruction, and 
hence $(v_5, x) \in A_{(u, v)}$, a contradiction. 
When we suppose $(z, y) = (v_5, v_2)$, 
we have $(x, v_3) \in P$ from $(x, v_2), (v_2, v_3) \in P$. 
Thus by similar arguments, we have a contradiction. 
\par
We next suppose $(z, x) \in A_{(u, v)}^{-1}$ and $(y, z) \notin A_{(u, v)}^{-1}$. 
We have $(x, z) \in A_{(u, v)} \subseteq F$ and 
$(y, z) \in F$ but $(y, z) \notin A_{(u, v)}$. 
Since $(x, z) \in A_{(u, v)}$, there is a regular obstruction 
$(u, v, v_2, v_3, v_4, v_5)$ of $P + F$ such that 
$(u, v)$ is its base and $(x, z)$ is its front. 
Suppose $(x, z) = (v_4, v_3)$ 
(Fig.~\ref{figs:Lemma4:Af}\subref{figs:Lemma4:A5}). 
We have $(v_3, u), (v, v_4) \in F$ 
from Proposition~\ref{prop:for regular obstruction}. 
If $(y, u) \in P$ then $(v_4, y), (u, v) \in P$ implies $(v_4, v) \in P$, 
contradicting $(v, v_4) \in F$. 
If $(u, y) \in P$ then $(u, y, v_3)$ is 
a forbidden configuration for $\mathbf{2+1}$ in $P + F$, a contradiction. 
Thus $uy \in E(\overline{G})$. 
Proposition~\ref{prop:for 2+2} implies that 
$(u, v_4, v, y)$ is $C_4$ in $\overline{G}$, and 
we have $(u, y) \in F$ from $(v, v_4) \in F$. 
Then $(u, v, v_2, v_3, v_4, y)$ is a regular obstruction, and 
hence $(y, v_2) \in A_{(u, v)}$, a contradiction. 
When we suppose $(x, z) = (v_5, v_3)$, 
we have $(v_4, y) \in P$ from $(v_4, v_5), (v_5, y) \in P$. 
Thus by similar arguments, we have a contradiction. 
\par
Suppose $(x, z) = (v_4, v_2)$ 
(Fig.~\ref{figs:Lemma4:Af}\subref{figs:Lemma4:A5s}). 
We have $(v_2, u), (v, v_4) \in F$ 
from Proposition~\ref{prop:for regular obstruction}. 
If $(y, u) \in P$ then $(v_4, y), (u, v) \in P$ implies $(v_4, v) \in P$, 
contradicting $(v, v_4) \in F$. 
If $(u, y) \in P$ then $(u, y, v_2)$ is 
a forbidden configuration for $\mathbf{2+1}$ in $P + F$, a contradiction. 
Thus $uy \in E(\overline{G})$. 
Proposition~\ref{prop:for 2+2} implies that 
$(u, v_4, v, y)$ is $C_4$ in $\overline{G}$, and 
we have $(u, y) \in F$ from $(v, v_4) \in F$. 
Then $(u, v, v_2, v_3, v_4, y)$ is a regular obstruction, and 
hence $(y, v_2) \in A_{(u, v)}$, a contradiction. 
When we suppose $(x, z) = (v_5, v_2)$, 
we have $(v_4, y) \in P$ from $(v_4, v_5), (v_5, y) \in P$. 
Thus by similar arguments, we have a contradiction. 
\end{proof}

\begin{figure*}[t]
  \centering
  \input{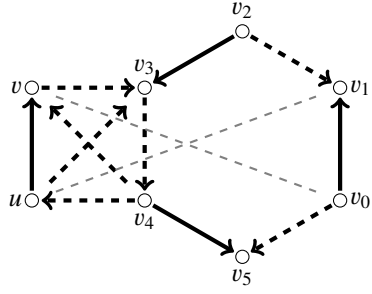}
  \caption{
    Illustrating the proof of Claim~\ref{claim:Lemma4:A:obstruction}. 
    Arrows, dashed arrows, and dashed lines denote the same type of edges 
    as in Fig.~\ref{figs:alternating-cycles}. 
  }
  \label{figs:Lemma4:As}
\end{figure*}

\begin{claim}\label{claim:Lemma4:A:obstruction}
No edge in $A_{(u, v)}^{-1}$ is an edge of a regular obstruction of $P + F'$. 
That is, reversing the direction of the edges in $A_{(u, v)}$ generates no new regular obstructions. 
\end{claim}
\begin{proof}[Proof of Claim~\ref{claim:Lemma4:A:obstruction}]
Suppose that $P + F'$ contains a regular obstruction 
$(v_0, v_1, v_2, v_3, v_4, v_5)$. 
Assume without loss of generality $(v_4, v_3) \in A_{(u, v)}^{-1}$. 
Suppose $(v_0, v_5) \in A_{(u, v)}^{-1}$. 
We have from Claim~\ref{claim:for base and fronts} that 
$(v_4, v_3) \in A_{(u, v)}^{-1}$ implies $(v_4, v) \in F$ and 
$(v_0, v_5) \in A_{(u, v)}^{-1}$ implies $(u, v_5) \in F$. 
Thus $(u, v, v_4, v_5)$ is 
a forbidden configuration for $\mathbf{2+2}$ in $P + F$, a contradiction. 
Thus $(v_0, v_5) \notin A_{(u, v)}^{-1}$. 
Similarly, we have $(v_2, v_1) \notin A_{(u, v)}^{-1}$. 
\par
Now, $(v_3, v_4) \in A_{(u, v)} \subseteq F$ and 
$(v_0, v_5), (v_2, v_1) \in F$ but $(v_0, v_5), (v_2, v_1) \notin A_{(u, v)}$. 
Claim~\ref{claim:for base and fronts} implies 
$(u, v_3), (v, v_3), (v_4, u), (v_4, v) \in F$ (Fig.~\ref{figs:Lemma4:As}). 
If $(v_0, v) \in P$ then $(v_0, v, v_4, v_5)$ is 
a forbidden configuration for $\mathbf{2+2}$ in $P + F$, a contradiction. 
If $(v, v_0) \in P$ then $(v_0, v_1) \in P$ implies $(v, v_1) \in P$, 
but then $(v, v_1, v_2, v_3)$ is 
a forbidden configuration for $\mathbf{2+2}$ in $P + F$, a contradiction. 
Thus $vv_0 \in E(\overline{G})$. 
If $(v_1, u) \in P$ then $(v_0, v_1), (u, v) \in P$ implies $(v_0, v) \in P$, 
contradicting $vv_0 \in E(\overline{G})$. 
If $(u, v_1) \in P$ then $(u, v_1, v_2, v_3)$ is 
a forbidden configuration for $\mathbf{2+2}$ in $P + F$, a contradiction. 
Thus $uv_1 \in E(\overline{G})$. 
Proposition~\ref{prop:for 2+2} implies that 
$(u, v_0, v, v_1)$ is $C_4$ in $\overline{G}$. 
If $(v, v_0), (u, v_1) \in F$ then 
$(u, v, v_4, v_5, v_0, v_1)$ is a regular obstruction, and 
hence $(v_0, v_5) \in A_{(u, v)}$, a contradiction. 
If $(v_0, v), (v_1, u) \in F$ then 
$(u, v, v_0, v_1, v_2, v_3)$ is a regular obstruction, and 
hence $(v_2, v_1) \in A_{(u, v)}$, a contradiction. 
\end{proof}

The following is trivial from the definition of $A_{(u, v)}$. 
\begin{claim}\label{claim:Lemma4:A:remove}
The orientation $P + F'$ contains no regular obstructions with base $(u, v)$. 
\end{claim}

Claims~\ref{claim:Lemma4:A:2+2}--\ref{claim:Lemma4:A:remove} ensure that 
by repeating the procedure for each edge of $P$, 
we can obtain an orientation $F''$ of $E_o$ such that 
$P + F''$ contains no forbidden configurations for 
$\mathbf{2+2}$, $\mathbf{3+1}$, or $\mathbf{2+1}$, and 
$P + F''$ contains no regular obstructions.

Now, let $F$ be an orientation of $E_o$ such that 
$P + F$ contains no forbidden configurations for 
$\mathbf{2+2}$, $\mathbf{3+1}$, or $\mathbf{2+1}$, 
and $P + F$ contains no regular obstructions. 
Then, let $(u, v)$ be an edge in $P$. 
We define that 
\begin{align*}
B_{(u, v)} & = \{(x, y) \in F \colon\ 
\text{there is a skewed obstruction such that $(u, v)$ is its base and $(x, y)$ is its front}\}, 
\\
B_{(u, v)}^{-1} & = \{(y, x) \colon\ (x, y) \in B_{(u, v)} \}, \text{ and }
\\
F' & = (F - B_{(u, v)}) + B_{(u, v)}^{-1}.
\end{align*}
In other words, 
$F'$ is an orientation of $E_o$ obtained from $F$ 
by reversing the direction of all edges in $B_{(u, v)}$. 

\begin{figure*}[t]
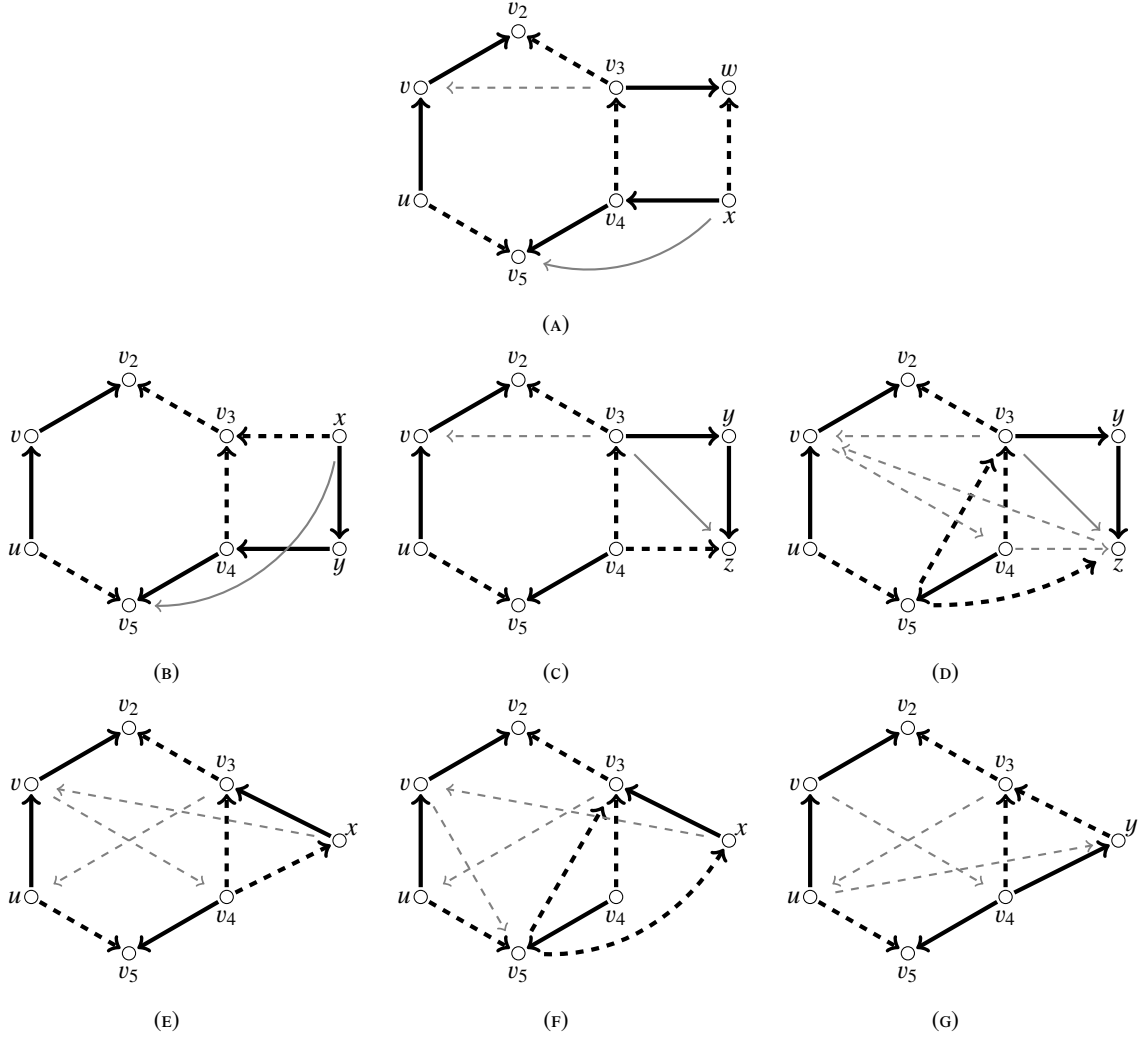

  \centering
  \subcaptionbox{\label{figs:Lemma4:B1}}{\input{figs/lemma4/B1}}
\\
  \subcaptionbox{\label{figs:Lemma4:B2}}{\input{figs/lemma4/B2}}
  \subcaptionbox{\label{figs:Lemma4:B3}}{\input{figs/lemma4/B3}}
  \subcaptionbox{\label{figs:Lemma4:B3s}}{\input{figs/lemma4/B3s}}
\\
  \subcaptionbox{\label{figs:Lemma4:B4}}{\input{figs/lemma4/B4}}
  \subcaptionbox{\label{figs:Lemma4:B4s}}{\input{figs/lemma4/B4s}}
  \subcaptionbox{\label{figs:Lemma4:B5}}{\input{figs/lemma4/B5}}
  \caption{
    Illustrating the proof of Claims~\ref{claim:Lemma4:B:2+2}--\ref{claim:Lemma4:B:2+1}. 
    Arrows and dashed arrows denote the same type of edges 
    as in Fig.~\ref{figs:alternating-cycles}. 
  }
  \label{figs:Lemma4:Bf}
\end{figure*}

\begin{claim}\label{claim:Lemma4:B:2+2}
The orientation $P + F'$ contains no forbidden configurations for 
$\mathbf{2+2}$. 
\end{claim}
\begin{proof}[Proof of Claim~\ref{claim:Lemma4:B:2+2}]
Suppose that $P + F'$ contains a forbidden configuration $(x, y, z, w)$ 
for $\mathbf{2+2}$. 
If $(z, y), (x, w) \in B_{(u, v)}^{-1}$ then $(y, z), (w, x) \in F$, 
that is, $(x, y, z, w)$ is an alternating cycle of length 4 in $P + F$, 
contradicting Proposition~\ref{prop:alternating 4-cycle}. 
Thus either $(z, y) \in B_{(u, v)}^{-1}$ or $(x, w) \in B_{(u, v)}^{-1}$. 
Assume without loss of generality 
$(z, y) \in B_{(u, v)}^{-1}$ and $(x, w) \notin B_{(u, v)}^{-1}$. 
Then $(y, z) \in B_{(u, v)} \subseteq F$ and 
$(x, w) \in F$ but $(x, w) \notin B_{(u, v)}$. 
\par
Since $(y, z) \in B_{(u, v)}$, there is a skewed obstruction 
$(u, v, v_2, v_3, v_4, v_5)$ of $P + F$ such that 
$(u, v)$ is its base and $(y, z)$ is its front. 
Suppose $(y, z) = (v_4, v_3)$ 
(Fig.~\ref{figs:Lemma4:Bf}\subref{figs:Lemma4:B1}). 
We have $(x, v_5) \in P$ from $(x, v_4), (v_4, v_5) \in P$. 
We also have $(v_3, v) \in F$ 
from Proposition~\ref{prop:for skewed obstruction}. 
Then $(u, v, v_3, w, x, v_5)$ is a regular obstruction of $P + F$, 
a contradiction. 
By similar arguments, we have a contradiction 
when we suppose $(y, z) = (v_5, v_3)$. 
\end{proof}

\begin{claim}\label{claim:Lemma4:B:3+1}
The orientation $P + F'$ contains no forbidden configurations for 
$\mathbf{3+1}$. 
\end{claim}
\begin{proof}[Proof of Claim~\ref{claim:Lemma4:B:3+1}]
Suppose that $P + F'$ contains a forbidden configuration $(x, y, z, w)$ 
for $\mathbf{3+1}$. Note that $(x, y), (y, z) \in P$ implies $(x, z) \in P$. 
If $(w, z), (x, w) \in B_{(u, v)}^{-1}$ then $(z, w), (w, x) \in F$, 
that is, $(x, z, w)$ is a forbidden configuration for $\mathbf{2+1}$ in $P + F$, 
a contradiction. 
Thus either $(w, z) \in B_{(u, v)}^{-1}$ or $(x, w) \in B_{(u, v)}^{-1}$. 
\par
We first suppose $(w, z) \in B_{(u, v)}^{-1}$ and $(x, w) \notin B_{(u, v)}^{-1}$. 
We have $(z, w) \in B_{(u, v)} \subseteq F$ and 
$(x, w) \in F$ but $(x, w) \notin B_{(u, v)}$. 
Since $(z, w) \in B_{(u, v)}$, there is a skewed obstruction 
$(u, v, v_2, v_3, v_4, v_5)$ of $P + F$ such that 
$(u, v)$ is its base and $(z, w)$ is its front. 
Suppose $(z, w) = (v_4, v_3)$ 
(Fig.~\ref{figs:Lemma4:Bf}\subref{figs:Lemma4:B2}). 
We have $(x, v_5) \in P$ from $(x, y), (y, v_4), (v_4, v_5) \in P$. 
Then $(u, v, v_2, v_3, x, v_5)$ is a skewed obstruction, and 
hence $(x, v_3) \in B_{(u, v)}$, a contradiction. 
By similar arguments, we have a contradiction 
when we suppose $(y, z) = (v_5, v_3)$. 
\par
We next suppose $(x, w) \in B_{(u, v)}^{-1}$ and $(w, z) \notin B_{(u, v)}^{-1}$. 
We have $(w, x) \in B_{(u, v)} \subseteq F$ and 
$(w, z) \in F$ but $(w, z) \notin B_{(u, v)}$. 
Since $(w, x) \in B_{(u, v)}$, there is a skewed obstruction 
$(u, v, v_2, v_3, v_4, v_5)$ of $P + F$ such that 
$(u, v)$ is its base and $(w, x)$ is its front. 
Suppose $(w, x) = (v_4, v_3)$ 
(Fig.~\ref{figs:Lemma4:Bf}\subref{figs:Lemma4:B3}). 
We have $(v_3, v) \in F$ 
from Proposition~\ref{prop:for skewed obstruction}. 
Then $(u, v, v_3, z, v_4, v_5)$ is a regular obstruction of $P + F$, 
a contradiction. 
\par
Suppose $(w, x) = (v_5, v_3)$ 
(Fig.~\ref{figs:Lemma4:Bf}\subref{figs:Lemma4:B3s}). 
We have $(v_3, v), (v, v_4) \in F$ 
from Proposition~\ref{prop:for skewed obstruction}. 
If $(z, v) \in P$ then $(v_3, z) \in P$ implies $(v_3, v) \in P$, 
contradicting $(v_3, v) \in F$. 
If $(v, z) \in P$ then $(u, v, z, v_5)$ is 
a forbidden configuration for $\mathbf{3+1}$ in $P + F$, a contradiction. 
Thus $vz \in E(\overline{G})$. 
Proposition~\ref{prop:for 2+2} implies that 
$(v, v_3, v_2, z)$ is $C_4$ in $\overline{G}$. 
Then we have $(z, v) \in F$ from $(v_3, v_2) \in F$. 
If $(v_4, z) \in P$ then $(v_4, z, v)$ is a forbidden configuration 
for $\mathbf{2+1}$ in $P + F$, a contradiction. 
If $(z, v_4) \in P$ then $(v_4, v_5) \in P$ implies $(z, v_5) \in P$, 
contradicting $(v_5, z) \in F$. 
Thus $v_4z \in E(\overline{G})$. 
Proposition~\ref{prop:for 3+1} implies that 
the vertices $v_3, y, z, v_4$ induce $K_{1, 3}$ in $\overline{G}$, 
and we have $(v_4, z) \in F$ from $(v_4, v_3) \in F$. 
Then $(u, v, v_3, z, v_4, v_5)$ is a regular obstruction of $P + F$, 
a contradiction.
\end{proof}

\begin{claim}\label{claim:Lemma4:B:2+1}
The orientation $P + F'$ contains no forbidden configurations for 
$\mathbf{2+1}$. 
\end{claim}
\begin{proof}[Proof of Claim~\ref{claim:Lemma4:B:2+1}]
Suppose that $P + F'$ contains a forbidden configuration $(x, y, z)$ 
for $\mathbf{2+1}$. 
If $(y, z), (z, x) \in B_{(u, v)}^{-1}$ then 
we have from Claim~\ref{claim:for base and fronts} that 
$(y, z) \in B_{(u, v)}^{-1}$ implies $(u, z) \in F$ and 
$(z, x) \in B_{(u, v)}^{-1}$ implies $(z, u) \in F$, 
a contradiction. 
Thus either $(y, z) \in B_{(u, v)}^{-1}$ or $(z, x) \in B_{(u, v)}^{-1}$. 
\par
We first suppose $(y, z) \in B_{(u, v)}^{-1}$ and $(z, x) \notin B_{(u, v)}^{-1}$. 
We have $(z, y) \in B_{(u, v)} \subseteq F$ and 
$(z, x) \in F$ but $(z, x) \notin B_{(u, v)}$. 
Since $(z, y) \in B_{(u, v)}$, there is a skewed obstruction 
$(u, v, v_2, v_3, v_4, v_5)$ of $P + F$ such that 
$(u, v)$ is its base and $(z, y)$ is its front. 
Suppose $(z, y) = (v_4, v_3)$ 
(Fig.~\ref{figs:Lemma4:Bf}\subref{figs:Lemma4:B4}). 
We have $(v_3, u), (v, v_4) \in F$ 
from Proposition~\ref{prop:for skewed obstruction}. 
If $(v, x) \in P$ then $(u, v), (x, v_3) \in P$ implies $(u, v_3) \in P$, 
contradicting $(v_3, u) \in F$. 
If $(x, v) \in P$ then $(x, v, v_4)$ is 
a forbidden configuration for $\mathbf{2+1}$ in $P + F$, a contradiction. 
Thus $vx \in E(\overline{G})$. 
Proposition~\ref{prop:for 2+2} implies that 
$(u, x, v, v_3)$ is $C_4$ in $\overline{G}$, and 
we have $(x, v) \in F$ from $(v_3, u) \in F$. 
Then $(u, v, x, v_3, v_4, v_5)$ is a regular obstruction of $P + F$, 
a contradiction.
\par
Suppose $(z, y) = (v_5, v_3)$ 
(Fig.~\ref{figs:Lemma4:Bf}\subref{figs:Lemma4:B4s}). 
We have $(v_3, u), (v, v_5) \in F$ 
from Proposition~\ref{prop:for skewed obstruction}. 
If $(v, x) \in P$ then $(u, v), (x, v_3) \in P$ implies $(u, v_3) \in P$, 
contradicting $(v_3, u) \in F$. 
If $(x, v) \in P$ then $(x, v, v_5)$ is 
a forbidden configuration for $\mathbf{2+1}$ in $P + F$, a contradiction. 
Thus $vx \in E(\overline{G})$. 
Proposition~\ref{prop:for 2+2} implies that 
$(u, x, v, v_3)$ is $C_4$ in $\overline{G}$, and 
we have $(x, v) \in F$ from $(v_3, u) \in F$. 
Then $(u, v, x, v_3, v_4, v_5)$ is a regular obstruction of $P + F$, 
a contradiction.
\par
We next suppose $(z, x) \in B_{(u, v)}^{-1}$ and $(y, z) \notin B_{(u, v)}^{-1}$. 
We have $(x, z) \in B_{(u, v)} \subseteq F$ and 
$(y, z) \in F$ but $(y, z) \notin B_{(u, v)}$. 
Since $(x, z) \in B_{(u, v)}$, there is a skewed obstruction 
$(u, v, v_2, v_3, v_4, v_5)$ of $P + F$ such that 
$(u, v)$ is its base and $(x, z)$ is its front. 
Suppose $(x, z) = (v_4, v_3)$ 
(Fig.~\ref{figs:Lemma4:Bf}\subref{figs:Lemma4:B5}). 
We have $(v_3, u), (v, v_4) \in F$ 
from Proposition~\ref{prop:for skewed obstruction}. 
If $(y, u) \in P$ then $(v_4, y), (u, v) \in P$ implies $(v_4, v) \in P$, 
contradicting $(v, v_4) \in F$. 
If $(u, y) \in P$ then $(u, y, v_3)$ is 
a forbidden configuration for $\mathbf{2+1}$ in $P + F$, a contradiction. 
Thus $uy \in E(\overline{G})$. 
Proposition~\ref{prop:for 2+2} implies that 
$(u, v_4, v, y)$ is $C_4$ in $\overline{G}$, and 
we have $(u, y) \in F$ from $(v, v_4) \in F$. 
Then $(u, v, v_2, v_3, v_4, y)$ is a skewed obstruction, and 
hence $(y, v_3) \in B_{(u, v)}$, a contradiction. 
When we suppose $(x, z) = (v_5, v_3)$, 
we have $(v_4, y) \in P$ from $(v_4, v_5), (v_5, y) \in P$. 
Thus by similar arguments, we have a contradiction. 
\end{proof}

\begin{claim}\label{claim:Lemma4:B:regular obstruction}
The orientation $P + F'$ contains no regular obstructions. 
\end{claim}
\begin{proof}[Proof of Claim~\ref{claim:Lemma4:B:regular obstruction}]
As in the proof of Claim~\ref{claim:Lemma4:A:obstruction}, 
if $P + F'$ contains a regular obstruction, then $P + F$ also contains a regular obstruction, 
a contradiction. 
\end{proof}

\begin{figure*}[t]
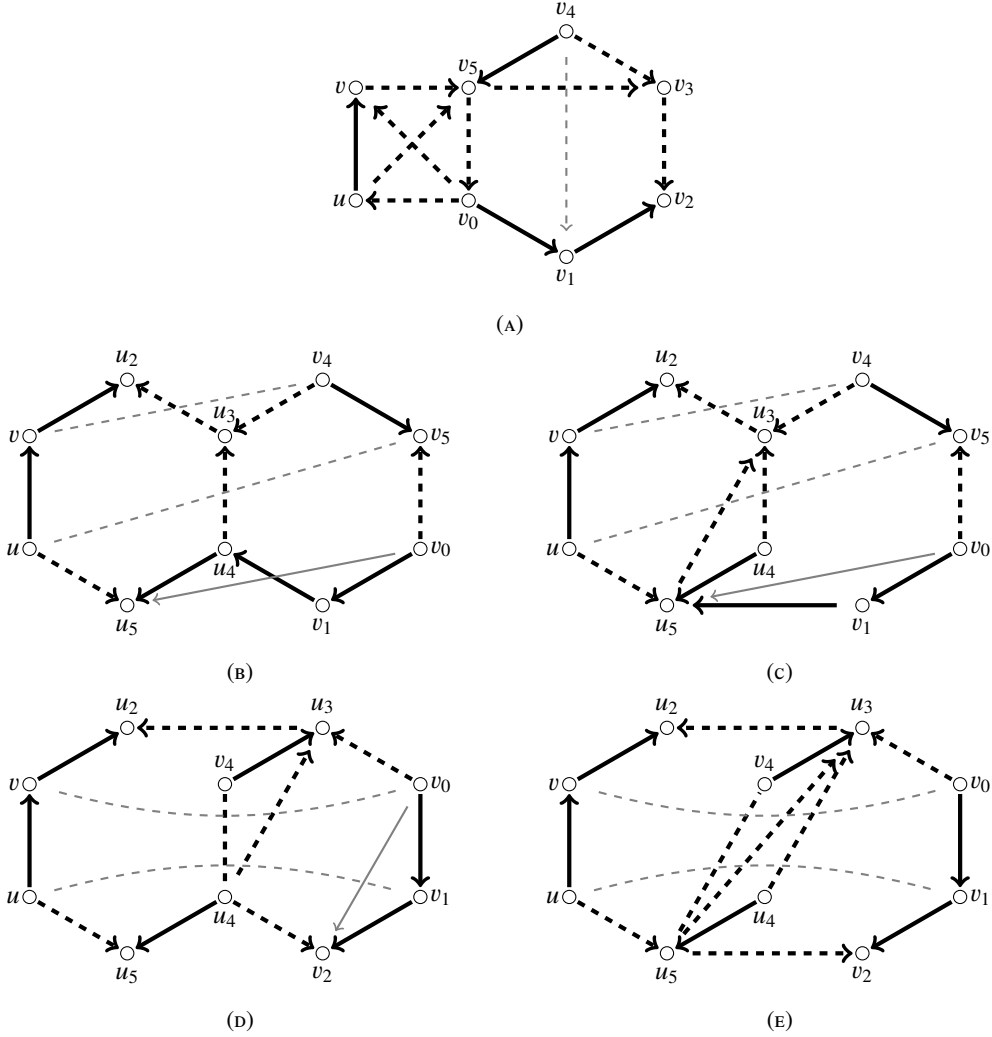

  \centering
  \subcaptionbox{\label{figs:Lemma4:B6}}{\input{figs/lemma4/B6}}
\\
  \subcaptionbox{\label{figs:Lemma4:B7}}{\input{figs/lemma4/B7}}
  \subcaptionbox{\label{figs:Lemma4:B7s}}{\input{figs/lemma4/B7s}}
\\
  \subcaptionbox{\label{figs:Lemma4:B8}}{\input{figs/lemma4/B8}}
  \subcaptionbox{\label{figs:Lemma4:B8s}}{\input{figs/lemma4/B8s}}
  \caption{
    Illustrating the proof of Claim~\ref{claim:Lemma4:B:skewed obstruction}. 
    Arrows, dashed arrows, and dashed lines denote the same type of edges 
    as in Fig.~\ref{figs:alternating-cycles}. 
  }
  \label{figs:Lemma4:Bs}
\end{figure*}

\begin{claim}\label{claim:Lemma4:B:skewed obstruction}
If $P + F$ contains no skewed obstructions with base $(v_0, v_1)$, 
then $P + F'$ contains no skewed obstructions with base $(v_0, v_1)$. 
\end{claim}
\begin{proof}[Proof of Claim~\ref{claim:Lemma4:B:skewed obstruction}]
Suppose that $P + F'$ contains a skewed obstruction 
$(v_0, v_2, v_2, v_3, v_4, v_5)$ with base $(v_0, v_1)$ 
whereas $P + F$ does not contain this obstruction. 
\par
Suppose $(v_0, v_5) \in B_{(u, v)}^{-1}$. 
From Claim~\ref{claim:for base and fronts}, 
we have $(u, v_5), (v, v_5), (v_0, u), (v_0, v) \in F$. 
If $(v_4, v_3) \in B_{(u, v)}^{-1}$ then 
from Claim~\ref{claim:for base and fronts}, 
we have $(v_4, v) \in F$, and hence $(u, v, v_4, v_5)$ is 
a forbidden configuration for $\mathbf{2+2}$ in $P + F$, a contradiction. 
If $(v_5, v_3) \in B_{(u, v)}^{-1}$ then 
from Claim~\ref{claim:for base and fronts}, 
we have $(v_5, u), (v_5, v) \in F$, 
contradicting $(u, v_5), (v, v_5) \in F$. 
If $(v_3, v_2) \in B_{(u, v)}^{-1}$ then 
from Claim~\ref{claim:for base and fronts}, 
we have $(u, v_2) \in F$, and hence $(u, v, v_0, v_2)$ is 
a forbidden configuration for $\mathbf{2+2}$ in $P + F$, a contradiction. 
Therefore, $(v_4, v_3), (v_5, v_3), (v_3, v_2) \notin B_{(u, v)}^{-1}$. 
\par
Now, $(v_5, v_0) \in B_{(u, v)} \subseteq F$ and 
$(v_4, v_3), (v_5, v_3), (v_3, v_2) \in F$ but 
$(v_4, v_3), (v_5, v_3), (v_3, v_2) \notin B_{(u, v)}$. 
From Claim~\ref{claim:for base and fronts}, we have 
$(u, v_5), (v, v_5), (v_0, u), (v_0, v) \in F$ 
(Fig.~\ref{figs:Lemma4:Bs}\subref{figs:Lemma4:B6}). 
If $(v_4, v_1) \in P$ then $(v_4, v_1, v_2, v_3)$ is 
a forbidden configuration for $\mathbf{3+1}$ in $P + F$, a contradiction. 
If $(v_1, v_4) \in P$ then 
$(v_0, v_1), (v_4, v_5) \in P$ implies $(v_0, v_5) \in P$, 
contradicting $(v_5, v_0) \in F$. 
Thus $v_4v_1 \in E(\overline{G})$. 
Proposition~\ref{prop:for 2+2} implies that 
$(v_4, v_0, v_5, v_1)$ is $C_4$ in $\overline{G}$, and 
we have $(v_4, v_1) \in F$ from $(v_5, v_0) \in F$. 
Then $(u, v, v_0, v_1, v_4, v_5)$ is a regular obstruction of $P + F$, 
a contradiction. 
\par
Suppose $(v_3, v_2) \in B_{(u, v)}^{-1}$. 
From Claim~\ref{claim:for base and fronts}, 
we have $(u, v_2), (v, v_2), (v_3, u), (v_3, v) \in F$. 
If $(v_4, v_3) \in B_{(u, v)}^{-1}$ or $(v_5, v_3) \in B_{(u, v)}^{-1}$, then 
from Claim~\ref{claim:for base and fronts}, 
we have $(u, v_3), (v, v_3) \in F$, 
contradicting $(v_3, u), (v_3, v) \in F$. 
If $(v_0, v_5) \in B_{(u, v)}^{-1}$ then 
from Claim~\ref{claim:for base and fronts}, 
we have $(v_0, v) \in F$, and thus $(u, v, v_0, v_2)$ is 
a forbidden configuration for $\mathbf{2+2}$ in $P + F$, a contradiction. 
Therefore, $(v_4, v_3), (v_5, v_3), (v_0, v_5) \notin B_{(u, v)}^{-1}$. 
\par
Now, $(v_2, v_3) \in B_{(u, v)} \subseteq F$ and 
$(v_4, v_3), (v_5, v_3), (v_0, v_5) \in F$ but 
$(v_4, v_3), (v_5, v_3), (v_0, v_5) \notin B_{(u, v)}$. 
Since $(v_2, v_3) \in B_{(u, v)}$, there is a skewed obstruction 
$(u, v, u_2, u_3, u_4, u_5)$ of $P + F$ such that 
$(u, v)$ is its base and $(v_2, v_3)$ is its front. 
Suppose $(v_2, v_3) = (u_4, u_3)$ 
(Fig.~\ref{figs:Lemma4:Bs}\subref{figs:Lemma4:B7}). 
We have $(v_0, u_5) \in P$ from $(v_0, v_1), (v_1, u_4), (u_4, u_5) \in P$. 
If $(u, v_5) \in P$ then $(u, v_5, v_0, u_5)$ is 
a forbidden configuration for $\mathbf{2+2}$ in $P + F$, a contradiction. 
If $(v_4, v) \in P$ then $(v_4, v, u_2, u_3)$ is 
a forbidden configuration for $\mathbf{3+1}$ in $P + F$, a contradiction. 
If $(v_5, u) \in P$ then $(v_4, v_5), (u, v) \in P$ implies $(v_4, v) \in P$, 
a contradiction. 
If $(v, v_4) \in P$ then $(u, v), (v_4, v_5) \in P$ implies $(u, v_5) \in P$, 
a contradiction. 
Thus $uv_5, vv_4 \in E(\overline{G})$. 
Proposition~\ref{prop:for 2+2} implies that 
$(u, v_4, v, v_5)$ is $C_4$ in $\overline{G}$. 
If $(u, v_5), (v, v_4) \in F$ then 
$(u, v, u_2, u_3, v_4, v_5)$ is a skewed obstruction, and 
hence $(v_4, u_3) \in B_{(u, v)}$, a contradiction. 
If $(v_5, u), (v_4, v) \in F$ then 
$(u, v, v_4, v_5, v_0, u_5)$ is a regular obstruction of $P + F$, 
a contradiction. 
When we suppose $(v_2, v_3) = (u_5, u_3)$ 
(Fig.~\ref{figs:Lemma4:Bs}\subref{figs:Lemma4:B7s}), 
we have $(v_0, u_5) \in P$ from $(v_0, v_1), (v_1, u_5) \in P$. 
Thus by similar arguments, we have a contradiction. 
\par
Suppose $(v_5, v_3) \in B_{(u, v)}^{-1}$. 
From Claim~\ref{claim:for base and fronts}, we have 
$(u, v_3), (v, v_3), (v_5, u), (v_5, v) \in F$. 
If $(v_3, v_2) \in B_{(u, v)}^{-1}$ then 
from Claim~\ref{claim:for base and fronts}, 
we have $(v_3, u), (v_3, v) \in F$, 
contradicting $(u, v_3), (v, v_3) \in F$. 
If $(v_0, v_5) \in B_{(u, v)}^{-1}$ then 
from Claim~\ref{claim:for base and fronts}, 
we have $(u, v_5), (v, v_5) \in F$, 
contradicting $(v_5, u), (v_5, v) \in F$. 
Therefore, $(v_3, v_2), (v_0, v_5) \notin B_{(u, v)}^{-1}$. 
\par
Now, $(v_3, v_5) \in B_{(u, v)} \subseteq F$ and 
$(v_3, v_2), (v_0, v_5) \in F$ but 
$(v_3, v_2), (v_0, v_5) \notin B_{(u, v)}$. 
Note that it does not matter whether 
$(v_4, v_3) \in B_{(u, v)}^{-1}$ or $(v_4, v_3) \notin B_{(u, v)}^{-1}$. 
Since $(v_3, v_5) \in B_{(u, v)}$, there is a skewed obstruction 
$(u, v, u_2, u_3, u_4, u_5)$ of $P + F$ such that 
$(u, v)$ is its base and $(v_3, v_5)$ is its front. 
Suppose $(v_3, v_5) = (u_4, u_3)$ 
(Fig.~\ref{figs:Lemma4:Bs}\subref{figs:Lemma4:B8}). 
If $(u, v_1) \in P$ then $(v_1, v_2) \in P$ implies $(u, v_2) \in P$. 
Thus $(u, v_2, u_4, u_5)$ is 
a forbidden configuration for $\mathbf{2+2}$ in $P + F$, a contradiction. 
If $(v_0, v) \in P$ then $(v_0, v, u_2, u_3)$ is 
a forbidden configuration for $\mathbf{3+1}$ in $P + F$, a contradiction. 
If $(v_1, u) \in P$ then $(v_0, v_1), (u, v) \in P$ implies $(v_0, v) \in P$, 
a contradiction. 
If $(v, v_0) \in P$ then $(u, v), (v_0, v_1) \in P$ implies $(u, v_1) \in P$, 
a contradiction. 
Thus $uv_1, vv_0 \in E(\overline{G})$. 
Proposition~\ref{prop:for 2+2} implies that 
$(u, v_0, v, v_1)$ is $C_4$ in $\overline{G}$. 
If $(u, v_1), (v, v_0) \in F$ then 
$(u, v, u_2, u_3, v_0, v_1)$ is a skewed obstruction, and 
hence $(v_0, u_3) \in B_{(u, v)}$, a contradiction. 
If $(v_1, u), (v_0, v) \in F$ then 
$(u, v, v_0, v_2, u_4, u_5)$ is a regular obstruction of $P + F$, 
a contradiction. 
\par
Suppose $(v_3, v_5) = (u_5, u_3)$ 
(Fig.~\ref{figs:Lemma4:Bs}\subref{figs:Lemma4:B8s}). 
If $(u, v_1) \in P$ then $(u, v_1, v_2, u_5)$ is 
a forbidden configuration for $\mathbf{3+1}$ in $P + F$, a contradiction. 
If $(v_0, v) \in P$ then $(v_0, v, u_2, u_3)$ is 
a forbidden configuration for $\mathbf{3+1}$ in $P + F$, a contradiction. 
If $(v_1, u) \in P$ then $(v_0, v_1), (u, v) \in P$ implies $(v_0, v) \in P$, 
a contradiction. 
If $(v, v_0) \in P$ then $(u, v), (v_0, v_1) \in P$ implies $(u, v_1) \in P$, 
a contradiction. 
Thus $uv_1, vv_0 \in E(\overline{G})$. 
Proposition~\ref{prop:for 2+2} implies that 
$(u, v_0, v, v_1)$ is $C_4$ in $\overline{G}$. 
If $(u, v_1), (v, v_0) \in F$ then 
$(u, v, u_2, u_3, v_0, v_1)$ is a skewed obstruction, and 
hence $(v_0, u_3) \in B_{(u, v)}$, a contradiction. 
If $(v_1, u), (v_0, v) \in F$ then 
$(v_0, v_1, v_2, u_5, u, v)$ is a skewed obstruction 
with base $(v_0, v_1)$, contradicting the assumption that 
$P + F$ contains no skewed obstructions with base $(v_0, v_1)$. 
\par
Suppose $(v_4, v_3) \in B_{(u, v)}^{-1}$. 
If $(v_5, v_3) \notin B_{(u, v)}^{-1}$ then 
$(v_5, v_3), (v_3, v_4) \in F$, that is, 
$(v_4, v_5, v_3)$ is 
a forbidden configuration for $\mathbf{2+1}$ in $P + F$, a contradiction. 
Thus we have $(v_5, v_3) \in B_{(u, v)}^{-1}$, a contradiction. 
\end{proof}

The following is trivial from the definition of $B_{(u, v)}$. 
\begin{claim}\label{claim:Lemma4:B:remove}
The orientation $P + F'$ contains no skewed obstructions with base $(u, v)$. 
\end{claim}

Claims~\ref{claim:Lemma4:B:2+2}--\ref{claim:Lemma4:B:remove} ensure that 
by repeating the procedure for each edge of $P$, 
we can obtain an orientation $F''$ of $E_o$ such that 
$P + F''$ contains no forbidden configurations for 
$\mathbf{2+2}$, $\mathbf{3+1}$, or $\mathbf{2+1}$, and 
$P + F''$ contains no obstructions. 
Therefore, we can compute the orientation by Algorithm~\ref{algo:lemma4}.

\input{figs/algo3.tex}

Each iteration on lines 3--7 takes $O(1)$ time for each pair of edges of $F$. 
Hence $A_{(u, v)}$ can be constructed in $O(\bar{m}^2)$ time 
and $|A_{(u, v)}| = O(\bar{m}^2)$. 
Therefore, we can remove all regular obstructions in $O(m \bar{m}^2)$ time 
(lines 1--9). 
Similarly, we can remove all skewed obstructions in $O(m \bar{m}^2)$ time 
(lines 10--18). 
Hence Algorithm~\ref{algo:lemma4} takes $O(m \bar{m}^2)$ time. 
\end{proof}

Now from the lemmas above, 
we have the main theorem and the following characterization. 

\begin{Theorem}\label{theorem:recognition:structural result}
A partial order $P$ on a set $V$ is a linear-semiorder 
if and only if there is an orientation $F$ of $E_o$ 
(or, equivalently, an orientation $F$ of 
the incomparability graph of $P$) such that 
$P + F$ contains no forbidden configurations for 
$\mathbf{2+2}$, $\mathbf{3+1}$, or $\mathbf{2+1}$. 
\end{Theorem}
\begin{proof}
Let $\overline{G}$ be the incomparability graph of $P$. 
If $P$ is a linear-semiorder, then 
Proposition~\ref{prop:consequence of characterization} implies that 
there is an orientation $F$ of $\overline{G}$ 
such that $P + F$ is acyclic and 
contains no forbidden configurations for $\mathbf{2+2}$ or $\mathbf{3+1}$. 
Since $P + F$ is acyclic, 
it contains no forbidden configurations for $\mathbf{2+1}$. 
\par
Suppose there is an orientation $F$ of $E_o$ such that 
$P + F$ contains no forbidden configurations for 
$\mathbf{2+2}$, $\mathbf{3+1}$, or $\mathbf{2+1}$. 
From Lemma~\ref{lemma:recognition:4}, we have 
another orientation $F'$ of $E_o$ such that 
$P + F'$ contains no forbidden configurations for 
$\mathbf{2+2}$, $\mathbf{3+1}$, or $\mathbf{2+1}$, and 
$P + F'$ contains no obstructions. 
Since $P + F'$ contains no obstructions, 
Lemma~\ref{lemma:recognition:3} implies that $P + F'$ is acyclic. 
By Lemma~\ref{lemma:recognition:1}, $P$ is a linear-semiorder. 
\end{proof}

\begin{proof}[Proof of Theorem~\ref{theorem:recognition:main}]
The first claim can be obtained from Theorem~\ref{theorem:recognition:structural result} 
and Lemma~\ref{lemma:recognition:2}. 
The second claim can be obtained from 
Lemmas~\ref{lemma:recognition:4},~\ref{lemma:recognition:3},~\ref{lemma:recognition:2} and~\ref{lemma:recognition:1}. 
\end{proof}

\section{Comparability invariance}\label{sec:compa}
A property of partial orders is called 
a \emph{comparability invariant} if either all orders 
with the same comparability graph 
have that property or none has that property. 
It is known that being a linear-interval order 
is a comparability invariant~\cite{COS08-ENDM}. 
In this section, we show the following. 
\begin{Theorem}\label{comparability invariant}
Being a linear-semiorder is a comparability invariant. 
\end{Theorem}

From Theorems~\ref{comparability invariant} 
and~\ref{theorem:recognition:main}, we have the following. 
\begin{Corollary}
Comparability graphs of linear-semiorders 
can be recognized in $O(n\bar{m})$ time, 
where $n$ and $\bar{m}$ are 
the number of vertices and non-edges 
of the given graph, respectively. 
Incomparability graphs of linear-semiorders 
can be recognized in $O(nm)$ time, 
where $m$ is the number of edges of the given graph. 
\end{Corollary}
\begin{proof}
We can recognize comparability graphs of linear-semiorders 
by the following algorithm. 
Let $G$ be the given graph. 
First, obtain a transitive orientation $P$ of $G$. 
If $G$ does not have a transitive orientation, 
$G$ is not a comparability graph. 
Then, test whether $P$ is a linear-semiorder 
by Algorithms~\ref{algo:E_o}--\ref{algo:lemma4}. 
Theorem~\ref{comparability invariant} ensures that 
every transitive orientation of $G$ is a linear-semiorder 
if $G$ is the comparability graph of a linear-semiorder. 
Thus the algorithm is correct. 
Both steps can be done in $O(n\bar{m})$ time 
as shown in Proposition~\ref{prop:obtain transitive orientation} 
and Theorem~\ref{theorem:recognition:main}. 
Thus the algorithm takes $O(n\bar{m})$ time. 
\par
Since the complement of a graph can be obtained in $O(n^2)$ time, 
incomparability graphs of linear-semiorders 
can be recognized in $O(nm)$ time. 
\end{proof}
\par
To prove Theorem~\ref{comparability invariant}, 
we use the proof technique developed in~\cite{FM98-Order}. 
Let $P$ be a partial order on a set $V$. 
A non-empty subset $A \subseteq V$ is \emph{autonomous} in $P$ 
if for any element $v \in V - A$, whenever $v \prec a$, 
$v \parallel a$, or $v \succ a$ holds in $P$ 
for some element $a \in A$, then the same holds 
for all elements $a \in A$. 
Let $P'$ be a partial order with the same comparability graph as $P$. 
The order $P'$ is said to be obtained from $P$ by a \emph{reversal} 
if there is an autonomous set $A$ of $P$ such that 
\begin{enumerate}[\bfseries --]
\item if not both of $u$ and $v$ are in $A$, then $u \prec v$ in $P'$ if and only if $u \prec v$ in $P$, and 
\item if both of $u$ and $v$ are in $A$, then $u \prec v$ in $P'$ if and only if $u \succ v$ in $P$. 
\end{enumerate}
We denote by $P \mid A$ the order obtained from $P$ 
by reversing $A$. 

The following theorem provides a simple scheme to show 
the comparability invariance. 
\begin{Theorem}[\cite{Gallai67}]\label{theorem:Gallai67}
Two orders $P$ and $P'$ have the same comparability graph 
if and only if there is a finite sequence of orders 
$P_0, P_1, \ldots, P_k$ such that $P_0 = P$, $P_k = P'$, and 
$P_i$ is obtained from $P_{i-1}$ by a reversal 
for each $i = 1, 2, \ldots, k$. 
\end{Theorem}
Therefore, to prove Theorem~\ref{comparability invariant}, 
it suffices to show the following. 
\begin{Lemma}\label{lemma:comparability invariant}
If an order $P$ on a set $V$ is a linear-semiorder and 
a subset $A \subseteq V$ is autonomous in $P$, 
then $P \mid A$ is a linear-semiorder. 
\end{Lemma}

\begin{proof}[Proof of Theorem~\ref{comparability invariant}]
Let $P$ and $P'$ be two orders with the same comparability graph. 
It suffices to show that 
$P$ is a linear-semiorder if and only if $P'$ is a linear-semiorder. 
From Theorem~\ref{theorem:Gallai67}, 
there is a finite sequence of orders 
$P = P_0, P_1, \ldots, P_k = P'$ such that 
$P_i$ is obtained from $P_{i-1}$ by a reversal 
for each $i = 1, 2, \ldots, k$. 
Suppose that $P = P_0$ is a linear-semiorder. 
From Lemma~\ref{lemma:comparability invariant}, 
$P_1$ is also a linear-semiorder. 
Continuing in this way, 
we have that $P_k = P'$ is a linear-semiorder. 
A symmetric argument would show that the converse also holds. 
\end{proof}

In the rest of this section, we prove Lemma~\ref{lemma:comparability invariant}. 
Recall that $L_1$ and $L_2$ are two horizontal lines 
in the $xy$-plane with $L_1$ above $L_2$. 
As a representation of a linear-semiorder, 
we use a set of triangles between $L_1$ and $L_2$ as follows. 

Let $P$ be a linear-semiorder on a set $V$, and 
let $L$ and $S$ be a linear order and a semiorder on $V$, 
respectively, with $L \cap S = P$. 
Let $\{p(v) \colon\ v \in V\}$ be a set of points on $L_1$ 
such that $p(u) < p(v)$ if and only if $u \prec v$ in $L$ 
for any two elements $u , v \in V$. 
Note that all the points are distinct. 
Let $\{I(v) \colon\ v \in V\}$ be a unit interval representation 
of $S$ on $L_2$. 
We assume that no two intervals share a common endpoint. 
Let $I(v) = [l(v), r(v)]$ for each element $v \in V$. 
Since each interval has unit length, $l(v) + 1 = r(v)$. 

Let $T(v)$ be the triangle spanned by $p(v)$ and $I(v)$. 
A triangle $T(u)$ \emph{lies completely to the left of} $T(v)$, 
and we write $T(u) \ll T(v)$, 
if $p(u) < p(v)$ on $L_1$ and $I(u) \ll I(v)$ on $L_2$. 
We have that $u \prec v$ in $P$ if and only if $T(u) \ll T(v)$ 
for any two elements $u, v \in V$; hence 
we call the set $\{T(v) \colon\ v \in V\}$ 
a \emph{triangle representation} of $P$. 

Now, we fix a pair of 
a linear-semiorder $P$ and an autonomous set $A$ of $P$. 
We also fix a triangle representation $\{T(v) \colon\ v \in V\}$ of $P$. 

An element $a \in A$ is \emph{isolated} in $A$ if 
$a \parallel a'$ in $P$ for any element $a' \in A - \{a\}$. 
Let $A^*$ be the subset of $A$ obtained by removing 
all isolated elements of $A$. 
We can observe the following, and then 
we assume without loss of generality $A^* \neq \emptyset$. 
\begin{Lemma}\label{lemma:2}
The set $A^*$ is autonomous in $P$, and $P \mid A^* = P \mid A$. 
\end{Lemma}

We define that 
$l_1 = \min_{v \in A^*} p(v)$ and 
$r_1 = \max_{v \in A^*} p(v)$, similarly, 
$l_2 = \min_{v \in A^*} l(v)$ and 
$r_2 = \max_{v \in A^*} r(v)$. 
Let $C(A^*)$ be the trapezoid spanned by 
the interval $[l_1, r_1]$ on $L_1$ and 
the interval $[l_2, r_2]$ on $L_2$, while 
$l_1$ is the upper-left  corner of $C(A^*)$ and 
$r_1$ is the upper-right corner of $C(A^*)$, similarly, 
$l_2$ is the lower-left  corner of $C(A^*)$ and 
$r_2$ is the lower-right corner of $C(A^*)$. 
Note that $C(A^*)$ can be viewed as the convex region 
spanned by the triangles $T(v)$ with $v \in A^*$. 
In addition, let $a_l^1$ and $a_r^1$ denote the elements of $A^*$ with 
$p(a_l^1) = l_1$ and 
$p(a_r^1) = r_1$, respectively. 
Similarly, 
let $a_l^2$ and $a_r^2$ denote the elements of $A^*$ with 
$l(a_l^2) = l_2$ and 
$r(a_r^2) = r_2$, respectively. 

Let $B$ be the set of elements $v \in V$ with $T(v) \subseteq C(A^*)$. 
Obviously, $B \supseteq A^*$. 
The convex region $C(B)$ spanned by the triangles $T(v)$ with $v \in B$ 
is equal to $C(A^*)$. 
\begin{Lemma}\label{lemma:3}
Let $v \in V - B$. If $T(v) \cap C(A^*) \neq \emptyset$ 
then $T(v)$ intersects with every triangle 
spanned by a point $p \in [l_1, r_1]$ on $L_1$ and 
a unit-length interval $I \subseteq [l_2, r_2]$ on $L_2$. 
\end{Lemma}
\begin{proof}
Since $v \notin B$, we have $T(v) \not\subseteq C(A^*)$, 
and hence $p(v) \notin [l_1, r_1]$ or $I(v) \not\subseteq [l_2, r_2]$. 
\par
Suppose $p(v) < l_1$. 
Since $T(v) \cap C(A^*) \neq \emptyset$, we have $l_2 < r(v)$. 
Thus $T(v) \cap T(a_l^2) \neq \emptyset$, 
and hence $v \parallel a_l^2$ in $P$. 
Since $A^*$ is autonomous, $v \parallel a_r^2$ in $P$, 
and hence $T(v) \cap T(a_r^2) \neq \emptyset$. 
Thus $l(a_r^2) < r(v)$. 
Therefore, $p(v) < l_1$ and $l(a_r^2) = r_2 - 1 < r(v)$, 
and the claim holds. 
A similar argument would show that 
the claim holds when $r_1 < p(v)$. 
\par
Suppose $l(v) < l_2$. 
If $l_2 < r(v)$ then $T(v) \cap T(a_l^2) \neq \emptyset$. 
If $r(v) < l_2$ then $T(v) \cap C(A^*) \neq \emptyset$ implies 
$l_1 < p(v)$, and hence $T(v) \cap T(a_l^1) \neq \emptyset$. 
Thus in both cases, $v \parallel a$ in $P$ for some element $a \in A^*$. 
Since $A^*$ is autonomous, $v \parallel a_r^1$ in $P$. 
Since $a_r^1$ is not isolated, 
there is an element $a' \in A^*$ with $T(a') \ll T(a_r^1)$. 
Since $l(v) < l_2 \leq l(a')$ and each interval has unit length, 
$r(v) < r(a')$, and hence $I(v) \ll I(a_r^1)$. 
It follows $p(a_r^1) < p(v)$. 
Therefore, $r_1 < p(v)$ and $l(v) < l_2$, 
and the claim holds. 
A similar argument would show that 
the claim holds when $r_2 < r(v)$. 
\end{proof}

From Lemma~\ref{lemma:3}, we have the following. 
\begin{Lemma}\label{lemma:4}
For any element $v \in V - B$, one of the following holds: 
\begin{enumerate}[\bfseries --]
\item $T(v)$ lies completely to the left of $C(A^*)$. 
\item $T(v)$ intersects with every triangle contained in $C(A^*)$. 
\item $T(v)$ lies completely to the right of $C(A^*)$. 
\end{enumerate}
\end{Lemma}

We also have the following from Lemma~\ref{lemma:4}
\begin{Lemma}\label{lemma:5}
The set $B$ is autonomous in $P$. 
\end{Lemma}
\begin{proof}
Let $v \in V - B$. 
Suppose that $v \parallel b$ in $P$ for some element $b \in B$. 
Then $T(v) \cap T(b) \neq \emptyset$ implies 
$T(v) \cap C(A^*) \neq \emptyset$, and hence 
$T(v) \cap T(b) \neq \emptyset$ for all elements $b \in B$. 
Thus $v \parallel b$ in $P$ for all elements $b \in B$. 
If $v \prec b$ in $P$ for some element $b \in B$, 
then $T(v)$ lies completely to the left of $C(A^*)$, 
and hence $v \prec b$ in $P$ for all elements $b \in B$. 
If $v \succ b$ in $P$ for some element $b \in B$, 
then $T(v)$ lies completely to the right of $C(A^*)$, 
and hence $v \succ b$ in $P$ for all elements $b \in B$. 
\end{proof}

Let $P_1 = P \mid B$. 
We define $\{T'(v) \colon\ v \in V\}$ 
as the set of triangles obtained from 
$\{T(v) \colon\ v \in V\}$ 
by \emph{flipping the representation of $B$ relative to $C(B)$}, 
that is, $T'(v) = T(v)$ for any element $v \in V - B$ 
whereas for each element $v \in B$, 
the triangle $T'(v)$ is the triangle spanned by the point 
$l_1 + r_1 - p(v)$ on $L_1$ and the interval 
$[l_2 + r_2 - r(v), l_2 + r_2 - l(v)]$ on $L_2$. 
Note that the convex region $C'(B)$ 
spanned by the triangles $T'(v)$ with $v \in B$ 
is equal to $C(B)$. 
\begin{Lemma}\label{lemma:6}
The order $P_1$ is a linear-semiorder with 
a representation $\{T'(v) \colon\ v \in V\}$. 
\end{Lemma}
\begin{proof}
Let $P'$ be the partial order on $V$ defined by 
$\{T'(v) \colon\ v \in V\}$, 
and let $u$ and $v$ be two elements of $V$. 
If neither $u$ nor $v$ is in $B$, then 
their relation in $P'$ remains the same as in $P$. 
If either $u$ or $v$ is in $B$, then 
Lemma~\ref{lemma:4} ensures that 
their relation in $P'$ remains the same as in $P$. 
If both $u$ and $v$ are in $B$, then 
$T'(u) \ll T'(v) \iff T(u) \gg T(v)$; 
hence $u \prec v$ in $P'$ 
if and only if $u \succ v$ in $P$. 
Thus $P' = P_1$. 
\end{proof}

Let $A_1 = B - A^*$. 
If $A_1 = \emptyset$ then $B = A^*$, and hence 
$P \mid A = P \mid A^* = P \mid B = P_1$ 
is a linear-semiorder. 
Thus we assume without loss of generality $A_1 \neq \emptyset$. 
\begin{Lemma}\label{lemma:7}
All elements of $A_1$ are incomparable to 
all elements of $A^*$ in $P$. 
\end{Lemma}
\begin{proof}
Suppose there is an element $a \in A_1$ 
with $a \prec a'$ in $P$ for some element $a' \in A^*$. 
Since $A^*$ is autonomous, 
$a \prec a_l^1$ and $a \prec a_l^2$ in $P$, 
contradicting $T(a) \cap C(A^*) \neq \emptyset$. 
Similarly, there are no elements $a \in A_1$ 
with $a' \prec a$ in $P$ for some element $a' \in A^*$. 
\end{proof}

\begin{Lemma}\label{lemma:8}
The set $A_1$ is autonomous in $P_1$, and 
$P_1 \mid A_1 = P \mid A$. 
\end{Lemma}
\begin{proof}
Let $v \in V - A_1$. 
If $v \in V - B$ then, since $B$ is autonomous in $P_1$, 
whenever $v \prec a$, $v \parallel a$, or $v \succ a$ holds in $P_1$ 
for some element $a \in A_1 \subseteq B$, 
the same holds for all elements $a \in A_1$. 
If $v \in B - A_1$ then $v \in A^*$, 
and we have from Lemma~\ref{lemma:7} that 
$v \parallel a$ in $P_1$ for all elements $a \in A_1$. 
Thus $A_1$ is autonomous in $P_1$. 
\par
Lemma~\ref{lemma:7} implies $P \mid B = (P \mid A^*) \mid A_1$. 
Thus $(P \mid B) \mid A_1 = P \mid A^*$. 
Since $P_1 = P \mid B$ and $P \mid A^* = P \mid A$, 
we have $P_1 \mid A_1 = P \mid A$. 
\end{proof}

We repeat the process by replacing $P$ and $A$ with $P_1$ and $A_1$, respectively. 
Let $A_1^*$ be the subset of $A_1$ obtained by removing all 
isolated elements of $A_1$. 
From Lemma~\ref{lemma:2}, 
the set $A_1^*$ is autonomous in $P_1$ 
and $P_1 \mid A_1^* = P_1 \mid A_1$. 
We assume without loss of generality 
$A_1^* \neq \emptyset$. 
We define $C'(A_1^*)$ as the convex region spanned by 
the triangles $T'(v)$ with $v \in A_1^*$. 
Let $B_1$ be the set of elements $v \in V$ with $T'(v) \subseteq C'(A_1^*)$. 
From Lemma~\ref{lemma:5}, the set $B_1$ is autonomous in $P_1$, 
and let $P_2 = P_1 \mid B_1$. 
From Lemma~\ref{lemma:6}, the order $P_2$ is a linear-semiorder. 
Let $A_2 = B_1 - A_1^*$, and we assume without loss of generality 
$A_2 \neq \emptyset$. 
From Lemma~\ref{lemma:8}, the set $A_2$ is autonomous in $P_2$, 
and $P_2 \mid A_2 = P_1 \mid A_1$. 

Let $A_2^*$ be the subset of $A_2$ obtained by removing all 
isolated elements of $A_2$. 
From Lemma~\ref{lemma:2}, 
the set $A_2^*$ is autonomous in $P_2$, 
and $P_2 \mid A_2^* = P_2 \mid A_2$. 
\begin{Lemma}\label{lemma:9}
The set $A_2^*$ is a proper subset of $A^*$. 
\end{Lemma}
\begin{proof}
Since $A_1^* \subseteq B$, 
we have $C'(A_1^*) \subseteq C'(B) = C(B) = C(A^*)$. 
Thus $B_1 \subseteq B$. 
Suppose $a_l^1 \in B_1$. 
Then $T'(a_l^1) \subseteq C'(A_1^*)$. 
Since $a_l^1 \in A^*$, we have $a_l^1 \notin A_1^* = B - A^*$. 
Thus there are two elements $a, b \in A_1^*$ with $p'(a) < p'(a_l^1) < p'(b)$, 
where $p'(v)$ is the point of $T'(v)$ on $L_1$. 
However, $p'(a_l^1) = l_1 + r_1 - p(a_l^1) = r_1$, 
that is, $r_1 < p'(b)$, 
contradicting $C'(A_1^*) \subseteq C(A^*)$. 
Therefore, $a_l^1 \notin B_1$. 
\par
The set $B$ can be partitioned into three sets 
$A^*$, $A_1 - A_1^*$, and $A_1^*$. 
By definition, all elements of $A_1 - A_1^*$ 
are incomparable to all other elements of $A_1$. 
From Lemma~\ref{lemma:7}, all elements of $A_1 - A_1^*$ 
are incomparable to all elements of $A^*$. 
Thus all elements of $A_1 - A_1^*$ 
are incomparable to all other elements of $B$. 
\par
We have $A_2^* \subseteq A_2 = B_1 - A_1^* \subseteq B - A_1^* = A^* \cup (A_1 - A_1^*)$. 
Since all elements of $A_1 - A_1^*$ 
are incomparable to all elements of $A_2 \subseteq B$, 
we have $A_2^* \cap (A_1 - A_1^*) = \emptyset$. 
Thus $A_2^* \subseteq A^*$. 
Since $a_l^1 \in A^*$ but $a_l^1 \notin B_1 \supseteq A_2^*$, 
the inclusion is proper. 
\end{proof}

Now, we prove Lemma~\ref{lemma:comparability invariant}. 
\begin{proof}[Proof of Lemma~\ref{lemma:comparability invariant}] 
Suppose that the lemma does not hold. 
Then there is a pair of a linear-semiorder $P$ 
and an autonomous set $A$ of $P$ such that 
$P \mid A$ is not a linear-semiorder. 
Among such pairs, we choose one with minimal $|A|$. 
Lemma~\ref{lemma:8} implies $P \mid A = P_1 \mid A_1 = P_2 \mid A_2 = P_2 \mid A_2^*$, but 
Lemma~\ref{lemma:9} implies $|A_2^*| < |A^*| \leq |A|$, contradicting the minimality of $A$. 
\end{proof}

\section{Miscellaneous}\label{sec:byproducts}
We finally show two byproducts of Theorem~\ref{theorem:characterization}. 

\subsection{Vertex ordering characterization}
Let $G = (V, E)$ be a graph. 
A \emph{vertex ordering} of $G$ is a linear order on $V$. 
A \emph{vertex ordering characterization} of a graph class $\mathcal{C}$ 
is a characterization of the following type: 
a graph $G$ is in $\mathcal{C}$ if and only if $G$ has 
a vertex ordering fulfilling some properties. 
For example, a graph $G$ is the incomparability graph of an order 
if and only if there is a vertex ordering $L$ of $G$ 
such that for any three vertices $u, v, w$ of $G$ with $u \prec v \prec w$ in $L$, 
if $uw \in E$ then $uv \in E$ or $vw \in E$~\cite{KS93-SIAMDM}. 
Equivalently, a graph is the incomparability graph of an order 
if and only if it has a vertex ordering that contains no suborderings 
in Fig.~\ref{figs:forbidden-patterns}\subref{figs:cpc}. 

Incomparability graphs of linear-interval orders can be characterized 
so that a graph $G$ is the incomparability graph of a linear-interval order 
if and only if there is a vertex ordering of $G$ 
that contains no suborderings 
in Figs.~\ref{figs:forbidden-patterns}\subref{figs:cpc}--\subref{figs:p2}~\cite{Takaoka18-DM}. 
Moreover, a vertex ordering of $G$ that contains no suborderings 
in Figs.~\ref{figs:forbidden-patterns}\subref{figs:cpc}--\subref{figs:p2} 
can be obtained in $O(nm)$ time~\cite{Takaoka20a-DAM}, 
where $n$ and $m$ are the number of vertices and edges of $G$, respectively. 
For linear-semiorders, the following can be obtained from 
Theorems~\ref{theorem:characterization} and~\ref{theorem:recognition:main}. 

\begin{Corollary}
A graph $G$ is the incomparability graph of a linear-semiorder 
if and only if there is a vertex ordering of $G$ 
that contains no suborderings 
in Figs.~\ref{figs:forbidden-patterns}\subref{figs:cpc}--\subref{figs:k2}. 
A vertex ordering of a graph $G$ that contains no suborderings 
in Figs.~\ref{figs:forbidden-patterns}\subref{figs:cpc}--\subref{figs:k2}
can be computed in $O(m^2 \bar{m})$ time 
if $G$ is the incomparability graph of a linear-semiorder. 
Here, $m$ and $\bar{m}$ are the number of edges and non-edges 
of $G$, respectively. 
\end{Corollary}

\begin{figure}[t]
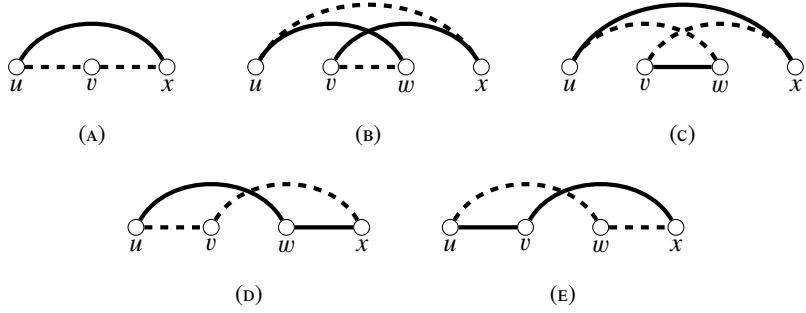

  \centering
  \subcaptionbox{\label{figs:cpc}}{\input{figs/forbidden/cpc.tex}}
  \subcaptionbox{\label{figs:p1}}{\input{figs/forbidden/p1.tex}}
  \subcaptionbox{\label{figs:p2}}{\input{figs/forbidden/p2.tex}}
\\
  \subcaptionbox{\label{figs:k1}}{\input{figs/forbidden/k1.tex}}
  \subcaptionbox{\label{figs:k2}}{\input{figs/forbidden/k2.tex}}
  \caption{
    Forbidden patterns in vertex orderings where $u \prec v \prec w \prec x$. 
    Lines and dashed lines denote edges and non-edges, respectively. 
    Edges that may or may not be present is not drawn. 
    } 
  \label{figs:forbidden-patterns}
\end{figure}

\begin{proof}
Assume there is a linear-semiorder $P$ on a set $V$ 
such that the incomparability graph of $P$ is $G$. 
Then $P$ has a linear extension $L$ fulfilling 
the $\mathbf{2+2}$ and $\mathbf{3+1}$ rules. 
Notice that $L$ can be viewed as a vertex ordering of $G$. 
Since $L$ is a linear extension of a partial order, 
it has no suborderings in Fig.~\ref{figs:forbidden-patterns}\subref{figs:cpc}. 
Since $L$ has no forbidden configurations for $\mathbf{2+2}$, 
it has no suborderings in Figs.~\ref{figs:forbidden-patterns}\subref{figs:p1} and~\ref{figs:forbidden-patterns}\subref{figs:p2}. 
Since $L$ has no forbidden configurations for $\mathbf{3+1}$, 
it has no suborderings in Figs.~\ref{figs:forbidden-patterns}\subref{figs:k1} and~\ref{figs:forbidden-patterns}\subref{figs:k2}. 
\par
Conversely, assume there is a vertex ordering $L$ of $G = (V, E)$ 
with no suborderings 
in Figs.~\ref{figs:forbidden-patterns}\subref{figs:cpc}--\subref{figs:k2}. 
Let $P$ be a binary relation on $V$ such that 
$(u, v) \in P$ if and only if $uv \notin E$ and $u \prec v$ in $L$ 
for any two elements $u, v \in V$. 
Since $L$ has no suborderings in Fig.~\ref{figs:forbidden-patterns}\subref{figs:cpc}, 
the relation $P$ is transitive, and hence a partial order. 
Then $L$ can be viewed as a linear extension of $P$. 
Since $L$ has no suborderings in Figs.~\ref{figs:forbidden-patterns}\subref{figs:p1} and~\ref{figs:forbidden-patterns}\subref{figs:p2}, 
it has no forbidden configurations for $\mathbf{2+2}$. 
Since $L$ has no suborderings in Figs.~\ref{figs:forbidden-patterns}\subref{figs:k1} and~\ref{figs:forbidden-patterns}\subref{figs:k2}, 
it has no forbidden configurations for $\mathbf{3+1}$. 
Therefore, Theorem~\ref{theorem:characterization} implies that 
$P$ is a linear-semiorder. 
Since $G$ is the incomparability graph of $P$, 
we have the first claim. 
\par
Let $G$ be the incomparability graph of a linear-semiorder. 
We can compute the vertex ordering of $G$ 
by the following algorithm. 
First, take a complement $\overline{G}$ of $G$. 
Then, obtain a transitive orientation $P$ of $\overline{G}$. 
Finally, compute a linear extension $L$ of $P$ 
fulfilling the $\mathbf{2+2}$ and $\mathbf{3+1}$ rules. 
As mentioned above, 
$L$ can be viewed as a vertex ordering of $G$ with no suborderings 
in Figs.~\ref{figs:forbidden-patterns}\subref{figs:cpc}--\subref{figs:k2}. 
Thus the algorithm is correct. 
The first step can be done in $O(n^2)$ time. 
The second and third steps 
can be done in $O(n m)$ time and $O(m^2 \bar{m})$ time, respectively, 
as shown in Proposition~\ref{prop:obtain transitive orientation} 
and Theorem~\ref{theorem:recognition:main}. 
Thus the algorithm takes $O(m^2 \bar{m})$ time. 
\end{proof}

\subsection{Linear-semiorders and interval orders}
The class of linear-interval orders contains 
interval orders and orders of dimension 2 
as proper subclasses~\cite{CK87-CN}. 
The following example shows that 
the class of interval orders is not a subclass of linear-semiorders. 
\begin{Example}\label{example}
The interval order $P_I$ in Fig.~\ref{figs:P_I}\subref{figs:P_I:order} 
is not a linear-semiorder. 
\end{Example}
\begin{proof}
Figure~\ref{figs:P_I}\subref{figs:P_I:representation} shows 
an interval representation of $P_I$. 
Suppose that $P_I$ is a linear-semiorder. 
Then Theorem~\ref{theorem:characterization} implies that 
$P_I$ has a linear extension $L$ 
fulfilling the $\mathbf{3+1}$ rule. 
\par
The order $\mathbf{3+1}$ consisting of $b_1, b_2, b_3, a_1$ requires 
that $a_1 \prec b_1$ in $L$ if and only if $a_1 \prec b_3$ in $L$. 
The order $\mathbf{3+1}$ consisting of $a_1, a_2, a_3, b_3$ requires 
that $a_1 \prec b_3$ in $L$ if and only if $a_3 \prec b_3$ in $L$. 
The order $\mathbf{3+1}$ consisting of $b_3, b_4, b_5, a_3$ requires 
that $a_3 \prec b_3$ in $L$ if and only if $a_3 \prec b_5$ in $L$. 
Thus $a_1 \prec b_1$ in $L$ if and only if $a_3 \prec b_5$ in $L$. 
\par
The order $\mathbf{3+1}$ consisting of $b_1, b_2, c_2, a_1$ requires 
that $a_1 \prec b_1$ in $L$ if and only if $a_1 \prec c_2$ in $L$. 
The order $\mathbf{3+1}$ consisting of $a_1, b_4, b_5, c_2$ requires 
that $a_1 \prec c_2$ in $L$ if and only if $b_5 \prec c_2$ in $L$. 
The order $\mathbf{3+1}$ consisting of $c_1, b_4, b_5, c_2$ requires 
that $b_5 \prec c_2$ in $L$ if and only if $c_1 \prec c_2$ in $L$. 
Thus $a_1 \prec b_1$ in $L$ if and only if $c_1 \prec c_2$ in $L$. 
\par
Similarly, 
the order $\mathbf{3+1}$ consisting of $c_1, b_4, b_5, a_3$ requires 
that $a_3 \prec b_5$ in $L$ if and only if $a_3 \prec c_1$ in $L$. 
The order $\mathbf{3+1}$ consisting of $b_1, b_2, a_3, c_1$ requires 
that $a_3 \prec c_1$ in $L$ if and only if $b_1 \prec c_1$ in $L$. 
The order $\mathbf{3+1}$ consisting of $b_1, b_2, c_2, c_1$ requires 
that $b_1 \prec c_1$ in $L$ if and only if $c_2 \prec c_1$ in $L$. 
Thus $a_3 \prec b_5$ in $L$ if and only if $c_2 \prec c_1$ in $L$, 
a contradiction. 
Therefore, $P_I$ has no linear extensions 
fulfilling the $\mathbf{3+1}$ rule, that is, 
$P_I$ is not a linear-semiorder. 
\end{proof}

\begin{figure}[t]
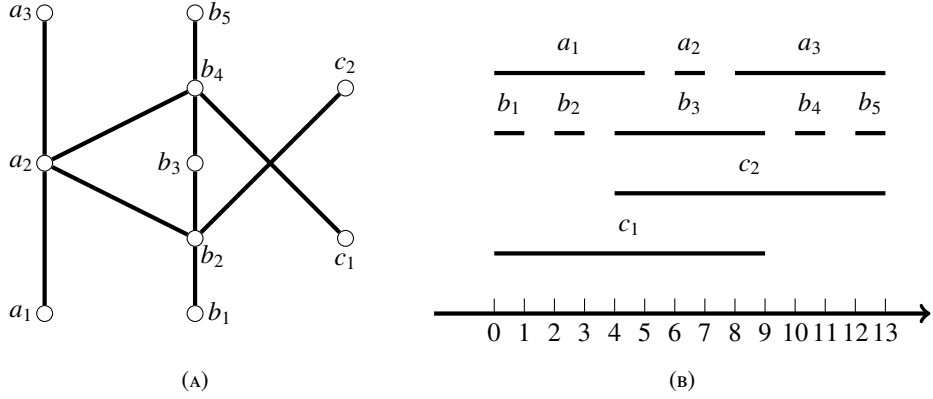

  \centering
  \subcaptionbox{\label{figs:P_I:order}}{\input{figs/P_I/order.tex}}
  \subcaptionbox{\label{figs:P_I:representation}}{\input{figs/P_I/representation.tex}}
  \caption{
    (a) An interval order $P_I$ that is a minimal forbidden order of linear-semiorders. 
    (b) An interval representation of $P_I$. 
    } 
  \label{figs:P_I}
\end{figure}

Note that one can check by inspection that 
every proper induced suborder of $P_I$ has a linear extension 
fulfilling the $\mathbf{3+1}$ rule. 
Since $P_I$ is an interval order, $P_I$ and every its induced suborder 
do not contain $\mathbf{2+2}$ as an induced suborder~\cite{Fishburn70-JMP}. 
Hence every proper induced suborder of $P_I$ is a linear-semiorder, that is, 
$P_I$ is a minimal forbidden order of linear-semiorders. 

\begin{Corollary}
The class of linear-semiorders is not a subclass of interval orders 
and vice versa; 
therefore, the class of linear-semiorders is a proper subclass 
of linear-interval orders. 
The class of linear-semiorders contains 
semiorders and orders of dimension 2 as proper subclasses. 
\end{Corollary}
\begin{proof}
Example~\ref{example} shows an interval order (hence a linear-interval order) 
that is not a linear-semiorder. 
The order $\mathbf{2+2}$ is a linear-semiorder but is not an interval order. 
\par
From definitions, 
every semiorder and order of dimension 2 is a linear-semiorder. 
Thus the classes of semiorders and orders of dimension 2 are 
subclasses of linear-semiorders. 
Since the linear-semiorder in Fig.~\ref{figs:chevron} 
is neither a semiorder (since it contains $\mathbf{2+2}$ as an induced suborder) 
nor an order of dimension 2~\cite{Trotter92-book}, the inclusion is proper. 
\end{proof}

\section{Concluding remarks}\label{sec:conclusion}
In this paper, we introduce linear-semiorders and 
show a characterization and a polynomial-time recognition algorithm 
for this new order class. 
We also prove that being a linear-semiorder is a comparability invariant, 
showing that incomparability graphs of linear-semiorders 
can also be recognized in polynomial time. 
As byproducts, we show a vertex ordering characterization of 
incomparability graphs of linear-semiorders and 
inclusion relationships between the order classes. 
\par
We finally list some open problems on linear-semiorders. 
\begin{enumerate}[\bfseries --]
\item 
Our algorithm takes $O(n \bar{m})$ time to recognize linear-semiorders and 
takes $O(m \bar{m}^2)$ time to produce 
a linear extension fulfilling the $\mathbf{2+2}$ and $\mathbf{3+1}$ rules, 
where $n$, $m$, and $\bar{m}$ are the number of 
elements, comparable pairs, and incomparable pairs of the given order, respectively. 
Is there any faster or simpler algorithm for recognizing linear-semiorders? 
\item 
The class of linear-semiorders is closed under taking induced suborders. 
Hence the orders can be characterized by forbidden induced suborders. 
Can we obtain the list of forbidden induced suborders for linear-semiorders? 
\item 
The isomorphism problem can be solved in linear time 
for interval graphs~\cite{LB79-JACM} and 
permutation graphs~\cite{Colbourn81-Networks}, 
whereas the complexity of the problem remains open 
for incomparability graphs of 
linear-interval orders~\cite{Takaoka15-IEICE,Takaoka20-DAM,Uehara14-DMTCS}. 
For incomparability graphs of linear-semiorders, 
can we determine the complexity of the isomorphism problem? 
\end{enumerate}

%% file: figs/algo1.tex
\IncMargin{1em}
\begin{algorithm}[t]
\caption{Algorithm for computing $E_o$}\label{algo:E_o}
\LinesNumbered%
\KwIn{A partial order $P$ and the incomparability graph $\overline{G}$ of $P$. }
\KwOut{The set $E_o$ of edges of $\overline{G}$. }
\BlankLine

$E_o \gets \emptyset$\;
\ForEach{$uv \in E(\overline{G})$}{
  \If{$U(u) \cap N_{\overline{G}}(v) \neq \emptyset$ and $N_{\overline{G}}(u) \cap U(v) \neq \emptyset$}{
    $E_o \gets E_o + uv$\;
    \lForEach{$w \in U(u) \cap N_{\overline{G}}(v)$}{$E_o \gets E_o + vw$}
    \lForEach{$z \in N_{\overline{G}}(u) \cap U(v)$}{$E_o \gets E_o + uz$}
  }
  \If{$L(u) \cap N_{\overline{G}}(v) \neq \emptyset$ and $N_{\overline{G}}(u) \cap L(v) \neq \emptyset$}{
    $E_o \gets E_o + uv$\;
    \lForEach{$w \in L(u) \cap N_{\overline{G}}(v)$}{$E_o \gets E_o + vw$}
    \lForEach{$z \in N_{\overline{G}}(u) \cap L(v)$}{$E_o \gets E_o + uz$}
  }
  \If{$U(u) \cap N_{\overline{G}}(v) \neq \emptyset$ and $L(u) \cap N_{\overline{G}}(v) \neq \emptyset$}{
    $E_o \gets E_o + uv$\;
    \lForEach{$w \in U(u) \cap N_{\overline{G}}(v)$}{$E_o \gets E_o + vw$}
    \lForEach{$z \in L(u) \cap N_{\overline{G}}(v)$}{$E_o \gets E_o + vz$}
  }
  \If{$N_{\overline{G}}(u) \cap U(v) \neq \emptyset$ and $N_{\overline{G}}(u) \cap L(v) \neq \emptyset$}{
    $E_o \gets E_o + uv$\;
    \lForEach{$w \in N_{\overline{G}}(u) \cap U(v)$}{$E_o \gets E_o + uw$}
    \lForEach{$z \in N_{\overline{G}}(u) \cap L(v)$}{$E_o \gets E_o + uz$}
  }
}
\KwRet{$E_o$}\;
\end{algorithm}
\DecMargin{1em}

%% file: figs/algo2.tex
\IncMargin{1em}
\begin{algorithm}[t]
\caption{Algorithm for Lemma~\ref{lemma:recognition:2}}\label{algo:lemma2}
\LinesNumbered
\KwIn{A partial order $P$ and the set $E_o$ of edges of $\overline{G}$. }
\KwOut{An orientation $F$ of $E_o$ such that $P + F$ contains no forbidden configurations. } 
\BlankLine

$\varphi \gets \emptyset$\;
\ForEach(\tcc*[f]{construct the 2-CNF formula $\varphi$}){$uv \in E_o$}{
  \If{$U(u) \cap N_{\overline{G}}(v) \neq \emptyset$ and $N_{\overline{G}}(u) \cap U(v) \neq \emptyset$}{
    \lForEach{$w \in U(u) \cap N_{\overline{G}}(v)$}{add the clauses $(x_{(u, v)} \vee x_{(v, w)})$ and $(x_{(w, v)} \vee x_{(v, u)})$ to $\varphi$}
    \lForEach{$z \in N_{\overline{G}}(u) \cap U(v)$}{add the clauses $(x_{(v, u)} \vee x_{(u, z)})$ and $(x_{(z, u)} \vee x_{(u, v)})$ to $\varphi$}
  }
  \If{$L(u) \cap N_{\overline{G}}(v) \neq \emptyset$ and $N_{\overline{G}}(u) \cap L(v) \neq \emptyset$}{
    \lForEach{$w \in L(u) \cap N_{\overline{G}}(v)$}{add the clauses $(x_{(u, v)} \vee x_{(v, w)})$ and $(x_{(w, v)} \vee x_{(v, u)})$ to $\varphi$}
    \lForEach{$z \in N_{\overline{G}}(u) \cap L(v)$}{add the clauses $(x_{(v, u)} \vee x_{(u, z)})$ and $(x_{(z, u)} \vee x_{(u, v)})$ to $\varphi$}
  }
  \If{$U(u) \cap N_{\overline{G}}(v) \neq \emptyset$ and $L(u) \cap N_{\overline{G}}(v) \neq \emptyset$}{
    \lForEach{$w \in U(u) \cap N_{\overline{G}}(v)$}{add the clauses $(x_{(u, v)} \vee x_{(v, w)})$ and $(x_{(w, v)} \vee x_{(v, u)})$ to $\varphi$}
    \lForEach{$z \in L(u) \cap N_{\overline{G}}(v)$}{add the clauses $(x_{(u, v)} \vee x_{(v, z)})$ and $(x_{(z, v)} \vee x_{(v, u)})$ to $\varphi$}
  }
  \If{$N_{\overline{G}}(u) \cap U(v) \neq \emptyset$ and $N_{\overline{G}}(u) \cap L(v) \neq \emptyset$}{
    \lForEach{$w \in N_{\overline{G}}(u) \cap U(v)$}{add the clauses $(x_{(v, u)} \vee x_{(u, w)})$ and $(x_{(w, u)} \vee x_{(u, v)})$ to $\varphi$}
    \lForEach{$z \in N_{\overline{G}}(u) \cap L(v)$}{add the clauses $(x_{(v, u)} \vee x_{(u, z)})$ and $(x_{(z, u)} \vee x_{(u, v)})$ to $\varphi$}
  }
  \lForEach{$w \in U(u) \cap N_{\overline{G}}(v)$ with $vw \in E_o$}{
      add the clause $(x_{(w, v)} \vee x_{(v, u)})$ to $\varphi$
  }
  \lForEach{$z \in N_{\overline{G}}(u) \cap U(v)$ with $uz \in E_o$}{
      add the clause $(x_{(z, u)} \vee x_{(u, v)})$ to $\varphi$
  }
}
find a satisfying truth assignment $\tau$ of $\varphi$\; 
orient each edge $uv \in E_o$ as $(u, v) \in F$ if $x_{(u, v)} = 0$ in $\tau$\; 
\KwRet{$F$}\; 
\end{algorithm}
\DecMargin{1em}

%% file: figs/algo3.tex
\IncMargin{1em}
\begin{algorithm}[t]
\caption{Algorithm for Lemma~\ref{lemma:recognition:4}}\label{algo:lemma4}
\LinesNumbered
\KwIn{An orientation $F$ of $E_o$ such that $P + F$ contains no forbidden configurations. }
\KwOut{An orientation $F'$ of $E_o$ such that $P + F'$ contains no forbidden configurations nor no obstructions. }
\BlankLine

\ForEach(\tcc*[f]{remove regular obstructions}){$(u, v) \in P$}{
  $A_{(u, v)} \gets \emptyset$\;
  \ForEach{pair $(v_4, v_3), (v_5, v_2) \in F$}{
    \If{$(v_2, v_3), (v_4, v_5) \in P$ and $(v_2, v), (u, v_5) \in F$}{
      $A_{(u, v)} \gets A_{(u, v)} + \{(v_4, v_3), (v_5, v_3), (v_4, v_2), (v_5, v_2)\}$\;
    }
  }
  $F \gets (F - A_{(u, v)}) + A_{(u, v)}^{-1}$\;
}
\ForEach(\tcc*[f]{remove skewed obstructions}){$(u, v) \in P$}{
  $B_{(u, v)} \gets \emptyset$\;
  \ForEach{pair $(v_4, v_3), (v_5, v_2) \in F$}{
    \If{$(v, v_2), (v_4, v_5) \in P$ and $(v_3, v_2), (v_5, v_3), (u, v_5) \in F$}{
      $B_{(u, v)} \gets B_{(u, v)} + \{(v_4, v_3), (v_5, v_3)\}$\;
    }
  }
  $F \gets (F - B_{(u, v)}) + B_{(u, v)}^{-1}$\;
}
\KwRet{$F$}\;
\end{algorithm}
\DecMargin{1em}

%% file: draft.bbl
\begin{thebibliography}{10}

\bibitem{APT79-IPL}
B.~Aspvall, M.~F. Plass, and R.~E. Tarjan.
\newblock \href {http://dx.doi.org/10.1016/0020-0190(79)90002-4} {A linear-time
  algorithm for testing the truth of certain quantified boolean formulas}.
\newblock {\em Inf. Process. Lett.}, 8(3):121--123, 1979.

\bibitem{BLR10-Order}
K.~P. Bogart, J.~D. Laison, and S.~P. Ryan.
\newblock \href {http://dx.doi.org/10.1007/s11083-010-9144-6} {Triangle,
  parallelogram, and trapezoid orders}.
\newblock {\em Order}, 27(2):163--175, 2010.

\bibitem{BMR98-Order}
K.~P. Bogart, R.~H. M{\"o}hring, and S.~P. Ryan.
\newblock \href {http://dx.doi.org/10.1023/A:1006042122315} {Proper and unit
  trapezoid orders and graphs}.
\newblock {\em Order}, 15(4):325--340, 1998.

\bibitem{BW99-DM}
K.~P. Bogart and D.~B. West.
\newblock \href {http://dx.doi.org/10.1016/S0012-365X(98)00310-0} {A short
  proof that 'proper = unit'}.
\newblock {\em Discrete Math.}, 201(1-3):21--23, 1999.

\bibitem{BLS99}
A.~Brandst\"{a}dt, V.~B. Le, and J.~P. Spinrad.
\newblock \href {http://dx.doi.org/10.1137/1.9780898719796} {{\em Graph
  Classes: A Survey}}.
\newblock {SIAM}, Philadelphia, PA, USA, 1999.

\bibitem{COS08-ENDM}
M.~R. Cerioli, F.~de~S.~Oliveira, and J.~L. Szwarcfiter.
\newblock \href {http://dx.doi.org/10.1016/j.endm.2008.01.021} {Linear-interval
  dimension and {PI} orders}.
\newblock {\em Electron. Notes Discrete Math.}, 30:111--116, 2008.

\bibitem{Colbourn81-Networks}
C.~J. Colbourn.
\newblock \href {http://dx.doi.org/10.1002/net.3230110103} {On testing
  isomorphism of permutation graphs}.
\newblock {\em Networks}, 11(1):13--21, 1981.

\bibitem{CK87-CN}
D.~G. Corneil and P.~A. Kamula.
\newblock Extensions of permutation and interval graphs.
\newblock {\em Congr. Numer.}, 58:267--275, 1987.

\bibitem{DGP88-DAM}
I.~Dagan, M.~C. Golumbic, and R.~Y. Pinter.
\newblock \href {http://dx.doi.org/10.1016/0166-218X(88)90032-7} {Trapezoid
  graphs and their coloring}.
\newblock {\em Discrete Appl. Math.}, 21(1):35--46, 1988.

\bibitem{FM98-Order}
S.~Felsner and R.~H. M{\"{o}}hring.
\newblock \href {http://dx.doi.org/10.1023/A:1006223715426} {Note: Semi-order
  dimension two is a comparability invariant}.
\newblock {\em Order}, 15(4):385--390, 1998.

\bibitem{Fishburn70-JMP}
P.~C. Fishburn.
\newblock \href {http://dx.doi.org/10.1016/0022-2496(70)90062-3} {Intransitive
  indifference with unequal indifference intervals}.
\newblock {\em J. Math. Psych.}, 7(1):144--149, 1970.

\bibitem{Gallai67}
T.~Gallai.
\newblock \href {http://dx.doi.org/10.1007/BF02020961} {Transitiv orientierbare
  graphen}.
\newblock {\em Acta Math. Acad. Sci. Hungar.}, 18(1):25--66, 1967.

\bibitem{Golumbic04}
M.~C. Golumbic.
\newblock {\em Algorithmic Graph Theory and Perfect Graphs}, volume~57 of {\em
  Ann. Discrete Math.}
\newblock Elsevier, 2 edition, 2004.

\bibitem{GRU83-DM}
M.~C. Golumbic, D.~Rotem, and J.~Urrutia.
\newblock \href {http://dx.doi.org/10.1016/0012-365X(83)90019-5} {Comparability
  graphs and intersection graphs}.
\newblock {\em Discrete Math.}, 43(1):37--46, 1983.

\bibitem{GT04}
M.~C. Golumbic and A.~N. Trenk.
\newblock \href {http://dx.doi.org/10.1017/CBO9780511542985} {{\em Tolerance
  Graphs}}.
\newblock Cambridge Univ. Press, 2004.

\bibitem{KS93-SIAMDM}
D.~Kratsch and L.~Stewart.
\newblock \href {http://dx.doi.org/10.1137/0406032} {Domination on
  cocomparability graphs}.
\newblock {\em {SIAM} J. Discrete Math.}, 6(3):400--417, 1993.

\bibitem{LB79-JACM}
G.~S. Lueker and K.~S. Booth.
\newblock \href {http://dx.doi.org/10.1145/322123.322125} {A linear time
  algorithm for deciding interval graph isomorphism}.
\newblock {\em J. {ACM}}, 26(2):183--195, 1979.

\bibitem{MS99-DM}
R.~M. McConnell and J.~P. Spinrad.
\newblock \href {http://dx.doi.org/10.1016/S0012-365X(98)00319-7} {Modular
  decomposition and transitive orientation}.
\newblock {\em Discrete Math.}, 201(1--3):189--241, 1999.

\bibitem{Mertzios15-SIAMDM}
G.~B. Mertzios.
\newblock \href {http://dx.doi.org/10.1137/140963108} {The recognition of
  simple-triangle graphs and of linear-interval orders is polynomial}.
\newblock {\em {SIAM} J. Discrete Math.}, 29(3):1150--1185, 2015.

\bibitem{MC11-DAM}
G.~B. Mertzios and D.~G. Corneil.
\newblock \href {http://dx.doi.org/10.1016/j.dam.2011.03.023} {Vertex splitting
  and the recognition of trapezoid graphs}.
\newblock {\em Discrete Appl. Math.}, 159(11):1131--1147, 2011.

\bibitem{Ryan98-Order}
S.~P. Ryan.
\newblock \href {http://dx.doi.org/10.1023/A:1006223420935} {Trapezoid order
  classification}.
\newblock {\em Order}, 15(4):341--354, 1998.

\bibitem{SS58-JSL}
D.~Scott and P.~Suppes.
\newblock \href {http://dx.doi.org/10.2307/2964389} {Foundational aspects of
  theories of measurement}.
\newblock {\em J. Symb. Log.}, 23(2):113--128, 1958.

\bibitem{Spinrad03}
J.~P. Spinrad.
\newblock {\em Efficient Graph Representations}, volume~19 of {\em Fields
  Institute Monographs}.
\newblock {AMS}, Providence, RI, USA, 2003.

\bibitem{Takaoka15-IEICE}
A.~Takaoka.
\newblock \href {http://dx.doi.org/10.1587/transfun.E98.A.1838} {{Graph
  Isomorphism Completeness for Trapezoid Graphs}}.
\newblock {\em IEICE Trans. Fundamentals}, E98--A(8):1838--1840, 2015.

\bibitem{Takaoka18-DM}
A.~Takaoka.
\newblock \href {http://dx.doi.org/10.1016/j.disc.2018.08.009} {A vertex
  ordering characterization of simple-triangle graphs}.
\newblock {\em Discrete Math.}, 341(12):3281--3287, 2018.

\bibitem{Takaoka20a-DAM}
A.~Takaoka.
\newblock \href {http://dx.doi.org/10.1016/j.dam.2019.11.009} {A recognition
  algorithm for simple-triangle graphs}.
\newblock {\em Discret. Appl. Math.}, 282:196--207, 2020.

\bibitem{Takaoka20-DAM}
A.~Takaoka.
\newblock \href {http://dx.doi.org/10.1016/j.dam.2019.10.028} {Recognizing
  simple-triangle graphs by restricted 2-chain subgraph cover}.
\newblock {\em Discret. Appl. Math.}, 279:154--167, 2020.

\bibitem{Trotter92-book}
W.~T. Trotter.
\newblock Combinatorics and partially ordered sets: Dimension theory.
\newblock 1992.

\bibitem{Uehara14-DMTCS}
R.~Uehara.
\newblock \href
  {http://www.dmtcs.org/dmtcs-ojs/index.php/dmtcs/article/view/2528} {The graph
  isomorphism problem on geometric graphs}.
\newblock {\em Discrete Math. Theor. Comput. Sci.}, 16(2):87--96, 2014.

\end{thebibliography}
